\pdfoutput=1
\pdfoutput=1
\newif\ifpaperTwoITSiamStyle
\ifdefined\SICOMPSUBMISSION
  \paperTwoITSiamStyletrue
  \documentclass[review]{siamart}
\else
  \ifdefined\SIDMASUBMISSION
    \paperTwoITSiamStyletrue
    \documentclass[review]{siamart}
  \else
    \paperTwoITSiamStylefalse
    \documentclass[11pt]{article}
  \fi
\fi

\ifpaperTwoITSiamStyle
  \usepackage[utf8]{inputenc}
  \usepackage[T1]{fontenc}
  \usepackage{amsmath,amssymb}
  \usepackage{graphicx}
  \usepackage{xcolor}
  \usepackage{fancyvrb}
  \usepackage{url}
  \usepackage{listings}
  \usepackage{algorithm}
  \usepackage{algpseudocode}

  \makeatletter
  \g@addto@macro{\UrlBreaks}{\UrlOrds}
  \makeatother

\usepackage{booktabs}
\usepackage{longtable}
\usepackage{array}
\usepackage{calc}
\usepackage[para,flushmargin]{footmisc}
\usepackage{listings}
\usepackage{ragged2e}
\usepackage{tabularx}
\usepackage{tikz}
\usetikzlibrary{positioning, arrows.meta, shapes}

\IfFileExists{content/lean_stats.tex}{

}{}

\newcolumntype{Y}{>{\RaggedRight\arraybackslash}X}
\newcolumntype{C}{>{\Centering\arraybackslash}X}

\providecommand{\leanmetaheadmark}{}
\makeatletter
\@ifundefined{axiom}{\newsiamthm{axiom}{Axiom}}{}
\@ifundefined{conjecture}{\newsiamthm{conjecture}{Conjecture}}{}
\@ifundefined{example}{\newsiamremark{example}{Example}}{}
\@ifundefined{remark}{\newsiamremark{remark}{Remark}}{}
\@ifundefined{observation}{\newsiamremark{observation}{Observation}}{}
\makeatother

\providecommand{\LH}[1]{}
\renewcommand{\LH}[1]{\leavevmode\nobreak\hyperlink{lh:#1}{\mbox{\ttfamily\nolinkurl{#1}}}}
\providecommand{\LHrng}[3]{}
\renewcommand{\LHrng}[3]{\leavevmode\nobreak\hyperlink{lh:#1#2}{\mbox{\ttfamily\nolinkurl{#1#2-#3}}}}
\newif\ifleanmetaseen
\leanmetaseenfalse
\makeatletter
\providecommand{\leanmetapendinghandles}{}
\global\let\leanmetapendinghandles\@empty
\providecommand{\leanmetapending}[1]{}
\renewcommand{\leanmetapending}[1]{\def\leanmetatest{#1}\ifx\leanmetatest\@empty\global\let\leanmetapendinghandles\@empty\else\gdef\leanmetapendinghandles{#1}\fi}
\providecommand{\leanmetaproofpendinghandles}{}
\global\let\leanmetaproofpendinghandles\@empty
\providecommand{\leanmetaproofpending}[1]{}
\renewcommand{\leanmetaproofpending}[1]{\def\leanmetatest{#1}\ifx\leanmetatest\@empty\global\let\leanmetaproofpendinghandles\@empty\else\gdef\leanmetaproofpendinghandles{#1}\fi}
\providecommand{\leanmetaheadmark}{}
\newif\ifleanmetaheadprinted
\leanmetaheadprintedfalse
\renewcommand{\leanmetaheadmark}{\ifx\leanmetapendinghandles\@empty\else\nobreak\hspace{0.2em}\refstepcounter{footnote}\footnotemark[\value{footnote}]\global\leanmetaheadprintedtrue\fi}
\providecommand{\leanmetaheadtext}{}
\renewcommand{\leanmetaheadtext}{\ifleanmetaheadprinted\footnotetext[\value{footnote}]{\ifleanmetaseen\else Lean: \global\leanmetaseentrue\fi\leanmetapendinghandles}\global\let\leanmetapendinghandles\@empty\global\leanmetaheadprintedfalse\fi}
\providecommand{\leanmetaproofmark}{}
\renewcommand{\leanmetaproofmark}{\ifx\leanmetaproofpendinghandles\@empty\else\ifhmode\unskip\footnote{\ifleanmetaseen\else Lean: \global\leanmetaseentrue\fi\leanmetaproofpendinghandles}\else\refstepcounter{footnote}\footnotemark[\value{footnote}]\footnotetext[\value{footnote}]{\ifleanmetaseen\else Lean: \global\leanmetaseentrue\fi\leanmetaproofpendinghandles}\fi\global\let\leanmetaproofpendinghandles\@empty\fi}
\@ifundefined{newtheoremstyle}{}{\@ifundefined{thmname}{}{%
  \newtheoremstyle{leanplain}
    {6pt}{6pt}{\itshape}{}{\bfseries}{.}{ }
    {\thmname{#1}\thmnumber{ #2}\thmnote{ (#3)}\leanmetaheadmark\leanmetaheadtext}
  \newtheoremstyle{leandefinition}
    {6pt}{6pt}{}{}{\bfseries}{.}{ }
    {\thmname{#1}\thmnumber{ #2}\thmnote{ (#3)}\leanmetaheadmark\leanmetaheadtext}
  \newtheoremstyle{leandef}
    {6pt}{6pt}{}{}{\bfseries}{.}{ }
    {\thmname{#1}\thmnumber{ #2}\thmnote{ (#3)}\leanmetaheadmark\leanmetaheadtext}
  \newtheoremstyle{leanremark}
    {6pt}{6pt}{}{}{\bfseries}{.}{ }
    {\thmname{#1}\thmnumber{ #2}\thmnote{ (#3)}\leanmetaheadmark\leanmetaheadtext}
}}
\@ifundefined{thm@thmcaption}{}{\@ifundefined{@opargbegintheorem}{}{%
  \def\@xthm#1#2#3{%
    \@begintheorem{#3}{\csname the#2\endcsname\leanmetaheadmark}%
    \leanmetaheadtext
    \ifx\thm@starredenv\@undefined
      \thm@thmcaption{#1}{{#3}{\csname the#2\endcsname}{}}\fi
    \ignorespaces}
  \def\@ythm#1#2#3[#4]{%
    \expandafter\global\expandafter\def\csname#1name\endcsname{#4}%
    \@opargbegintheorem{#3}{\csname the#2\endcsname}{#4\leanmetaheadmark}%
    \leanmetaheadtext
    \ifx\thm@starredenv\@undefined
      \thm@thmcaption{#1}{{#3}{\csname the#2\endcsname}{#4}}\fi
    \ignorespaces}
}}
\makeatother
\providecommand{\leanmeta}[1]{}
\renewcommand{\leanmeta}[1]{\ifhmode\unskip\footnote{\ifleanmetaseen\else Lean: \global\leanmetaseentrue\fi#1}\else\refstepcounter{footnote}\footnotemark[\value{footnote}]\footnotetext[\value{footnote}]{\ifleanmetaseen\else Lean: \global\leanmetaseentrue\fi#1}\fi}

\makeatletter
\@ifundefined{refstepcounter@noarg}{}{
  \def\refstepcounter@optarg[#1]#2{\refstepcounter@noarg{#2}}
}
\makeatother

\AtBeginDocument{%
}

\lstnewenvironment{code}{
  \lstset{
    basicstyle=\ttfamily\footnotesize,
    breaklines=true,
    breakatwhitespace=true,
    columns=flexible,
    keepspaces=true,
    xleftmargin=0.5em,
    frame=none,
    linewidth=\linewidth
  }
}{}

  \makeatletter
  \@ifundefined{keywords}{\newenvironment{keywords}{\paragraph{Keywords.}}{}}{}
  \@ifundefined{MSCcodes}{\newenvironment{MSCcodes}{\paragraph{MSC codes}}{}}{}
  \makeatother

  \makeatletter
  \renewcommand{\LH}[1]{\leavevmode\nobreak\mbox{\ttfamily\nolinkurl{#1}}}
  \renewcommand{\LHrng}[3]{\leavevmode\nobreak\mbox{\ttfamily\nolinkurl{#1#2-#3}}}
  \makeatother

  \hypersetup{
    colorlinks=true,
    linkcolor=blue,
    citecolor=blue,
    urlcolor=blue
  }
  \IfFileExists{build_submission_hook.tex}{\input{build_submission_hook.tex}}{}
  \renewcommand{\footercopyright}{}
\else
  \usepackage{paper-preamble}
\fi

\begin{document}

\title{Coordinate-View Confusability Graphs and Matroid Rank Certificates}
\ifpaperTwoITSiamStyle
  \headers{Coordinate-View Confusability Graphs}{T. Simas}
\fi

\ifpaperTwoITSiamStyle
  \author{Tristan Simas\thanks{McGill University, 845 Sherbrooke Street West, Montreal, QC H3A 0G4, Canada (\email{tristan.simas@mail.mcgill.ca}).}}
\else
  \author{Tristan Simas\\
  McGill University, Montreal, Quebec, Canada\\
  \texttt{tristan.simas@mail.mcgill.ca}}
\fi
\date{\today}

\maketitle

\begin{abstract}
A coordinate-view presentation specifies a large confusability graph by coordinates rather than by an edge list. The problem is to certify zero-error recovery and Shannon capacity from the succinct presentation, before the exponentially large graph is materialized.

For affine presentations over \(\mathbb F_q\), Gaussian elimination gives a polynomial-time rank upper certificate from \(O(rd\log q+Ld)\) input bits for a graph with \(q^r\) vertices. Exactness of that certificate is equivalent to positivity of a Grassmannian avoidance count \(N_{t^*}\). Positivity of \(N_k\) is NP-complete for every fixed \(k\), already over \(\mathbb F_2\), while the ambient-rank parameter gives a fixed-parameter algorithm. The rank-one reduction is parsimonious for \(\#\textsc{SAT}\).

For arbitrary \(t\), the finite decoder for \(N_t\) reads the signed rank profile of the represented forbidden-point matroid; the kernel-intersection lattice alone does not determine the count. On the positive side, kernel sections give the exact formula
\[
\Theta(G)=\log\alpha(G)=\min_{S\in\mathcal V}t(S)\log q,
\]
with a projection-equality matrix optimal for Haemers minrank. A finite-field blocking theorem gives exactness when \(q\ge L\), and Reed--Solomon/MDS codes give exact all-\(k\)-view families beyond that regime. Full-tuple coordinate-view graphs also have a polynomial-time cofinal-antichain normal form; transitive confusability is exactly intersection closure.

\end{abstract}

\ifpaperTwoITSiamStyle
  \begin{keywords}
  confusability graphs, coordinate projections, Shannon capacity, strong graph powers, matroids, graph coloring, zero-error information theory
  \end{keywords}

  \begin{MSCcodes}
  05C15, 05C69, 05B35, 94A15, 68Q87
  \end{MSCcodes}
\else
  \noindent\textbf{Keywords:}
  confusability graphs; coordinate projections; Shannon capacity; strong graph powers; matroids; graph coloring; zero-error information theory

  \medskip
  \noindent\textbf{MSC 2020:}
  05C15; 05C69; 05B35; 94A15; 68Q87
\fi

\section{Introduction}\label{introduction}

A succinct coordinate presentation can hide an exponentially large graph. A message has several coordinates, and a receiver may see only the coordinates in one allowed view. Two messages are confusable when at least one allowed view makes them look identical. A short auxiliary label can separate confusable messages; over repeated independent uses, the limiting zero-error rate is the Shannon capacity of the induced confusability graph.

Exact one-shot recovery is graph coloring, block recovery is strong powering, and the limiting block rate is Shannon capacity~\cite{witsenhausen1976zero,korner2002zero,shannon1956zero}. For arbitrary materialized graphs, exact independent-set computation is NP-hard and exact Shannon capacity remains open even for small explicit graphs, including the $7$-cycle and odd cycles beyond $C_5$~\cite{karp1972reducibility,alon2006shannon,deboer2024asymptotic}. Coordinate-view presentations keep algebraic structure before the graph is enumerated.

The affine case turns exactness into a subspace-arrangement problem. Let
\[
A=a_0+V\subseteq\mathbb F_q^d,\qquad \dim V=r,
\]
and let \(\mathcal V\) be a finite list of coordinate views. For a view \(S\), write \(t(S)\) for the rank of the restricted coordinate projection \(V\to\mathbb F_q^S\), and put \(t^*=\min_{S\in\mathcal V}t(S)\). The view kernels
\[
K_S=\ker(\pi_S|_V)
\]
form the forbidden arrangement. A \(t^*\)-dimensional subspace \(R\le V\) avoiding every \(K_S\) turns the rank upper bound into an independent set of size \(q^{t^*}\). Direct sums tensorize the section and remove the Shannon gap.

The boundary count is
\[
N_t(\mathcal A)=\#\{R\le V:\dim R=t,\ R\cap K_S=0\text{ for every }S\in\mathcal V\}.
\]
Four statements drive the boundary. First, \(N_{t^*}>0\) is exactly the linear kernel-section question. Second, positivity of \(N_k\) is NP-complete for every fixed \(k\), and the rank-one reduction is parsimonious for \(\#\textsc{SAT}\). Third, the finite decoder for \(N_t\) reads the signed rank profile of the represented forbidden-point matroid. The kernel-intersection lattice is too coarse; common-core lifts keep the multiway kernel-intersection shell fixed while quotient forbidden-point configurations change \(N_t\). Fourth, the size-indexed fixed-shell theorem transports SAT decision and satisfying-assignment counts through one fixed shell for each Boolean variable count.

Kernel-section positivity gives the capacity formula. On the kernel-section class,
\[
\Theta(G)=\log\alpha(G)=t^*\log q
\quad\text{and}\quad
\operatorname{minrank}_{\mathbb K}(G)=q^{t^*},
\]
with the projection-equality matrix as an optimal Haemers witness~\cite{haemers1978upper}. Logarithmic capacity adds under direct-sum block products inside the class. A greedy finite-field union theorem guarantees the section when \(q\ge L\). Reed--Solomon/MDS codes give exact families beyond that field-size regime: for the full cube with all \(k\)-coordinate views, the section condition is the Singleton boundary.

The central statements give both sides of the same boundary: Gaussian elimination computes the rank certificate before the graph is materialized; \(N_{t^*}>0\) makes the certificate exact; \(q\ge L\) and the Reed--Solomon/MDS all-\(k\)-view families lie on the positive side; positivity of \(N_k\) is NP-complete for every fixed \(k\); and \(N_1\) is parsimoniously \(\#\textsc{P}\)-hard to count.

The binary square supplies the smallest non-clique witness. With state space \(\{0,1\}^2\) and two one-coordinate views, adjacent corners agree in one coordinate and are confusable, while opposite corners are distinguishable. The induced graph is \(C_4\), not a clique. The diagonal line avoids both singleton kernels, so the rank bound is exact. The four-coordinate binary threshold example moves to the other side: \(N_2=0\), and the rank certificate is loose.

Full tuple spaces supply the structural layer beneath the affine boundary. On a full labeled product, a pair is adjacent exactly when its agreement set lies in the upward-closed family generated by the admissible views. Graph equality recovers this generated family, except for the stated edgeless degeneracies, and the cofinal antichain is the canonical presentation. Transitive confusability is exactly intersection closure of the generated family and exactly the meet-witness condition on the finite view list; in that case the graph is a cluster graph and capacity is component counting.

\medskip
\noindent\textbf{Contributions.}
\begin{enumerate}
\item \textbf{Grassmannian avoidance and hardness.} The exact linear-section boundary is positivity of \(N_{t^*}\). Positivity of \(N_k\) is NP-complete for every fixed \(k\), and computing \(N_1\) is \(\#\textsc{P}\)-hard by a parsimonious rank-one reduction. The decision problem is fixed-parameter tractable in the ambient affine rank \(r\) and para-\(\textsc{NP}\)-complete for the target rank/field parameters \(t^*\), \(q\), and \((q,t^*)\). For arbitrary \(t\), the Grassmannian avoidance count is decoded from the signed rank profile of the represented forbidden-point matroid. A common-core lift gives one fixed multiway kernel-intersection profile for all clause lists with the same Boolean variable count; inside that profile, lifted positivity is SAT and the lifted witness count is the satisfying-assignment count.

\item \textbf{Exact affine capacity and minrank.} Restricted coordinate ranks are computed by Gaussian elimination on the explicit affine presentation. In the dense finite-field input model, the presentation has size \(O(rd\log q+Ld)\) bits while the induced graph has \(q^r\) vertices. If a rank-minimizing view has a kernel-avoiding section, then \(\Theta(G)=t^*\log q\), \(\alpha(G)=q^{t^*}\), and \(\operatorname{minrank}_{\mathbb K}(G)=q^{t^*}\) over every coefficient field. Direct-sum block products preserve the class and add logarithmic capacity. The field-size theorem guarantees a kernel-avoiding section when \(q\ge L\), and Reed--Solomon/MDS codes give exact capacity for the full cube with all \(k\)-coordinate views. A finite \(\mathbb F_3^6\) instance has \(\alpha=3^{t^*}\) by a nonlinear transversal while \(N_{t^*}=0\), separating rank-tight transversals from linear sections.

\item \textbf{Coordinate-view graph structure.} Deterministic coordinate-subset views generate the graph class \(\mathcal{CVG}\). On full tuple spaces, graph equality determines the generated upward family \(\mathcal U_{\mathcal V}\), because every proper coordinate subset is realized as an agreement set of an actual state pair. The order-theoretic part is the standard cofinal-antichain representation of an upward-closed family. Full tuple spaces also admit an exact cluster route: over a nontrivial alphabet, intersection closure of the generated agreement family, meet-witnessing of the finite view list, and transitive confusability are equivalent. In the positive case the graph is a cluster graph and capacity is obtained by counting connected components; in the negative case, the construction gives an obstruction triple.
\end{enumerate}

\medskip
\noindent\textbf{Roadmap.} Sections~\ref{sec:foundations}--\ref{sec:graph-characterization} fix the coordinate-product model, the counting preliminary, and the agreement-set graph invariant. Section~\ref{sec:affine} proves the affine rank certificate, kernel-section exactness, signed-profile compression, fixed-rank hardness, fixed-shell transport, and coding-theoretic exact families. Section~\ref{sec:capacity-theory} records the capacity interface, theta comparison, and cluster-collapse criterion. Section~\ref{sec:related} compares coordinate-view presentations with standard zero-error and graph-capacity tools.


\section{Model and Basic Quantities}\label{sec:foundations}

\subsection{Latent states, views, and auxiliary labels}\label{sec:epistemic}

The model starts from a finite latent state family $\mathcal X$. In the coordinate-view specialization, each state is a tuple of definite coordinate values. A partial view is a deterministic observation of selected coordinates, and an auxiliary label is a finite-valued symbol available to the decoder.

The auxiliary label splits the ambiguity left by the view. If a view sends several latent states to the same transcript, then the label must distinguish the states inside that transcript fiber. Section~\ref{sec:converse} makes this counting obstruction explicit before the graph model collects all view fibers into one confusability relation.

\begin{definition}[Latent State Family]\label{def:latent-state-family}
A \emph{latent state family} is a finite set $\mathcal X$ of states. In the coordinate model, $\mathcal X\subseteq A_1\times\cdots\times A_d$ for finite coordinate alphabets $A_i$.
\end{definition}

\begin{definition}[Observation Fiber]\label{def:observation-fiber}
For a deterministic observation map $\mathsf O:\mathcal X\to\mathcal Y$, the \emph{observation fiber} at $y\in\mathcal Y$ is $\mathcal F_y=\{x\in\mathcal X: \mathsf O(x)=y\}$.
\end{definition}

\begin{definition}[Auxiliary Label]\label{def:auxiliary-label}
An \emph{auxiliary label} is a map $\tau:\mathcal X\to\mathcal T$ into a finite label alphabet $\mathcal T$.
\end{definition}

\subsection{Coordinate products}\label{sec:coupling}

The graph and matroid constructions are generated from finite coordinate data: a state family $\mathcal X\subseteq A_1\times\cdots\times A_d$ and a finite view family $\mathcal V\subseteq 2^{[d]}$. The mechanism generates a particular class of finite graphs rather than an arbitrary graph class. The structural theorems below describe the graphs that arise from coordinate views and the extra invariants retained by the presentation. The product-space convention is fixed before those constructions are used.

\begin{definition}[Projection-valued state family]\label{def:projection-valued-state-family}
A state family is \emph{projection-valued} when every state has a single value in each coordinate and every admissible view is the restriction of the tuple to a specified coordinate subset.
\end{definition}

For a projection-valued family $\mathcal X\subseteq A_1\times\cdots\times A_d$, equality of coordinates defines agreement sets, coordinate subsets define views, products of tuples define block composition, and coordinate functionals define the affine matroid specialization. The same finite product data generate the graph, capacity, and matroid objects analyzed below.

\subsection{Standing notation}\label{sec:assumptions}

All state, observation, and label alphabets are finite. Observations and labels are deterministic functions of the latent state. In the coordinate-view model, $\mathcal X$ is a finite subset of $A_1\times\cdots\times A_d$, and every admissible view is a coordinate-subset projection. Logarithms are base $2$ throughout unless a local statement explicitly declares an arbitrary fixed base; bare $\log$ therefore means $\log_2$. The affine rank inequalities of Section~\ref{sec:affine} hold in any fixed base.

\section{Finite Counting Preliminary}\label{sec:converse}

The only information-theoretic preliminary needed before the graph theory is the standard injectivity obstruction for exact deterministic decoding~\cite{witsenhausen1976zero,cover2006elements}. Let $\mathcal T$ be a finite label alphabet, let $\mathsf O:\mathcal X\to\mathcal Y$ be a deterministic observation map, and let $\tau:\mathcal X\to\mathcal T$ be the auxiliary label. Exact recovery is possible exactly when the pair map
\[
x\longmapsto (\mathsf O(x),\tau(x))
\]
is injective, equivalently when $\tau$ is injective on every observation fiber of $\mathsf O$.

\leanmetapending{\LHrng{OBS}{1}{2}}
\begin{theorem}[Pair-injectivity characterization]\label{thm:pair-injective}
Let $\mathcal X$ be a finite latent alphabet, let $\mathsf O:\mathcal X\to\mathcal Y$ be a deterministic observation map, and let $\tau:\mathcal X\to\mathcal T$ be an auxiliary label with finite alphabet $\mathcal T$. Zero-error recovery from $(\mathsf O,\tau)$ is possible if and only if the pair map $x \mapsto (\mathsf O(x),\tau(x))$ is injective on $\mathcal X$, equivalently if $\tau$ is injective on every observation fiber of $\mathsf O$.
\end{theorem}

\begin{proof}
Equal observation-label pairs cannot be separated by any deterministic decoder. Conversely, an injective pair map has a well-defined inverse on its image; use that inverse as the decoder and extend arbitrarily off the image.
\end{proof}

\leanmetapending{\LH{OBS5}, \LH{CIA3}}
On an observation fiber where \(K\) latent states have the same transcript, exact recovery requires \(K\) distinct label outcomes. Thus a \(T\)-valued auxiliary label can recover at most \(T\) states in that ambiguity class; if the label is represented by \(B\) bits, then \(B\ge\log_2K\). The graph-theoretic development below uses only this fiber-counting consequence: a confusability clique of size \(k\) forces at least \(k\) exact labels, equivalently at least \(\log_2 k\) bits.

\section{Coordinate-View Confusability Graphs}\label{sec:graph-characterization}

From this section onward, $\mathcal X\subseteq A_1\times\cdots\times A_d$ denotes the finite coordinate state family unless stated otherwise.

\subsection{General Graph Characterization}

\begin{definition}[Coordinate-Subset View]\label{def:coordinate-view}
For a $d$-coordinate latent tuple $x=(x_1,\dots,x_d)$ and a coordinate set $S\subseteq[d]$, the \emph{coordinate-subset view} indexed by $S$ is the projection
\[
\pi_S(x) := x|_S.
\]
\end{definition}

\begin{definition}[Admissible View Family]\label{def:admissible-view-family}
An \emph{admissible view family} is a finite family $\mathcal V=\{V_1,\dots,V_L\}$ of coordinate subsets of $[d]$, where view $V_\ell$ reveals exactly the coordinates indexed by $V_\ell$.
\end{definition}

\begin{definition}[Agreement Set]\label{def:agreement-set}
For latent tuples $x,y\in\mathcal X$, their \emph{agreement set} is
\[
\operatorname{Agr}(x,y) := \{i\in[d] : x_i = y_i\}.
\]
\end{definition}

Agreement sets are the coordinate bookkeeping device behind confusability. A coordinate-subset view $S$ gives the same transcript on $x$ and $y$ exactly when $S\subseteq\operatorname{Agr}(x,y)$. Agreement patterns that create ambiguity are therefore closed upward.

\begin{definition}[Upward-Closed Family]\label{def:upward-closed-family}
A family $\mathcal U\subseteq 2^{[d]}$ is \emph{upward-closed} if whenever $S\in\mathcal U$ and $S\subseteq T\subseteq[d]$, one also has $T\in\mathcal U$.
\end{definition}

For a view family $\mathcal V$, the generated upward-closed family is
\[
\mathcal U_{\mathcal V}:=\{S\subseteq[d]: \exists V\in\mathcal V,\ V\subseteq S\}.
\]

\begin{definition}[Confusability Graph]\label{def:confusability-graph}
Given latent state space $\mathcal X$ and admissible view family $\mathcal V$, the \emph{confusability graph} is the graph
\[
G_{\mathrm{conf}}=(\mathcal X,E_{\mathrm{conf}})
\]
whose vertices are the latent states and whose edges are the unordered pairs $\{x,y\}$ with $x\neq y$ for which there exists an admissible view $V_\ell\in\mathcal V$ such that $\pi_{V_\ell}(x)=\pi_{V_\ell}(y)$. Equivalently, $\{x,y\}\in E_{\mathrm{conf}}$ iff some admissible view is contained in $\operatorname{Agr}(x,y)$.
\end{definition}

The definition applies to any projection-valued state family $\mathcal X\subseteq A_1\times\cdots\times A_d$. The complete agreement-family invariant below is the full-tuple specialization $\mathcal X=A_1\times\cdots\times A_d$: in that case every proper coordinate subset is realized as the agreement set of an actual pair of distinct states. For restricted state sets, the same adjacency definition is used, but some agreement patterns can be absent; Section~\ref{sec:equality} gives the corresponding local-composition criterion. For $d$ coordinates over alphabet size $q$, explicit graph materialization has $q^d$ vertices and $q^{2d}$ ordered state pairs, while an adjacency query takes $O(d+\sum_\ell |V_\ell|)$ time by computing $\operatorname{Agr}(x,y)$ and scanning the view list.

In the exact coordinate-view model, confusability depends only on agreement sets. Any admissible view family $\mathcal V$ defines the upward-closed family $\mathcal U_{\mathcal V}$ above, and two distinct tuples are adjacent if and only if their agreement set lies in $\mathcal U_{\mathcal V}$. Conversely, every upward-closed family arises from some coordinate-view family. The notation $\mathcal{CVG}$ denotes the resulting class of full-tuple coordinate-view confusability graphs, with labeled coordinates and fixed coordinate alphabets understood from context.

\leanmetapending{\LH{MFT202}}
\begin{theorem}[Canonical presentation]\label{prop:agreement-family-complete-invariant}
Assume every coordinate alphabet is nontrivial, i.e. $|A_i|\ge 2$ for each coordinate $i$, and assume both finite view lists are nonempty. For two view families $\mathcal V$ and $\mathcal W$ on the same labeled coordinate set, the labeled confusability graphs on the full tuple space are equal if and only if
\[
\mathcal U_{\mathcal V}=\mathcal U_{\mathcal W}.
\]
Equivalently, $\mathcal U_{\mathcal V}$ is a complete labeled invariant inside the nonempty full-tuple class $\mathcal{CVG}$. Moreover, each generated upward family has a unique minimal presentation by the cofinal antichain of its inclusion-minimal members, obtained from any view list by deleting duplicates and redundant supersets.
For an explicit bit-vector view list, the cofinal antichain is computed in \(O(L^2d)\) time, and labeled equality of two full-tuple presentations is checked in polynomial time without materializing either graph.
\end{theorem}

\begin{proof}
If the upward families are equal, the agreement-set characterization gives the same edge relation for every state pair. Conversely, suppose the two labeled graphs have the same edge relation. Fix a proper subset $S\subsetneq[d]$. For each coordinate $i$, choose two distinct symbols $0_i,1_i\in A_i$. Define two tuples by setting $x_i=y_i=0_i$ for $i\in S$ and $x_i=0_i$, $y_i=1_i$ for $i\notin S$. Then $x\ne y$ and $\operatorname{Agr}(x,y)=S$. Every proper coordinate subset is therefore realized as the agreement set of some distinct pair of tuples, so equality of edge relations forces membership agreement on all proper subsets. Since the view lists are nonempty, the full coordinate set belongs to both generated upward families. Hence the two upward families are equal.
\end{proof}

\paragraph{Degenerate edgeless case}
The nonempty-view-list hypothesis removes the only full-tuple degeneracy in the labeled invariant. An empty view list has $\mathcal U_{\mathcal V}=\emptyset$ and induces the edgeless graph, while the one-view list $\{[d]\}$ has $\mathcal U_{\mathcal V}=\{[d]\}$ and also induces the edgeless graph because no distinct pair agrees on all coordinates. No other recovery failure occurs: actual state pairs realize every proper agreement set, and once the view list is nonempty the full set $[d]$ is forced into the generated upward family.

The \(O(L^2d)\) normalization deletes duplicate views and every view that contains another listed view.

Coordinate relabeling transports cofinal antichains and yields graph isomorphisms even when the original view lists differ in length. Conversely, any graph isomorphism that preserves agreement sets up to a coordinate permutation forces the generated agreement family to be the corresponding coordinate image. Coordinatewise alphabet permutations preserve confusability and induce graph automorphisms.

\subsection{The Binary Square Witness}\label{sec:binary-square}

\begin{figure}[!b]
\centering
  \resizebox{0.82\textwidth}{!}{%
  \begin{tikzpicture}
    \begin{scope}[xshift=0cm, yshift=-1.2cm]
      \node at (1.5, 3.2) {\textbf{(a) Clique ambiguity}};
      \node at (1.5, 2.7) {\small Total ambiguity};

      \node[circle, draw, fill=blue!20, minimum size=9mm, inner sep=1pt] (a1) at (0, 0) {$(0,0)$};
      \node[circle, draw, fill=blue!20, minimum size=9mm, inner sep=1pt] (a2) at (3, 0) {$(0,1)$};
      \node[circle, draw, fill=blue!20, minimum size=9mm, inner sep=1pt] (a3) at (3, 2) {$(1,1)$};
      \node[circle, draw, fill=blue!20, minimum size=9mm, inner sep=1pt] (a4) at (0, 2) {$(1,0)$};

      \draw[thick] (a1) -- (a2) -- (a3) -- (a4) -- (a1);
      \draw[thick] (a1) -- (a3);
      \draw[thick] (a2) -- (a4);

      \node[text width=4.3cm, align=center] at (1.5, -1.3) {\small Every state confusable\\with every other};
    \end{scope}

    \begin{scope}[xshift=6.5cm, yshift=-1.2cm]
      \node at (1.5, 3.2) {\textbf{(b) Cycle ambiguity}};
      \node at (1.5, 2.7) {\small Partial ambiguity};

      \node[circle, draw, fill=green!20, minimum size=9mm, inner sep=1pt] (b1) at (0, 0) {$(0,0)$};
      \node[circle, draw, fill=green!20, minimum size=9mm, inner sep=1pt] (b2) at (3, 0) {$(0,1)$};
      \node[circle, draw, fill=green!20, minimum size=9mm, inner sep=1pt] (b3) at (3, 2) {$(1,1)$};
      \node[circle, draw, fill=green!20, minimum size=9mm, inner sep=1pt] (b4) at (0, 2) {$(1,0)$};

      \draw[thick] (b1) -- (b2);
      \draw[thick] (b2) -- (b3);
      \draw[thick] (b3) -- (b4);
      \draw[thick] (b4) -- (b1);

      \draw[dashed, gray, thick] (b1) -- (b3);
      \draw[dashed, gray, thick] (b2) -- (b4);

      \node[text width=4.3cm, align=center] at (1.5, -1.3) {\small Some pairs distinguishable;\\2 labels suffice};
    \end{scope}
  \end{tikzpicture}%
  }%
\caption{Two ambiguity topologies. In (a), the view family reveals no coordinate, so the confusability graph is a clique: every state is confusable with every other, and exact recovery requires one label per state. In (b), single-coordinate views produce the 4-cycle: solid segments are confusability edges, while dashed diagonals are non-edges between distinguishable opposite corners. The 4-cycle is 2-colorable, so exact recovery is possible with two labels.}
\label{fig:ambiguity-topology}
\end{figure}

For the binary square \(\{0,1\}^2\) with views \(\mathcal V=\{\{1\},\{2\}\}\), two distinct states are confusable exactly when they agree in one coordinate. The generated upward family is \(\{\{1\},\{2\},\{1,2\}\}\), and the confusability graph is the 4-cycle shown in Figure~\ref{fig:ambiguity-topology}(b):
\[
(0,0) \sim (0,1) \sim (1,1) \sim (1,0) \sim (0,0),
\]
with opposite corners nonadjacent. The parity map \(c(x_1,x_2)=x_1\oplus x_2\) is a proper \(2\)-coloring, so exact recovery needs exactly one auxiliary bit. With one label value, a decoder can be exact on at most two of the four states under the uniform source.

\subsection{Exact Recovery and Block Composition}\label{sec:graph-recovery}

The zero-error coloring and strong-product reductions below are classical~\cite{witsenhausen1976zero,korner2002zero}. In the coordinate-view model they fix the graph invariants used by the affine certificate theory.

\begin{definition}[Strong Product]\label{def:strong-product}
For graphs $G$ and $H$, the \emph{strong product} $G\boxtimes H$ has vertex set $V(G)\times V(H)$. Two distinct vertices $(u,v)$ and $(u',v')$ are adjacent exactly when either $u=u'$ and $v\sim v'$, or $u\sim u'$ and $v=v'$, or $u\sim u'$ and $v\sim v'$.
\end{definition}

The strong product is the standard graph product~\cite{imrich2008product}. Block composition of coordinate-view systems realizes exactly this graph operation.

\leanmetapending{\LHrng{MFT}{3}{4}, \LH{MFT9}, \LHrng{MFT}{14}{16}, \LH{MFT20}}
Exact zero-error recovery with a \(T\)-ary auxiliary label is equivalent to \(T\)-colorability of the induced confusability graph~\cite{witsenhausen1976zero,korner2002zero,west2001introduction}. Exact decoding on a designated success set \(S\) is the same condition on the induced subgraph \(G[S]\). Hence, for \(T>0\),
\[
N_T := \max\{|S| : S\subseteq \mathcal X \text{ induces a $T$-colorable subgraph}\}.
\]
The optimal exact success probability under the uniform source is \(N_T/|\mathcal X|\). For a nonuniform finite source, the same reduction maximizes source mass over \(T\)-colorable induced subgraphs. Product colorings multiply under block composition: if the component graphs are \(T_1\)- and \(T_2\)-colorable, the block-composed graph is \(T_1T_2\)-colorable.

\paragraph{Block-composed product system}
A block-composed coordinate product has product latent space, componentwise coordinates in disjoint blocks, and admissible product views obtained by applying one admissible view from each component to the corresponding block.

\leanmetapending{\LHrng{MFT}{25}{26}}
\begin{proposition}[Block composition gives strong product]\label{prop:block-composition-strong-product}
The confusability graph of the block-composed product of two coordinate-view systems is the strong product of the two component confusability graphs.

\end{proposition}

\begin{proof}
Let the component graphs be $G_1$ and $G_2$. A product state is a pair $(x_1,x_2)$. Two distinct product states are confusable exactly when an allowed product view fails to separate them. Product-view confusability has precisely the three strong-product cases: $x_1=x_1'$ and $x_2\sim_{G_2}x_2'$, or $x_1\sim_{G_1}x_1'$ and $x_2=x_2'$, or both component pairs are confusable. These are exactly the adjacency clauses in Definition~\ref{def:strong-product}.
\end{proof}

Clique lower bounds multiply under this product: cliques of sizes $c_1$ and $c_2$ in the two components give a clique of size $c_1c_2$ in the product.

The affine specialization now supplies polynomial-time rank certificates for the graph invariants induced by coordinate-view presentations.

\section{Affine Coordinate Matroids and Polynomial-Time Bounds}\label{sec:affine}

Affine structure makes rank certificates for Shannon-capacity upper bounds computable before graph materialization. The coordinate functionals form a representable matroid; its rank function supplies the certificate. The same state family generates both this coordinate matroid and the confusability graph, so restricted coordinate ranks give polynomial-time upper bounds on the independence number and Shannon capacity $\Theta(G)$. For affine presentations in $\mathcal{CVG}$, capacity certification is performed on explicit linear data rather than on the exponentially large graph. The input model is explicit: the affine family $A=a_0+V$ is given by a linear presentation of $V$, such as a basis or generator matrix over $\mathbb F_q$, together with the coordinate-subset view family. Under the explicit representation, each relevant quantity is the rank of a restricted coordinate map, computed from the coordinate presentation rather than from the vertex set of the induced graph.

Two hypotheses are separate. The representation theorem and intersection-closure dichotomy are full-tuple results; their proofs realize every proper agreement set by actual state pairs or triples. The affine rank bounds below apply to arbitrary affine state families, including proper affine subspaces. For those restrictions, exact cluster collapse is governed by the local-composition criterion of Section~\ref{sec:equality}. The capacity framework in Sections~\ref{sec:capacity-theory}--\ref{sec:theta-theory} applies to the finite graph induced by the chosen presentation.

Logarithms follow the global base-$2$ convention. The rank inequalities are uniform under any fixed logarithm base.

If $V$ is presented by an $r\times d$ generator matrix $M$, then for a coordinate set $S\subseteq[d]$ the quantity $t(S)$ is the rank of the submatrix $M_S$ formed by the columns indexed by $S$. Arithmetic-operation counts in this section are field-operation counts over $\mathbb F_q$. A single rank computation costs $O(r|S|\min\{r,|S|\})$ field operations by standard Gaussian elimination, and all ranks for a view family $\mathcal V$ of size $L$ are computable in $O(Lrd\min\{r,d\})$ field operations by independent eliminations. For fixed $q$, these are the usual finite-field arithmetic counts.

\paragraph{Bit complexity}
When $q$ is part of the input, assume the field is supplied by an explicit dense representation, so each field element uses $O(\log q)$ bits and each field operation has bit complexity polynomial in $\log q$. Under that convention the rank certificate has bit complexity
\[
O\!\left(Lrd\min\{r,d\}\operatorname{poly}(\log q)\right)
\]
after reading the generator matrix and explicit bit-vector view list. Multiplying the field-operation count above by the polynomial bit cost of each arithmetic operation in the chosen dense representation gives the displayed bit bound. No step enumerates the $q^r$ affine states or the $q^{2r}$ candidate state pairs of the materialized graph. General Shannon-capacity analysis starts from the materialized graph; Lov\'asz-$\vartheta$ bounds require semidefinite optimization~\cite{alon2006shannon,deboer2024asymptotic,grotschel1981ellipsoid}. The rank certificate is polynomial in the explicit affine presentation.

\paragraph{Rank-computation example}
Let
\[
M=\begin{bmatrix}1&0&1\\0&1&1\end{bmatrix}.
\]
The matrix above is a generator matrix, and the presented subspace $V=\operatorname{rowspan}(M)$ lies in the three-dimensional vector space over $\mathbb F_2$.
The field size in this example is $q=2$. For the view family $\mathcal V=\{\{1,2\},\{3\}\}$, Gaussian elimination gives $t(\{1,2\})=2$ and $t(\{3\})=1$. Theorem~\ref{thm:matroid-capacity-bounds}, stated below, therefore gives
\[
\Theta(G)\le \min\{t(\{1,2\}),t(\{3\})\}\log q=\min\{2,1\}\log 2=\log 2,
\]
in the fixed logarithm base. The certificate uses only the displayed generator matrix and the view family.

\paragraph{A dependent-coordinate example}
For a four-coordinate affine presentation over $\mathbb F_2$, let
\[
M=\begin{bmatrix}1&0&0&1\\0&1&0&0\\0&0&1&0\end{bmatrix},
\]
so $V=\operatorname{rowspan}(M)$ has dimension $3$. For the view family $\mathcal V=\{\{1,2\},\{2,3\},\{1,4\}\}$, the restricted ranks are $2$, $2$, and $1$, respectively; the last rank is $1$ because columns $1$ and $4$ coincide on $V$. Theorem~\ref{thm:matroid-capacity-bounds} gives
\[
\Theta(G)\le \min\{2,2,1\}\log 2=\log 2,
\]
strictly below the trivial state-count bound $\log_2 |A|=3$. The gain comes from a matroid dependence among observed coordinates rather than from the binary-square four-cycle geometry.

The same instance has a Haemers-type interpretation~\cite{haemers1978upper}. The rank-one view $\{1,4\}$ has only two realized projection values, so its fiber-equality matrix gives a rank-$2$ Haemers-style certificate for the full confusability graph. The certificate is a feasible Haemers point constructed from the presentation in polynomial time, before the eight-vertex graph is materialized.

\leanmetapending{\LHrng{AFM}{1}{5}, \LHrng{AFM}{20}{22}}
\begin{theorem}[Matroidal specialization of affine coordinate views]\label{thm:affine-fact-matroid}
Assume the state family is affine over a field, so that the valid latent tuples form an affine translate of a linear subspace of the ambient coordinate space. For any coordinate set $S$ and coordinate index $i$, say that $S$ \emph{semantically determines} $i$ when any two states that agree on every coordinate in $S$ also agree on coordinate $i$. Then semantic determination of $i$ by $S$ is equivalent to membership of the coordinate functional for $i$ in the linear span of the coordinate functionals indexed by $S$. Consequently the coordinate indices carry a representable matroid: a coordinate set is independent exactly when its coordinate functionals are linearly independent, the minimal determining coordinate sets are exactly the matroid bases, and every minimal determining set has cardinality equal to both the finite-dimensional span rank of all coordinate functionals and the full coordinate-rank value $t([d])$.
\end{theorem}

\begin{proof}
Let the affine state family be $A=a_0+V$, where $a_0$ is a fixed origin and $V$ is a linear subspace of the ambient coordinate space. Fix a coordinate set $S$ and a coordinate index $i$.

By definition, coordinate $i$ is semantically determined by $S$ exactly when for all $x,x'\in A$, agreement on every coordinate in $S$ forces agreement on coordinate $i$. Writing $\delta=x-x'$, the condition is equivalent to the statement that every direction vector $\delta\in V$ whose coordinates in $S$ vanish also satisfies $\delta_i=0$.

Let $\varepsilon_j:=e_j|_V$ denote the $j$th coordinate functional restricted to the direction space $V$. The preceding condition says precisely that
\[
\bigcap_{j\in S}\ker(\varepsilon_j)\subseteq \ker(\varepsilon_i).
\]
The finite-dimensional duality step is the elementary annihilator identity~\cite[Chap.~2]{roman2008advanced}
\[
\begin{aligned}
\bigcap_{j\in S}\ker(\varepsilon_j)\subseteq \ker(\varepsilon_i)
&\Longleftrightarrow
\varepsilon_i\in \left(\bigcap_{j\in S}\ker(\varepsilon_j)\right)^\perp\\
&\Longleftrightarrow
\varepsilon_i\in \operatorname{span}\{\varepsilon_j:j\in S\}.
\end{aligned}
\]
Indeed, the first equivalence says exactly that $\varepsilon_i$ vanishes on the common kernel, and the second uses $\left(\bigcap_j\ker(\varepsilon_j)\right)^\perp=\operatorname{span}\{\varepsilon_j\}$ in finite dimension. Thus semantic determination by $S$ is exactly span membership of the corresponding coordinate functional.

The coordinate functionals therefore realize a representable matroid on coordinate indices. In that matroid, independence is linear independence of the coordinate functionals, bases are the maximal independent spanning sets, and every basis has the common rank cardinality. Since span corresponds exactly to semantic determination, a minimal determining coordinate set is a minimal spanning set of the matroid, and minimal spanning sets are bases. The common cardinality is the finite rank of the span of all coordinate functionals, equivalently the rank of the full coordinate set.
\end{proof}

The matroid and the confusability graph are associated objects induced by the same affine state family. Under the common specialization to finite affine families with coordinate-projection views, the matroid rank provides \emph{computationally tractable} upper bounds on graph invariants before the induced graph is enumerated. For each coordinate set $S$ the quantity $t(S)$ is the rank of the restricted coordinate map $\pi_S|_V$. The basic size bounds are immediate from matrix rank: $t(S)\le |S|$ and $t(S)\le \dim V$; moreover $t(S)=|S|$ on linearly independent coordinate families, while $t(S)$ reaches the full ambient determining rank exactly on determining sets. Hence each $t(S)$ is computable by Gaussian elimination on an explicit presentation of $V$, and the bounds of Theorem~\ref{thm:matroid-capacity-bounds} are computable in time polynomial in the bit-size of the generator matrix and view list, even when the state space itself is exponentially large in the presentation dimension. Exact Shannon capacity has open computational complexity for general graphs, while these affine coordinate-view certificates reduce to linear matroid rank.

\subsection{View-fiber dimensions and capacity}
\label{sec:affine-to-graph}

Let $A=a_0+V$ with $V$ an $r$-dimensional linear subspace over a finite field $\mathbb{F}_q$, and let admissible views be coordinate-subset projections. For a coordinate set $S$ write $t(S)$ for the rank of the corresponding coordinate functionals on $V$.

\leanmetapending{\LHrng{AFM}{6}{12}, \LH{MFT86}}
\begin{proposition}[View-fiber clique sizes]\label{prop:view-fiber-clique-sizes}
Let $A$ and $t(S)$ be as defined for the affine presentation. Each projection-fiber for view $S$ is an affine subspace of dimension $r-t(S)$ and size $q^{r-t(S)}$. The single-view confusability graph $G_S$ is therefore a disjoint union of $q^{t(S)}$ cliques each of size $q^{r-t(S)}$. In particular
\[
\alpha(G_S)=q^{t(S)},\qquad \Theta(G_S)=t(S)\log q,
\]
where $\alpha$ is the independence number and $\Theta$ is Shannon capacity in the global base-$2$ convention.
\end{proposition}

\begin{proof}
The projection onto coordinates $S$ is an affine-linear map whose linear part restricted to $V$ has rank $t(S)$. Its kernel is therefore a linear subspace of $V$ of dimension $r-t(S)$, so every fiber is an affine translate of that kernel and has cardinality $q^{r-t(S)}$. The fibers partition $A$, and within each fiber every pair of states is confusable under $S$, so $G_S$ is a disjoint union of equal-size cliques. The number of fibers equals $|A|/q^{r-t(S)}=q^{t(S)}$, so one can select one representative per fiber to obtain an independent set of size $q^{t(S)}$; conversely, any independent set contains at most one vertex from each clique. Hence $\alpha(G_S)=q^{t(S)}$.

For block powers, a vertex of $G_S^{\boxtimes n}$ records an $n$-tuple of projection fibers. Two vertices with the same fiber word lie in a strong product of cliques, hence in one clique; if their fiber words differ in some coordinate, then that coordinate is neither equal nor adjacent in $G_S$, so the two vertices are not adjacent in the strong power. Thus $G_S^{\boxtimes n}$ is a disjoint union of $q^{nt(S)}$ cliques, one for each fiber word, and $\alpha(G_S^{\boxtimes n})=q^{nt(S)}$. Normalization yields $\Theta(G_S)=t(S)\log q$.
\end{proof}

\leanmetapending{\LHrng{AFM}{6}{12}}
\begin{theorem}[Polynomial-time matroid capacity certificates]\label{thm:matroid-capacity-bounds}
Let $G$ be the full confusability graph generated by a family of coordinate-subset views $\mathcal V$. Then
\[
\alpha(G)\le \min_{S\in\mathcal V} q^{t(S)}
\qquad\text{and}\qquad
\Theta(G) \le \min_{S\in\mathcal V} t(S)\log q.
\]
The matroid rank function $t(\cdot)$ on coordinate sets supplies explicit upper bounds on the independence number and Shannon capacity $\Theta(G)$ of the induced confusability graph. The certificate is computed from the coordinate presentation and view list; the proof only compares it to the graph after the bound has already been obtained. The same inequality is uniform under any fixed logarithm base.
\end{theorem}

The certificate is strictly pre-materialization in the input parameters of the affine presentation. The explicit input has size $O(rd\log q+Ld)$ bits, up to the fixed overhead for the dense field representation. After reading the $r\times d$ generator matrix and the $L$ bit-vector views, the direct computation uses $O(Lrd\min\{r,d\})$ field operations, or $O(Lrd\min\{r,d\}\operatorname{poly}(\log q))$ bit operations in the dense finite-field model. No step enumerates the $q^r$ affine states or the $q^{2r}$ candidate pairs.

\begin{proof}
An independent set for $G$ must be independent in each single-view graph $G_S$, hence its size is at most $\alpha(G_S)=q^{t(S)}$ for every $S\in\mathcal V$. Since $G$ contains $G_S$ as an edge subgraph for each $S$, monotonicity under strong powers gives $\alpha(G^{\boxtimes n})\le \alpha(G_S^{\boxtimes n})$ for every $n$. Taking the per-letter limit yields $\Theta(G)\le \Theta(G_S)=t(S)\log q$, and the bound holds after minimizing over $S\in\mathcal V$.
\end{proof}

\subsection{Kernel-avoiding sections and exact capacity}
\label{sec:kernel-section-exactness}

The upper certificate becomes an exact capacity formula when it is met by a transversal of a rank-minimizing projection. Let
\[
t^*:=\min_{S\in\mathcal V} t(S),
\qquad
K_{\mathcal V}:=\bigcup_{W\in\mathcal V}\ker(\pi_W|_V).
\]
For a rank-minimizing view $S$, the fibers of $\pi_S:A\to\pi_S(A)$ form $q^{t^*}$ ambiguity classes. An independent set can contain at most one point from each $S$-fiber. Hence $\alpha(G)=q^{t^*}$ holds if and only if $\pi_S$ has a section
\[
\sigma:\pi_S(A)\to A,\qquad \pi_S(\sigma(y))=y,
\]
whose image is pairwise nonconfusable, equivalently
\[
\sigma(y)-\sigma(y')\notin K_{\mathcal V}
\quad\text{for all }y\ne y'.
\]
The condition is set-theoretic; it does not require the section to be linear. Once such a one-shot independent set exists, capacity exactness follows by repetition: $I^n$ is independent in $G^{\boxtimes n}$ for every $n$, so the lower rate matches the matroid upper bound.

\leanmetapending{\LHrng{AKS}{1}{5}}
\begin{theorem}[Kernel-section exactness]\label{thm:kernel-section-exactness}
Let $A=a_0+V\subseteq\mathbb F_q^d$ be an affine coordinate-view system with nonempty view family $\mathcal V$. Let $t^*=\min_{S\in\mathcal V}t(S)$. Suppose there is a linear subspace $R\subseteq V$ such that
\[
\dim R=t^*,
\qquad
R\cap\ker(\pi_W|_V)=\{0\}\quad\text{for every }W\in\mathcal V.
\]
Then the matroid rank certificate is exact:
\[
\alpha(G)=q^{t^*},
\qquad
\Theta(G)=t^*\log q.
\]
\end{theorem}

\begin{proof}
The set $I:=a_0+R$ has cardinality $q^{t^*}$. If $x=a_0+r$ and $x'=a_0+r'$ are two distinct elements of $I$, then $r-r'$ is a nonzero vector of $R$. The kernel-avoidance assumption gives $r-r'\notin \ker(\pi_W|_V)$ for every admissible view $W$, so no admissible view identifies $x$ and $x'$. Thus $I$ is independent in $G$.

Theorem~\ref{thm:matroid-capacity-bounds} gives $\alpha(G)\le q^{t^*}$ and $\Theta(G)\le t^*\log q$. Since $I$ is independent and has size $q^{t^*}$, equality holds for $\alpha(G)$. Moreover $I^n$ is independent in the strong power $G^{\boxtimes n}$, so
\[
\alpha(G^{\boxtimes n})\ge |I|^n=q^{nt^*}.
\]
The normalized rates therefore have lower bound $t^*\log q$, which matches the upper bound.
\end{proof}

\leanmetapending{\LH{AKS32}}
\begin{corollary}[No-Shannon-gap kernel-section instances]\label{cor:kernel-section-no-gap}
Under the hypotheses of Theorem~\ref{thm:kernel-section-exactness},
\[
\Theta(G)=\log\alpha(G)=t^*\log q.
\]
Consequently, a graph with $\Theta(G)>\log\alpha(G)$ cannot be certified exactly by a kernel-avoiding linear section.
\end{corollary}

\begin{proof}
Theorem~\ref{thm:kernel-section-exactness} gives $\alpha(G)=q^{t^*}$ and $\Theta(G)=t^*\log q$, so the displayed equality follows. The final sentence is the contrapositive.
\end{proof}

The theorem identifies the algebraic object behind the one-shot sandwich: a $t^*$-dimensional subspace whose nonzero vectors avoid the union of admissible view kernels. Set-theoretic sections tensorize by the same no-gap mechanism. If $\alpha(G)=q^{t^*}$, then the independent set $I$ witnessing equality has $I^n$ independent in $G^{\boxtimes n}$, while Theorem~\ref{thm:matroid-capacity-bounds} gives $\alpha(G^{\boxtimes n})\le q^{nt^*}$; hence $\Theta(G)=\log\alpha(G)=t^*\log q$. Conversely, if $\Theta(G)=\log\alpha(G)=t^*\log q$, then $\alpha(G)=q^{t^*}$, and any maximum independent set is a kernel-avoiding set section of every rank-minimizing projection. Linearity supplies the algebraic certificate used by Theorem~\ref{thm:kernel-section-exactness}.

\leanmetapending{\LHrng{AKS}{20}{21}, \LHrng{AKS}{36}{38}}
\begin{proposition}[Support-hitting characterization]\label{prop:support-hitting-sections}
Let $R\subseteq V$ be a linear subspace with $\dim R=t^*$. Suppose that every nonzero $r\in R$ meets every admissible view:
\[
\operatorname{supp}(r)\cap W\ne\varnothing
\qquad\text{for every }0\ne r\in R\text{ and every }W\in\mathcal V.
\]
Then $R$ is kernel-avoiding, and therefore $\Theta(G)=t^*\log q$. Conversely, every kernel-avoiding subspace satisfies the displayed support-hitting condition.

In coding-theoretic language, it is enough that $R$ be a $[d,t^*,\delta]_q$ subcode of $V$ with
\[
\delta>d-\min_{W\in\mathcal V}|W|.
\]
In particular, an MDS subcode of dimension $t^*$ gives exactness whenever every admissible view has size at least $t^*$~\cite{macwilliams1977theory}.
\end{proposition}

\begin{proof}
If $0\ne r\in R$ lies in $\ker(\pi_W|_V)$, then all coordinates of $r$ indexed by $W$ vanish, so $\operatorname{supp}(r)\cap W=\varnothing$. The support-hitting hypothesis excludes this for every admissible $W$, hence $R$ is kernel-avoiding. Theorem~\ref{thm:kernel-section-exactness} gives the capacity formula.

Conversely, if $R$ is kernel-avoiding and some nonzero $r\in R$ misses an admissible view $W$, then $\pi_W(r)=0$, so $r\in R\cap\ker(\pi_W|_V)$, contradicting kernel avoidance.

For the distance condition, let $s=\min_{W\in\mathcal V}|W|$. If a nonzero codeword of $R$ missed some admissible view $W$, its support would lie in the complement of $W$ and would have size at most $d-|W|\le d-s$, contradicting $\delta>d-s$. An MDS $[d,t^*]_q$ subcode has distance $d-t^*+1$, so the displayed inequality follows from $|W|\ge t^*$ for all admissible views.
\end{proof}

\leanmetapending{\LHrng{AKS}{36}{38}}
\begin{corollary}[Arrangement-relative support profile]\label{cor:arrangement-support-profile}
For an arbitrary view hypergraph $\mathcal V$, define the arrangement-relative support profile
\[
\Delta_j^{\mathcal V}(V):=
\max_{\substack{R\le V\\ \dim R=j}}
\min_{0\ne r\in R}
\min_{W\in\mathcal V} |\operatorname{supp}(r)\cap W|.
\]
Then a $j$-dimensional linear subspace avoiding every admissible view kernel exists exactly when
\[
\Delta_j^{\mathcal V}(V)\ge1.
\]
In particular, the linear kernel-section certificate at rank $t^*$ exists exactly when $\Delta_{t^*}^{\mathcal V}(V)\ge1$.
\end{corollary}

\begin{proof}
For a fixed $j$-dimensional subspace $R$, the inner minimum is at least $1$ exactly when every nonzero vector of $R$ has support meeting every admissible view. Proposition~\ref{prop:support-hitting-sections} identifies that condition with kernel avoidance. Maximizing over $j$-dimensional subspaces gives the equivalence.
\end{proof}

\leanmetapending{\LH{AKS34}}
\begin{corollary}[Subcode-distance boundary for all-cardinality views]\label{cor:subcode-distance-boundary}
Let $V\le\mathbb F_q^d$ and let the admissible views be all coordinate subsets of fixed size $s$. For a $j$-dimensional subcode $R\le V$, the support-hitting condition is equivalent to
\[
\min_{0\ne r\in R}\operatorname{wt}(r)>d-s.
\]
Consequently, if
\[
\Delta_j(V):=\max_{\substack{R\le V\\ \dim R=j}}\ \min_{0\ne r\in R}\operatorname{wt}(r)
\]
denotes the best minimum distance among $j$-dimensional subcodes of $V$, then a $j$-dimensional support-hitting subspace exists exactly when $\Delta_j(V)>d-s$. In particular, the linear kernel-section certificate at rank $t^*$ exists for the all-$s$-view family exactly when $\Delta_{t^*}(V)>d-s$.
\end{corollary}

\begin{proof}
Fix $R\le V$ and $0\ne r\in R$. The support of $r$ hits every $s$-coordinate view if and only if the complement of $\operatorname{supp}(r)$ has size smaller than $s$: otherwise an $s$-subset contained in the complement is a missed view, and conversely a missed view is contained in the complement. Since $|\operatorname{supp}(r)^c|=d-\operatorname{wt}(r)$, the condition is equivalent to $\operatorname{wt}(r)>d-s$. Quantifying over every nonzero $r\in R$ gives the displayed minimum-distance condition, and maximizing over $j$-dimensional subcodes gives the profile criterion.
\end{proof}

The arrangement-relative support profile is the curve between the general view-kernel arrangement and the scalar Hamming-threshold slices. For all $s$-views, the hypergraph hitting condition collapses to the inequality $\operatorname{wt}(r)>d-s$ for every nonzero selected codeword. The resulting profile is adjacent to, but distinct from, ordinary generalized Hamming weights~\cite{wei1991generalized}. Generalized Hamming weights minimize the union support of a $j$-dimensional subcode; kernel avoidance here asks for a $j$-dimensional subcode whose \emph{minimum nonzero word support} is large relative to the admissible-view hypergraph. Already for \(V=\mathbb F_q^3\), \(j=1\), and all two-coordinate views, the first generalized Hamming weight is \(1\), realized by a coordinate line, while \(\Delta^{\mathcal V}_1(V)=2\), realized by the all-ones line.

\leanmetapending{\LHrng{ABD}{1}{12}}
\paragraph{Arrangement diagnostics}
At rank one, the linear boundary has a classical arrangement count. A kernel-avoiding line exists exactly when the central subspace arrangement $\{K_W\}_{W\in\mathcal V}$ has a vector outside its union; equivalently, the complement count is positive. By inclusion--exclusion over the intersection lattice, the complement count is the characteristic polynomial of the arrangement evaluated at $q$~\cite{orlik1992arrangements}. The same finite inclusion--exclusion identity gives the forbidden-point sieve used below. The supplementary checker \texttt{char\_poly\_small\_arrangement} compares the direct count, Boolean inclusion--exclusion, and M\"obius-lattice count on small arrangements, including the binary square and the four-coordinate threshold kernels; for the latter it also verifies that five avoiding vectors exist but no avoiding plane exists.

\paragraph{Grassmannian avoidance formula}
The higher-dimensional analogue counts avoiding $t$-subspaces in the Grassmannian. Let $P_{\mathcal A}$ be the finite set of projective points contained in $\bigcup_{W\in\mathcal V}K_W$. For $B\subseteq P_{\mathcal A}$, write $\rho(B)$ for the dimension of the span of any representatives of the points in $B$, and put $\binom{a}{b}_q=0$ when $b<0$ or $b>a$. The number of $t$-dimensional linear subspaces avoiding every kernel is
\[
N_t(\mathcal A)
=
\sum_{B\subseteq P_{\mathcal A}}
(-1)^{|B|}
\binom{r-\rho(B)}{t-\rho(B)}_q .
\]
The proof is the ordinary sieve on forbidden projective points. A $t$-subspace avoids the kernel arrangement exactly when it contains no point of $P_{\mathcal A}$. For a fixed $B\subseteq P_{\mathcal A}$, containing all points of $B$ is equivalent to containing their span. If that span has dimension $\rho(B)$, quotienting by it leaves $\binom{r-\rho(B)}{t-\rho(B)}_q$ possible $t$-subspaces. The finite forbidden-point sieve gives the displayed alternating sum. For $t=1$, the formula collapses to the projective complement count, and multiplying by $q-1$ recovers the usual vector-count characteristic-polynomial evaluation.

\leanmetapending{\LHrng{ABD}{21}{23}, \LHrng{ABD}{79}{91}, \LHrng{SPB}{240}{245}, \LHrng{SPB}{288}{289}, \LHrng{SPB}{379}{383}}
\begin{theorem}[Signed-profile compression for Grassmannian avoidance]\label{thm:signed-profile-compression}
Let \(P\) be a finite projective forbidden set, let \(\rho(B)\) be the represented rank of \(B\subseteq P\), and let
\[
w_t(n):=\binom{r-n}{t-n}_q .
\]
Define the signed rank profile
\[
s_n(P,\rho):=
\sum_{\substack{B\subseteq P\\ \rho(B)=n}}(-1)^{|B|}.
\]
Then the Grassmannian avoidance count is computed by the finite decoder
\[
N_t(P,\rho)=\sum_{n=0}^{|P|} s_n(P,\rho)\,w_t(n).
\]
Consequently, for any semantic shell \(\Sigma\) on a family of physical arrangements and any fixed shell value \(\sigma\), either all arrangements in the fiber \(\Sigma^{-1}(\sigma)\) have the same \(N_t\) value, or the signed rank profile \(n\mapsto s_n(P,\rho)\) varies inside that same fiber. In particular, any shell that determines the signed rank profile determines \(N_t\) and the positivity predicate \(N_t>0\).
\end{theorem}

\begin{proof}
Group the forbidden-point sieve by the value of \(\rho(B)\). Since \(\rho(B)\le |P|\) for every represented forbidden subset, the grouped sum is finite over \(0,\dots,|P|\), and its coefficient at rank \(n\) is exactly \(s_n(P,\rho)\). The fixed-fiber alternative is the abstract refinement fact applied to the decoder: a target computed from an intermediate invariant is constant on any semantic fiber where the intermediate invariant is constant; the contrapositive gives a signed-profile collision whenever \(N_t\) varies in that fiber.
\end{proof}

The signed-profile theorem compresses higher-rank avoidance. The kernel-intersection lattice does not determine \(N_t\). The rank-sieve evaluation needs the signed rank profile of the represented forbidden-point matroid. Computing \(N_t\) reduces to computing that profile and evaluating the displayed \(O(|P|)\)-term decoder; residual hardness belongs to profile computation, not to the semantic lattice once the profile is supplied.

The supplementary checker \texttt{grassmannian\_avoidance\_polynomial} compares the formula with direct Grassmannian enumeration on small arrangements, including the four-coordinate threshold case where the line count is positive and the plane count is zero. The same check compares two arrangements of four kernel lines in $\mathbb F_3^3$ with the same full kernel-intersection profile: four collinear projective points yield no avoiding plane, while four points in general position yield three. The projectivized forbidden-point matroid carries information not present in the kernel-intersection data alone; Theorem~\ref{thm:signed-profile-compression} identifies the signed profile as the exact finite refinement read by the rank-sieve count.

For higher-rank sections, vector counting alone does not supply the boundary. The natural independence system whose finite sets are linearly independent and whose spans avoid every forbidden kernel is not a matroid: an explicit two-line arrangement in $\mathbb F_2^3$ violates exchange, as mirrored by the checker \texttt{avoidance\_system\_nonmatroid}. Schubert-count union bounds give another sufficient test by counting $t$-subspaces that meet a fixed forbidden subspace, but the checker \texttt{schubert\_union\_bound} shows that the bound certifies the singleton-view endpoint only for $d<q+1$, whereas the diagonal section works for every $d$. The rank-one characteristic polynomial is the first member of a projective forbidden-set counting theory; for \(t^*\ge2\), the exact count is the signed-profile decoder above.

\leanmetapending{\LHrng{AKS}{10}{16}}
\paragraph{Remark: nonlinear rank-tightness without a linear section}
A finite \(\mathbb F_3^6\) arrangement has \(t^*=2\), an independent transversal \(I\) of size \(9=3^{t^*}\), and no linear two-dimensional kernel-avoiding section. The construction uses a four-dimensional kernel \(H=\{v:v_1=v_2=0\}\) and a nine-point transversal whose difference set represents \(36\) projective directions disjoint from \(H\); all remaining projective directions become line kernels. The checker \texttt{nonlinear\_section\_witness} verifies \(|I|=9\), \(|I-I\setminus\{0\}|=72\), \(H\cap(I-I)=\{0\}\), the absence of a complete projective line in the allowed difference directions, and \(\omega=\chi=81\). Thus \(\alpha=3^{t^*}\) and \(\Theta=\log9\) hold by a nonlinear transversal while \(N_{t^*}=0\).

\leanmetapending{\LHrng{AKS}{6}{7}, \LHrng{AKS}{17}{18}}
\begin{lemma}[Finite-union subspace avoidance]\label{lem:finite-union-subspace-avoidance}
Let $E$ be a vector space over $\mathbb F_q$, and let $U_1,\dots,U_L$ be proper linear subspaces of $E$. If $L\le q$, then $\bigcup_{i=1}^L U_i$ is a proper subset of $E$.
\end{lemma}

\begin{proof}
The standard finite-field subspace-union bound is equivalent by duality to the affine blocking-set lower bound of Jamison and Brouwer--Schrijver~\cite{jamison1977covering,brouwer1978blocking}. Equivalently, at most $q$ proper subspaces cannot cover a vector space over $\mathbb F_q$. The proof is the usual induction on $L$. Suppose $U_1\cup\cdots\cup U_L=E$. Fix any $x\in U_1$ and choose $y\notin U_1$. If $x=0$, then $x$ lies in every remaining subspace. If $x\ne0$, the affine line $\{ax+y:a\in\mathbb F_q\}$ is disjoint from $U_1$, so its $q$ points must lie in the remaining $L-1\le q-1$ subspaces. Two points lie in the same remaining subspace; subtracting them forces $x$ into that subspace. Thus every $x\in U_1$ is covered by the remaining subspaces, reducing the cover and closing the induction.
\end{proof}

\leanmetapending{\LH{AFM24}, \LHrng{AKS}{6}{9}, \LHrng{AKS}{17}{18}}
\begin{theorem}[Field-size exactness]\label{thm:field-size-exactness}
Let $A=a_0+V\subseteq\mathbb F_q^d$ be an affine coordinate-view system with nonempty view family $\mathcal V$ of size $L$, and let $t^*=\min_{S\in\mathcal V}t(S)$. If $q\ge L$, then there exists a linear kernel-avoiding subspace $R\subseteq V$ of dimension $t^*$. Consequently
\[
\Theta(G)=t^*\log q.
\]
For fixed explicit affine input over fields satisfying $q\ge L$, Shannon capacity is computed exactly by the coordinate-matroid rank minimum.
\end{theorem}

\begin{proof}
Write $r=\dim V$ and $K_W=\ker(\pi_W|_V)$. Since $t(W)\ge t^*$ for every admissible $W$, each $K_W$ has dimension at most $r-t^*$.

Construct subspaces
\[
0=R_0\subseteq R_1\subseteq\cdots\subseteq R_{t^*}\subseteq V
\]
with $\dim R_j=j$ and $R_j\cap K_W=\{0\}$ for every $W\in\mathcal V$. The case $j=0$ is immediate. Suppose $R_j$ has been constructed with $j<t^*$. For each $W$ the sum $K_W+R_j$ is a proper subspace of $V$, because
\[
\dim(K_W+R_j)\le \dim K_W+\dim R_j\le (r-t^*)+j<r.
\]
By Lemma~\ref{lem:finite-union-subspace-avoidance}, the union of the $L$ forbidden subspaces $K_W+R_j$ is a proper subset of $V$. Choose $v\in V$ outside this union and set $R_{j+1}:=R_j+\langle v\rangle$.

The choice of $v$ makes $R_{j+1}$ one dimension larger than $R_j$. It also preserves kernel avoidance. If $u+av\in K_W$ with $u\in R_j$, then $a\ne0$ would imply
\[
v=a^{-1}\bigl((u+av)-u\bigr)\in K_W+R_j,
\]
contrary to the choice of $v$. Hence $a=0$, and then $u\in R_j\cap K_W=\{0\}$. Thus $R_{j+1}\cap K_W=\{0\}$ for every $W$. After $t^*$ steps, Theorem~\ref{thm:kernel-section-exactness} applies.
\end{proof}

\paragraph{Typical regimes under bounded view count}
The field-size theorem is distribution-free. In any random affine-presentation model whose sampled view list always satisfies \(L\le q\), the rank certificate is exact with probability one:
\[
\Theta(G)=\min_{S\in\mathcal V}t(S)\log q.
\]
Average-case tightness questions therefore begin in the dense-view regime \(L>q\), where the obstruction is the represented view-kernel arrangement rather than random graph materialization.

\leanmetapending{\LHrng{KSH}{1}{5}, \LHrng{KSH}{9}{11}}
\begin{theorem}[Rank-one linear kernel-section existence is NP-complete]\label{thm:rank-one-kernel-section-np-hard}
The explicit-input decision problem asking whether an affine coordinate-view presentation admits a linear kernel-avoiding section is NP-complete. NP-hardness already holds over $\mathbb F_2$ with rank minimum $t^*=1$.
\end{theorem}

\begin{proof}
Membership in \textsc{NP} is by a basis certificate. Given candidate basis vectors for $R$, Gaussian elimination verifies their rank, and for every admissible view $W$ it verifies that the restricted projection $\pi_W|_R$ has trivial kernel, equivalently full column rank on the displayed basis. The bit cost is polynomial in the generator matrix, field representation, and view list.

For hardness, use a polynomial reduction from $3$-\textsc{SAT}; NP-hardness of $3$-\textsc{SAT} is standard~\cite{karp1972reducibility}. Let $\varphi$ have Boolean variables $x_1,\dots,x_n$ and clauses $C_1,\dots,C_m$, each with three literal occurrences. Work over $\mathbb F_2$ with homogeneous parameter space
\[
E=\mathbb F_2^{\{z\}\cup[n]}.
\]
The anchor coordinate is the linear form $z$. For every literal occurrence, add one ambient coordinate: a positive occurrence $x_i$ contributes the linear form $x_i$, and a negative occurrence $\neg x_i$ contributes the linear form $z+x_i$. Let $V$ be the image of the resulting linear map from $E$ into the ambient coordinate space. The view family consists of the singleton anchor view and, for each clause, the three occurrence coordinates belonging to that clause. The anchor view has rank $1$, and every clause view has rank at least $1$, so $t^*=1$.

If $a:[n]\to\{0,1\}$ satisfies $\varphi$, take the homogeneous vector $(z,x)=(1,a)$. The anchor coordinate is nonzero. In each clause, at least one literal is true; by construction the corresponding homogeneous literal coordinate is nonzero. Thus the image vector spans a one-dimensional subspace of $V$ meeting none of the admissible view kernels nontrivially.

Conversely, suppose a one-dimensional kernel-avoiding subspace exists. Over $\mathbb F_2$ it has a unique nonzero vector. Kernel-avoidance for the anchor view forces the anchor coordinate of that vector to be $1$. Choose any homogeneous preimage $(1,a)$. Kernel-avoidance for the clause view says that at least one of the three corresponding literal coordinates is nonzero; on the chart $z=1$, the coordinate $x_i$ is the truth value of the positive literal and the coordinate $z+x_i$ is the truth value of the negative literal. Hence every clause of $\varphi$ is satisfied by $a$.

The construction uses one anchor coordinate, three literal-occurrence coordinates per clause, and one view per clause plus the anchor view. Its size is linear in the $3$-\textsc{SAT} instance. The homogeneous-coordinate equivalence between clause satisfaction and nonzero clause views is the standard Karp reduction wrapper.
\end{proof}

\leanmetapending{\LHrng{KSH}{6}{8}, \LHrng{KSH}{12}{13}}
\begin{corollary}[Rank-one nonlinear section existence is NP-complete]\label{cor:rank-one-nonlinear-section-np-hard}
The explicit-input decision problem asking whether an affine coordinate-view presentation over $\mathbb F_2$ with $t^*=1$ admits a nonlinear rank-one section meeting the rank bound is NP-complete.
\end{corollary}

\begin{proof}
Membership in \textsc{NP} uses the two representatives as a certificate. One checks that they lie in the two fibers of a rank-one projection and then scans the view list to verify that no admissible view identifies the pair. In the rank-one binary slice of the preceding reduction, a nonlinear exact section is a pair of representatives, one in each anchor fiber. The difference of those two representatives is the unique nonzero vector of a one-dimensional subspace. Pairwise nonconfusability of the two representatives says exactly that this difference avoids every admissible view kernel. Thus nonlinear rank-one section existence is equivalent to the linear kernel-section point formalized in the reduction, and the hardness part of Theorem~\ref{thm:rank-one-kernel-section-np-hard} applies.
\end{proof}

\leanmetapending{\LH{AKS35}, \LHrng{KSH}{1}{5}, \LHrng{KSH}{9}{11}}
\begin{corollary}[Fixed-rank linear kernel-section existence is NP-complete]\label{cor:fixed-rank-kernel-section-np-complete}
For every fixed integer $k\ge1$, the explicit-input decision problem asking whether an affine coordinate-view presentation over $\mathbb F_2$ with rank minimum $t^*=k$ admits a linear kernel-avoiding section is NP-complete.
\end{corollary}

\begin{proof}
Membership in \textsc{NP} is the same basis-certificate verification used in Theorem~\ref{thm:rank-one-kernel-section-np-hard}.

For hardness, start from the rank-one instance produced in that theorem. Let $E$ be its direction space and $K_W\le E$ the view kernels. Add a padding space $P=\mathbb F_2^{k-1}$, add padding coordinates realizing all coordinates of $P$, and replace each old admissible view $W$ by the union of $W$ with all padding coordinates. The new direction space is $E\oplus P$, and the kernel of the padded view is $K_W\oplus0$. Every view rank increases by $k-1$, so the rank minimum becomes $k$.

If the original instance has a kernel-avoiding line $R_0\le E$, then $R_0\oplus P$ is a $k$-dimensional subspace of $E\oplus P$ meeting every $K_W\oplus0$ only at $0$. Conversely, suppose a $k$-dimensional padded section $R\le E\oplus P$ avoids every $K_W\oplus0$. Since $\dim P=k-1<\dim R$, the projection of $R$ to $P$ has a nonzero kernel; equivalently, $R$ contains a nonzero vector $(e,0)$. Kernel avoidance of $R$ implies $e\notin K_W$ for every old view $W$, so $\langle e\rangle$ is a kernel-avoiding line for the original rank-one instance. The padding construction is linear in the size of the input for fixed $k$.
\end{proof}

\leanmetapending{\LH{ABD10}, \LH{AKS35}, \LHrng{KSH}{1}{5}, \LHrng{KSH}{9}{11}}
\begin{corollary}[Grassmannian avoidance-count positivity is NP-complete]\label{cor:grassmannian-avoidance-positivity-np-complete}
For every fixed integer $k\ge1$, the explicit-input decision problem asking whether the Grassmannian avoidance count
\[
N_k(\mathcal A)
\]
of the induced view-kernel arrangement over $\mathbb F_2$ is positive is NP-complete.
\end{corollary}

\begin{proof}
For the kernel arrangement $\mathcal A=\{K_W:W\in\mathcal V\}$, the finite forbidden-point sieve defining $N_k(\mathcal A)$ counts exactly the $k$-dimensional linear subspaces meeting every $K_W$ only in $0$. Hence $N_k(\mathcal A)>0$ exactly when a $k$-dimensional kernel-avoiding linear section exists. Membership in \textsc{NP} uses the same basis certificate as Corollary~\ref{cor:fixed-rank-kernel-section-np-complete}. NP-hardness is that corollary with the decision predicate rewritten as positivity of $N_k(\mathcal A)$.
\end{proof}

\leanmetapending{\LHrng{KSH}{14}{15}}
\begin{corollary}[Rank-one avoidance counting is parsimoniously \(\#\textsc{P}\)-hard]\label{cor:rank-one-avoidance-counting-hard}
Under the rank-one binary reduction above, accepted linear kernel-section certificates are in bijection with satisfying assignments of the source \(3\)-\textsc{SAT} instance. Consequently, computing the number of rank-one avoiding lines, equivalently \(N_1(\mathcal A)\), is \(\#\textsc{P}\)-hard under parsimonious reductions.
\end{corollary}

\begin{proof}
The map sends a satisfying assignment \(a\) to the accepted homogeneous certificate \((z,x)=(1,a)\). Conversely, any accepted certificate has anchor coordinate \(z=1\), so its remaining coordinates read back a satisfying assignment. Over \(\mathbb F_2\), a rank-one subspace has a unique nonzero vector, so accepted homogeneous certificates and avoiding lines are counted by the same number.
\end{proof}

\leanmetapending{\LHrng{ABD}{21}{91}, \LHrng{SPB}{1}{26}, \LHrng{SPB}{160}{205}, \LHrng{SPB}{379}{383}, \LHrng{GCB}{3}{42}, \LHrng{GCR}{1}{230}, \LHrng{SCR}{67}{87}}
The signed-profile compression theorem identifies the finite invariant read by \(N_k\): not the kernel-intersection lattice, but the signed rank profile of the projectivized forbidden-point matroid. Let \(P\) be the set of projective points lying in the forbidden union, and let \(\rho(B)\) be the vector-space dimension of the span of representatives of \(B\subseteq P\). Inclusion--exclusion gives
\[
N_t
=\sum_{B\subseteq P}(-1)^{|B|}
\binom{r-\rho(B)}{t-\rho(B)}_q .
\]
The full subset-rank function \(\rho\) determines the signed profile, and the signed profile computes \(N_t\) by Theorem~\ref{thm:signed-profile-compression}. Equal represented forbidden-point matroids therefore give equal rank-sieve values, but the full matroid is more information than a fixed \(t\)-level rank-sieve evaluation needs. The fixed-lattice discriminators below separate the kernel-intersection lattice from this signed-profile/refined-matroid data: the kernel-intersection lattice can identify arrangements whose signed profiles, projective matroids, and \(N_t\) values differ.

The matroid-versus-lattice distinction gives the boundary of the hardness theory. The fixed-rank NP-completeness theorem above is a decision theorem for section existence. Fixed-lattice hardness requires a stronger input family: the kernel-intersection lattice stays fixed while the represented forbidden-point matroid varies inside the fiber. Common-core lifts supply the transport mechanism.

\leanmetapending{\LHrng{ABD}{28}{35}, \LH{ABD42}, \LHrng{ABD}{48}{56}}
\begin{lemma}[Common-core lift]\label{lem:common-core-lift}
Let \(C\) be a common core and let \(K_i=C\times L_i\) be product lifts of quotient kernels \(L_i\). If \(0\subseteq L_i\) for every \(i\), and distinct quotient kernels meet exactly in \(0\), then every intersection of at least two lifted kernels is \(C\times 0\). Hence two quotient families satisfying these hypotheses have the same lifted multiway kernel-intersection shell. Moreover, for a graph section
\[
R_f=\{(f(u),u):u\in U\}
\]
with \(f(U)\subseteq C\), the lifted avoidance condition
\[
R_f\cap (C\times L_i)=R_f\cap(C\times0)\quad\text{for all }i
\]
is equivalent to the quotient avoidance condition \(U\cap L_i=0\) for all \(i\). Consequently the common-core lift preserves quotient-avoidance positivity, and in the finite counted form the lifted witness count is obtained by summing the allowed graph-section maps over the quotient-avoiding candidates.
\end{lemma}

\begin{proof}
The fixed-shell part is the common-core intersection calculation: once every pair of lifted kernels meets in \(C\times0\), any intersection containing at least two lifted kernels is already forced to be \(C\times0\). The graph-section part is the identity
\[
R_f\cap(C\times L_i)=\{(f(u),u):u\in U\cap L_i\},
\]
valid because \(f(U)\subseteq C\). Therefore equality with \(R_f\cap(C\times0)\) is exactly equality \(U\cap L_i=0\). Applying the equivalence across the finite view family gives positivity preservation; counting witnesses partitions them by the quotient candidate and then counts the admissible graph-section maps over that candidate.
\end{proof}

Arrangements of the form \(K_i=C\oplus\langle p_i\rangle\) therefore freeze the kernel-intersection lattice while leaving the quotient projective configuration \(\{p_i\}\) to carry higher dependencies.

The local four-line example realizes the split. The collinear and general-position quotient configurations have \(N_2=0\) and \(N_2=3\), respectively, while the common-core lift gives the same lifted kernel-intersection profile. The checker \texttt{quotient\_lift\_rank\_sieve\_bridge} verifies the lift multiplier \(3^2\) and the lifted cone rank-sieve values \(0\) and \(27\). Fixed lattice fibers that realize rich projective-matroid variation can carry hardness through that variation. Fibers that collapse to polynomially structured projective matroids belong on the dichotomy side.

The natural quotient-side hard source is the Crapo--Rota critical problem for representable matroids: it asks for a projective subspace of prescribed codimension disjoint from a represented point set, exactly the quotient-side form of the \(N_t>0\) question~\cite{crapo1970combinatorial,kung1996critical}. Graphic matroids make the source concrete. If \(M(G)\) is represented over \(\mathbb F_q\) by oriented edge-incidence columns and \(r=\operatorname{rank}M(G)\), then an avoiding \((r-1)\)-subspace is equivalent to a proper \(q\)-coloring of \(G\) modulo componentwise additive shifts and scalar rescaling; for a graph with \(c\) connected components and at least one edge,
\[
N_{r-1}(M(G))=\frac{P_G(q)}{q^c(q-1)} ,
\]
where \(P_G(q)\) is the chromatic polynomial value~\cite{tutte1954contribution}. The quotient-side positivity problem therefore contains ordinary \(q\)-colorability, including the NP-complete \(q=3\) case~\cite{garey1976simplified}. The checker \texttt{graphic\_matroid\_critical\_bridge} verifies the identity by direct coloring counts and by the rank-sieve formula on small graphs.

\leanmetapending{\LHrng{GCB}{3}{42}, \LHrng{GCR}{121}{125}}
\begin{corollary}[Graphic-matroid rank-sieve counting hardness]\label{cor:graphic-matroid-rank-sieve-counting-hard}
For every fixed integer \(q\ge3\), computing \(N_{r-1}(M(G))\) from an oriented incidence representation of a loopless graphic matroid over \(\mathbb F_q\) is \(\#\textsc{P}\)-hard.
\end{corollary}

\begin{proof}
For a graph with \(c\) connected components and at least one edge, the identity above gives
\[
P_G(q)=q^c(q-1)N_{r-1}(M(G)).
\]
The number \(c\) is computable in polynomial time from the graph. Thus an oracle for \(N_{r-1}(M(G))\) computes the fixed-\(q\) chromatic-polynomial value \(P_G(q)\). For every fixed \(q\ge3\), computing \(P_G(q)\), equivalently the number of proper \(q\)-colorings, is \(\#\textsc{P}\)-hard~\cite{jaeger1990complexity}. Hence the graphic-matroid rank-sieve count is \(\#\textsc{P}\)-hard.
\end{proof}

The critical problem supplies quotient-side hardness. The common-core lift is the transport step: it preserves quotient positivity and counts while fixing the complete multiway kernel-intersection profile across all same-shape sources. Size-indexed versions keep the semantic shell fixed inside each source-shape fiber and preserve satisfying-witness counts for clause-list sources.

\leanmetapending{\LHrng{SCB}{979}{1013}}
\begin{theorem}[Size-indexed fixed-lattice SAT and counting spine]\label{thm:fixed-lattice-sat-spine}
Fix a nonempty common core. There is a clause-list-to-arrangement construction with the following properties. For each fixed number of Boolean variables, the fixed-active common-core lift assigns every clause list on those variables to the same complete multiway kernel-intersection profile. Inside that fixed profile, lifted positivity is equivalent to satisfiability of the clause list, and the lifted witness count is exactly the number of satisfying assignments. The paired target consisting of positivity together with the witness count is transported parsimoniously. No invariant obtained by decoding that fixed multiway profile determines the positivity target, the count target, or the paired target.
\end{theorem}

\begin{proof}
The construction uses the SAT-list source with shape equal to the number of variables. The fixed-active common-core lift introduces a common core and gates the active kernels so that all same-shape clause lists have the same multiway intersection profile. The physical witness relation is conjunctive: one lifted witness satisfies all local clause obligations exactly when the original assignment satisfies every clause. Consequently the lifted positivity target is SAT, and the lifted witness count is the satisfying-assignment count.

The same-shape reductions are parsimonious for the count target and preserve the Boolean positivity target. The paired reduction records both targets together. Since there are same-shape satisfiable and unsatisfiable clause lists with the same fixed profile, and also same-shape lists with different satisfying-assignment counts, equality of the fixed profile cannot force equality of any of the three targets. Postcomposition with an arbitrary decoding of the multiway profile only discards information, so every decoded lattice/profile invariant has the same non-determination property.
\end{proof}

Finite graph-plane discriminators locate the obstruction. In $\mathbb F_2^4=\mathbb F_2^2\oplus\mathbb F_2^2$, candidate rank-two sections are graph planes $R_X=\{(u,Xu):u\in\mathbb F_2^2\}$. There are codimension-two kernel families with identical labeled intersection profile but different avoiding-graph-plane counts. There are also two-kernel pairs with identical meet and join profiles whose avoiding-graph-plane counts are \(8\) and \(0\). The checker \texttt{determinant\_slot\_search} verifies these pairs and searches nearby small graph-plane slots. Consequently, global SAT composition needs a conjunctive fixed-shell scaffold: a source instance of fixed shape must map into one semantic shell, and one shared witness \(X\) must satisfy all encoded local constraints exactly when the original clauses are simultaneously satisfiable.

The point-kernel specialization marks the limit of the fixed-lattice route. If every forbidden kernel is a distinct one-dimensional subspace, then the intersection lattice is fixed by the number of listed kernels. For $t=2$, the decision problem is whether a projective line avoids a listed point set $P$. The point-line case has a direct polynomial-time procedure from incidence counting. In $\operatorname{PG}(r-1,q)$, the number of projective lines is $\binom{r}{2}_q$, and each point lies on $\binom{r-1}{1}_q$ lines. If
\[
|P|\binom{r-1}{1}_q<\binom{r}{2}_q,
\]
then some line avoids $P$. Otherwise $|P|\ge(q^r-1)/(q^2-1)>q^{r-2}$, so for $r\ge3$ one has $q^r\le |P|^3$; direct enumeration of projective lines is polynomial in the explicit point-list size. The checker \texttt{fixed\_lattice\_line\_avoidance} records this diagnostic and rechecks the $\mathbb F_3^3$ four-point toggle. The checker \texttt{point\_kernel\_duality} verifies the equivalent dual form: a \(t\)-subspace avoiding listed projective points corresponds to an \((r-t)\)-dimensional dual rowspace whose forms do not all vanish on any listed point. A fixed-lattice hardness theorem therefore needs higher-dimensional kernels or another fixed-lattice representation in which the assignment constraints do not collapse to point-line blocking.

\paragraph{Parameterized enumeration}
The same decision problem is fixed-parameter tractable in the ambient affine rank and field bit-size. In an $r$-dimensional direction space over $\mathbb F_q$, the number of $t^*$-dimensional subspaces is the Gaussian binomial coefficient $\binom{r}{t^*}_q$, bounded above by $q^{r^2}$. Enumerating those subspaces and checking the $L$ restricted projection ranks decides linear section existence in
\[
q^{O(r^2)}\operatorname{poly}(L,r,d,\log q)
=2^{O(r^2\log q)}\operatorname{poly}(L,r,d,\log q)
\]
bit time in the dense finite-field model. The problem is FPT for parameter $(r,\log q)$, and for parameter $r$ over a fixed finite field. The NP-completeness results above locate the hard regime in growing presentation dimension and view-kernel arrangement.

The complementary parameterized boundary is sharp. The explicit-input linear problem is para-\textsc{NP}-complete when parameterized by \(t^*\), by \(q\), or by the pair \((q,t^*)\): hardness already occurs at \(q=2\) and \(t^*=1\), while membership in \textsc{NP} is the basis-certificate verification used above. The tractable parameter is the ambient affine rank \(r\), not the target rank or the field size alone.

\begin{example}[Small-field boundary]
The binary square lies exactly on the field-size boundary. Here $q=2$, $L=2$, $t^*=1$, and the two view kernels are the coordinate axes in $\mathbb F_2^2$. Theorem~\ref{thm:field-size-exactness} gives a kernel-avoiding line; concretely, the diagonal line $R=\operatorname{span}\{(1,1)\}$ avoids both kernels, so Theorem~\ref{thm:kernel-section-exactness} recovers $\Theta(C_4)=\log 2$.

The binary four-coordinate threshold example in Subsection~\ref{sec:theta-theory} is a genuine small-field loose case. It has $q=2$ and $L=6$ views, so it is beyond the $L\le q$ guarantee. With all two-coordinate views on $\mathbb F_2^4$, one has $t^*=2$, and a linear kernel-avoiding section would be a binary $[4,2,3]$ linear code: every nonzero codeword would have Hamming weight at least $3$. Such a code does not exist. Indeed, a two-dimensional binary subspace has three nonzero vectors; if one has weight $4$ and another has weight at least $3$, their sum has weight at most $1$, and if two distinct nonzero vectors both have weight $3$, their sum has weight $2$. Thus the rank bound $\log4$ is not tight in that instance.
\end{example}

\leanmetapending{\LH{AKS19}}
\begin{proposition}[Hamming-endpoint family]\label{prop:hamming-endpoint-family}
Let $A=\mathbb F_q^d$ and let the admissible views be the $d$ singleton coordinates. The induced graph has two words adjacent exactly when they agree in at least one coordinate, equivalently when their Hamming distance is at most $d-1$. The rank minimum is $t^*=1$. The repetition-code diagonal
\[
R=\{(a,a,\ldots,a):a\in\mathbb F_q\}
\]
avoids every singleton-coordinate kernel, since a nonzero diagonal vector has no zero coordinate. Theorem~\ref{thm:kernel-section-exactness} gives
\[
\Theta(G)=\log q.
\]
\end{proposition}

\begin{proof}
For singleton views, two states are confusable exactly when one coordinate agrees. That condition is the Hamming-distance condition stated above. Each singleton projection has rank $1$, so $t^*=1$. The diagonal $R$ is one-dimensional. If a diagonal vector lies in the kernel of the $i$th singleton projection, then its $i$th coordinate is zero, and hence the scalar defining the diagonal vector is zero. Thus $R$ meets every singleton kernel only at $0$, and Theorem~\ref{thm:kernel-section-exactness} gives the displayed capacity.
\end{proof}

All singleton views give the no-gap endpoint of the Hamming-threshold examples. The repetition-code section works for every $d$ and $q$, including regimes with $d=L>q$ where the field-size theorem does not apply.

\leanmetapending{\LHrng{MFT}{102}{106}}
\begin{proposition}[All-\(k\)-view Hamming-threshold specialization]\label{prop:all-k-view-hamming-threshold}
Let the full state space be $\mathbb F_q^d$, fix $0\le k\le d$, and let $\mathcal V_k$ consist of all $k$-coordinate views. Then two distinct words are adjacent exactly when they agree in at least $k$ coordinates, equivalently when their Hamming distance is at most $d-k$. Hence independent sets in $G_{\mathcal V_k}$ are precisely $q$-ary codes of length $d$ and minimum distance at least $d-k+1$.

Every such independent set has size at most $q^k$, and therefore $\alpha(G_{\mathcal V_k})\le q^k$. The one-shot bound is tight exactly when there is a code of size $q^k$ with this distance condition.
\end{proposition}

\begin{proof}
An admissible $k$-view fails to separate two words $x,y$ exactly when the view is contained in $\operatorname{Agr}(x,y)$. Such a view exists exactly when $|\operatorname{Agr}(x,y)|\ge k$, which is equivalent to $d_H(x,y)\le d-k$. Taking the negation for every distinct pair in a vertex set gives the stated code condition for independence.

For the size bound, fix any $k$ coordinates. The restriction map from an independent code to those coordinates is injective: two codewords with the same restriction would agree in at least $k$ coordinates, hence would be adjacent. The image has at most $q^k$ words.

The final equivalence follows because the maximum independent-set size is attained: a size-$q^k$ independent code gives $\alpha(G_{\mathcal V_k})=q^k$ by the upper bound, and equality of $\alpha$ supplies such a code.
\end{proof}

On the full cube, every $k$-coordinate projection has rank $k$, so Theorem~\ref{thm:matroid-capacity-bounds} gives
\[
\Theta(G_{\mathcal V_k})\le k\log q.
\]
In graph-capacity language, the Singleton bound reads: a $q$-ary code of length $d$ and distance $d-k+1$ has $|C|\le q^{d-(d-k+1)+1}=q^k$~\cite{macwilliams1977theory}. In standard coding terminology, the one-shot equality case is an MDS code of size $q^k$ with distance $d-k+1$; repetition then gives $\Theta(G_{\mathcal V_k})=k\log q$. The coordinate-view statement recovers this code-distance criterion and uses it to certify graph capacity from the presentation.

\paragraph{Reed--Solomon exact all-\(k\) family}
\leanmetapending{\LH{AKS33}}
The all-$k$ specialization contains a classical infinite exact family~\cite{reed1960polynomial}. Let $1\le k\le d\le q$, and let $\mathcal V_k$ be the family of all $k$-coordinate views on the full state space $\mathbb F_q^d$. Then
\[
\alpha(G_{\mathcal V_k})=q^k,
\qquad
\Theta(G_{\mathcal V_k})=k\log q.
\]

Choose distinct evaluation points $a_1,\dots,a_d\in\mathbb F_q$. For every polynomial $f\in\mathbb F_q[X]$ of degree $<k$, take the word
\[
(f(a_1),\dots,f(a_d))\in\mathbb F_q^d.
\]
These are the length-$d$ Reed--Solomon codewords~\cite{reed1960polynomial}. There are $q^k$ such words, one for each coefficient vector of a degree-$<k$ polynomial. If $f\ne g$, then $f-g$ is a nonzero polynomial of degree $<k$, so it has at most $k-1$ roots. Hence the two codewords agree in at most $k-1$ coordinates and are independent in $G_{\mathcal V_k}$ by Proposition~\ref{prop:all-k-view-hamming-threshold}. The same proposition gives $\alpha(G_{\mathcal V_k})\le q^k$, so equality holds. Repetition of the one-shot code gives the lower bound $\Theta(G_{\mathcal V_k})\ge k\log q$, and the rank upper bound above gives the reverse inequality.

Here the number of admissible views is the binomial coefficient ``$d$ choose $k$.'' The Reed--Solomon family gives exact instances beyond the field-size criterion whenever that coefficient is larger than $q$ while $d\le q$.

\paragraph{Haemers-style certificate relationship}
Each admissible view also gives a feasible Haemers-style minrank certificate~\cite{haemers1978upper}. For a graph $G$ and field $\mathbb K$, $\operatorname{minrank}_{\mathbb K}(G)$ is the minimum rank over matrices indexed by $V(G)$ with nonzero diagonal and zero entries on every nonedge. Fix an admissible view $S$ and index the materialized graph by affine states $x\in A$. Define the projection-equality matrix
\[
H^{(S)}_{x,y}=
\begin{cases}
1,&\pi_S(x)=\pi_S(y),\\
0,&\pi_S(x)\ne\pi_S(y).
\end{cases}
\]
The matrix has nonzero diagonal and vanishes on every nonedge of the full confusability graph: if $x$ and $y$ are nonadjacent, no admissible view identifies them, so in particular $\pi_S(x)\ne\pi_S(y)$. After ordering states by their realized $S$-projection, $H^{(S)}$ is block diagonal with one all-ones block for each realized projection value. Therefore
\[
\operatorname{rank} H^{(S)}=|\pi_S(A)|=q^{t(S)}.
\]
Consequently, for the materialized graph and any field $\mathbb K$ over which the displayed $0$--$1$ matrix is interpreted,
\[
\operatorname{minrank}_{\mathbb K}(G)\le \min_{S\in\mathcal V} q^{t(S)}.
\]

\leanmetapending{\LH{MFT95}, \LHrng{AKS}{1}{5}}
\begin{corollary}[Haemers optimality on the kernel-section class]\label{cor:haemers-optimal-kernel-section}
Assume the hypotheses of Theorem~\ref{thm:kernel-section-exactness}, and let $S$ be a rank-minimizing view. Then the projection-equality matrix $H^{(S)}$ is an optimal Haemers minrank witness:
\[
\operatorname{minrank}_{\mathbb K}(G)=q^{t^*}
\]
over every coefficient field $\mathbb K$ used for Haemers minrank.
\end{corollary}

\begin{proof}
Haemers' bound gives $\Theta(G)\le\log\operatorname{minrank}_{\mathbb K}(G)$. Corollary~\ref{cor:kernel-section-no-gap} gives $\Theta(G)=t^*\log q=\log q^{t^*}$, so $\operatorname{minrank}_{\mathbb K}(G)\ge q^{t^*}$. The feasible matrix $H^{(S)}$ has rank $q^{t(S)}=q^{t^*}$, giving the reverse inequality.
\end{proof}

The matroid certificate is an explicit feasible point for the Haemers minrank program, and it is optimal on the kernel-section exactness class. Outside that class, optimal minrank can improve the numerical bound after materialization; the coordinate-derived certificate supplies a polynomial-time feasible point before the vertex set and nonedge constraints have been materialized.

\paragraph{Refinements and exactness}
The base certificate compares the full edge union with each single-view fiber graph and takes the best rank. Theorems~\ref{thm:kernel-section-exactness} and~\ref{thm:field-size-exactness} identify a rank-shaped equality regime where the same number is the exact capacity.

\leanmetapending{\LHrng{MFT}{92}{94}}
\begin{proposition}[Subfamily certificate hierarchy]\label{prop:subfamily-certificate-hierarchy}
For a subfamily $\mathcal W\subseteq\mathcal V$, let $G_{\mathcal W}$ be the graph generated by the views in $\mathcal W$. If
\[
\mathcal W\subseteq\mathcal W'\subseteq\mathcal V,
\]
then
\[
\Theta(G_{\mathcal V})\le \Theta(G_{\mathcal W'})\le \Theta(G_{\mathcal W}).
\]
Consequently, any upper bound on $\Theta(G_{\mathcal W})$ is an upper certificate for the full graph.

More generally, let $B(\mathcal W)$ be any valid computable upper bound on $\Theta(G_{\mathcal W})$, and define the level-$k$ subfamily certificate
\[
U_k^B(\mathcal V):=
\min_{\substack{\emptyset\ne\mathcal W\subseteq\mathcal V\\|\mathcal W|\le k}} B(\mathcal W)
\qquad(1\le k\le L).
\]
Then $\Theta(G_{\mathcal V})\le U_k^B(\mathcal V)$ for every $k$, and $U_{k+1}^B(\mathcal V)\le U_k^B(\mathcal V)$ whenever both sides are defined. Taking $B(\mathcal W)=\Theta(G_{\mathcal W})$ gives the exact subfamily hierarchy, which satisfies $U_L^\Theta(\mathcal V)=\Theta(G_{\mathcal V})$. For fixed $k$, scanning all subfamilies of size at most $k$ costs $O(L^k)$ selected subgraphs before the cost of evaluating $B$ on each selected subgraph.
\end{proposition}

\begin{proof}
Adding admissible views only adds confusability edges. Hence $G_{\mathcal W}$ is an edge subgraph of $G_{\mathcal W'}$, and $G_{\mathcal W'}$ is an edge subgraph of $G_{\mathcal V}$. Independence numbers, and hence normalized strong-power rates, are monotone decreasing under edge addition. The displayed inequalities follow. The fixed-$k$ enumeration statement is the direct count of $\sum_{j\le k}\binom Lj$ selected subfamilies.

The level-$k$ certificate statements follow by minimizing valid upper bounds over nested finite search spaces. The search space for level $k+1$ contains the search space for level $k$, so the minimum cannot increase. At level $L$, the full view family itself is an admissible selected subfamily; the exact hierarchy therefore contains both $\Theta(G_{\mathcal V})$ as an upper candidate and only quantities at least $\Theta(G_{\mathcal V})$ by the first paragraph.
\end{proof}

\paragraph{Concrete evaluators for \(B(\mathcal W)\)}
The hierarchy becomes algorithmic once the evaluator $B$ is instantiated. The single-view rank evaluator is
\[
B_{\mathrm{rank}}(\mathcal W):=\min_{S\in\mathcal W} t(S)\log q,
\]
the first-level matroid certificate applied to the selected subfamily. A cluster/local-composition evaluator uses the equality criteria of Section~\ref{sec:equality}: when the selected full-tuple view family is meet-witnessed, or when an affine selected family satisfies the kernel-sum criterion, $B(\mathcal W)$ is the exact component-counting capacity. A quotient evaluator applies when the selected views depend on a combined projection of small realized rank $m$. Fibers of the combined projection are cliques, and adjacency between fibers is constant, so $G_{\mathcal W}$ is a complete blow-up of a graph on at most $q^m$ quotient states; exact graph algorithms or semidefinite bounds can be run on that quotient rather than on the original $q^r$ states. In symmetric Hamming cases, coding bounds and association-scheme theta calculations supply closed-form evaluators; Proposition~\ref{prop:all-k-view-hamming-threshold} and Example~\ref{ex:subfamily-hierarchy-separation} are the two calibrating instances used below.

\leanmetapending{\LHrng{MFT}{92}{95}, \LHrng{MFT}{102}{106}}
\begin{example}[Strict subfamily-hierarchy separation]\label{ex:subfamily-hierarchy-separation}
Let the state space be $\mathbb F_2^4$ and let $\mathcal V$ be the six two-coordinate views. The first-level rank certificate inspects one two-coordinate view at a time and gives $\log4$. Every proper subfamily $\mathcal W\subsetneq\mathcal V$ has $\alpha(G_{\mathcal W})=4$, while the single-view rank bound gives $\Theta(G_{\mathcal W})\le\log4$; hence $\Theta(G_{\mathcal W})=\log4$ for every proper selected subfamily. The full six-view graph has adjacency determined by Hamming distance at most $2$, has $\alpha(G_{\mathcal V})=2$, and has the Walsh-character theta certificate $\Theta(G_{\mathcal V})\le\log(8/3)<\log4$. Therefore the exact subfamily hierarchy satisfies
\[
U_5^\Theta(\mathcal V)=\log4,
\qquad
U_6^\Theta(\mathcal V)=\Theta(G_{\mathcal V})\le\log(8/3)<\log4.
\]
The first strict improvement occurs only when the full symmetric six-view subfamily is selected. The supplementary checker \texttt{subfamily\_threshold\_levels} verifies the one-shot independence-number assertions for all subfamilies.
\end{example}

Affine structure gives a polynomial-time upper certificate in general and a polynomial-time exact-capacity formula on the kernel-section classes above. Subfamily analysis adds graph-level upper certificates when the selected views have enough quotient structure or symmetry to be evaluated directly; Example~\ref{ex:subfamily-hierarchy-separation} shows that the resulting hierarchy can be strictly stronger than the first-level single-view certificate.

\leanmetapending{\LHrng{AFM}{18}{19}}
\begin{proposition}[Matroid direct sum under block composition]\label{prop:matroid-direct-sum-block}
For affine families $A_1=a_1+V_1$ over $\mathbb F_q^{d_1}$ and $A_2=a_2+V_2$ over $\mathbb F_q^{d_2}$, the representable matroid on $[d_1+d_2]$ induced by the product family $A_1\times A_2$ is the direct sum of the component matroids. A coordinate set $I\subseteq [d_1+d_2]$ is independent in the product matroid if and only if $I\cap [d_1]$ is independent in the first component matroid and $\{j\in[d_2]: d_1+j\in I\}$ is independent in the second component matroid.
\end{proposition}

\begin{proof}
The direction space of the product family is $V_1\times V_2$. For coordinate subsets $S_1\subseteq[d_1]$ and $S_2\subseteq[d_2]$, write $S_1\sqcup S_2:=S_1\cup\{d_1+j:j\in S_2\}\subseteq[d_1+d_2]$. The restricted coordinate projection $\pi_{S_1\sqcup S_2}|_{V_1\times V_2}$ decomposes as $\pi_{S_1}|_{V_1}\times \pi_{S_2}|_{V_2}$ because $V_1$ and $V_2$ occupy disjoint coordinate blocks. Linear independence of coordinate functionals on $V_1\times V_2$ therefore decomposes into independent conditions on each block.
\end{proof}

\leanmetapending{\LH{AFM17}}
\begin{corollary}[Rank additivity]\label{cor:rank-additivity-block}
For coordinate sets $S_1\subseteq [d_1]$ and $S_2\subseteq [d_2]$, write $S_1\sqcup S_2:=S_1\cup\{d_1+j:j\in S_2\}\subseteq[d_1+d_2]$. If the component view families are $\mathcal V_1$ and $\mathcal V_2$, write
\[
\mathcal V_{\mathrm{product}}:=\{S_1\sqcup S_2:S_1\in\mathcal V_1,\ S_2\in\mathcal V_2\}.
\]
Then
\[
t_{\mathrm{product}}(S_1\sqcup S_2) = t_1(S_1) + t_2(S_2).
\]
Consequently, the matroid capacity bound for $\mathcal V_{\mathrm{product}}$ decomposes additively under block composition:
\[
\Theta(G_{\mathrm{product}}) \le \min_{S_1\in\mathcal V_1,\;S_2\in\mathcal V_2} \bigl(t_1(S_1) + t_2(S_2)\bigr) \log q.
\]
\end{corollary}

\begin{proof}
The restricted product projection decomposes as the direct sum of the two component restricted projections, so
\[
\operatorname{rank}(\pi_{S_1}|_{V_1}\times \pi_{S_2}|_{V_2})
=
\operatorname{rank}(\pi_{S_1}|_{V_1})
+
\operatorname{rank}(\pi_{S_2}|_{V_2}).
\]
Apply Theorem~\ref{thm:matroid-capacity-bounds} to the product family.
\end{proof}

\leanmetapending{\LHrng{AFM}{17}{19}, \LHrng{AKS}{1}{5}, \LHrng{MFT}{25}{26}}
\begin{corollary}[Direct-sum closure of kernel-section exactness]\label{cor:kernel-section-direct-sum}
Let two affine coordinate-view systems over the same field satisfy the kernel-section hypothesis with rank minima $t_1^*$ and $t_2^*$. For the block product view family
\[
\mathcal V_{\mathrm{product}}=\{S_1\sqcup S_2:S_1\in\mathcal V_1,\ S_2\in\mathcal V_2\},
\]
the product system also satisfies the kernel-section hypothesis, and
\[
\Theta(G_1\boxtimes G_2)=\Theta(G_1)+\Theta(G_2)
=(t_1^*+t_2^*)\log q .
\]
Equivalently, the classical non-logarithmic capacity multiplies on the kernel-section class.
\end{corollary}

\begin{proof}
Let $R_i$ be kernel-avoiding subspaces of dimensions $t_i^*$ in the two components. The product section $R_1\oplus R_2$ has dimension $t_1^*+t_2^*$. For a product view $S_1\sqcup S_2$, its kernel is the direct sum of the two component kernels. If $(r_1,r_2)\in R_1\oplus R_2$ lies in that kernel, then $r_i$ lies in the $S_i$-kernel in the $i$th component; kernel avoidance gives $r_1=r_2=0$. Hence the product section is kernel-avoiding. Corollary~\ref{cor:rank-additivity-block} identifies the product rank minimum as $t_1^*+t_2^*$, and Proposition~\ref{prop:block-composition-strong-product} identifies the product confusability graph with $G_1\boxtimes G_2$. Theorem~\ref{thm:kernel-section-exactness} gives the displayed equality.
\end{proof}

The matroid rank $t(S)$ is the informational dimension captured by coordinates in $S$, directly analogous to the rank function in representable matroids studied in combinatorial optimization~\cite{oxley2011matroid,welsh1976matroid,recski1989matroid} and coding theory~\cite{cover2006elements}. The equality $t(S)=r$ holds exactly when $S$ determines the entire state, and any admissible view containing a basis removes confusability entirely. The affine specialization makes the rank certificate polynomial-time computable. Generalizations to non-affine state families are open.

\section{Strong Powers and Asymptotic Capacity}\label{sec:capacity-theory}

The rank certificate of Section~\ref{sec:affine} bounds the asymptotic zero-error invariant obtained from block composition. Proposition~\ref{prop:block-composition-strong-product} identifies block composition with strong graph product, so the standard Shannon-capacity framework applies to the induced graph.

\begin{definition}[Shannon Capacity]\label{def:shannon-capacity}
For a finite graph $G$, its \emph{Shannon capacity} is
\[
\Theta(G) := \sup_{n\ge 1} \frac{1}{n}\log \alpha\bigl(G^{\boxtimes n}\bigr),
\]
where $\alpha$ denotes independence number and $G^{\boxtimes n}$ is the $n$-fold strong power.
\end{definition}

\leanmetapending{\LHrng{MFT}{43}{49}, \LH{MFT56}}
Here \(\Theta\) is logarithmic: it is the logarithm of the usual multiplicative capacity \(\sup_n \alpha(G^{\boxtimes n})^{1/n}\). For an induced confusability graph, \(\Theta(G_{\mathrm{conf}})\) is the asymptotic zero-error rate of the block-composed system. If \(\alpha_n=\alpha(G^{\boxtimes n})\), then
\[
\alpha_{m+n}\ge\alpha_m\alpha_n,
\]
because independent sets multiply under strong product. Fekete's lemma~\cite{fekete1923uber} gives convergence of the normalized block rates
\[
\frac{\log \alpha_n}{n}
\]
to the displayed supremum.

\subsection{Rank Certificates and Lov\'asz-\(\vartheta\)}\label{sec:theta-theory}

Lov\'asz-$\vartheta$ is the polynomial-time graph-level comparator for the affine rank certificate~\cite{lovasz1979shannon}. The algorithmic distinction is the input. Gaussian elimination computes the rank certificate from the \(r\times d\) affine presentation and view list before the \(q^r\)-vertex graph is enumerated. Lov\'asz-$\vartheta$ solves an SDP on the materialized graph, or on a symmetry-reduced quotient when enough graph regularity is available.

\begin{center}
\begin{tabular}{p{0.26\linewidth}p{0.31\linewidth}p{0.32\linewidth}}
\toprule
Regime & Rank certificate & Lov\'asz-\(\vartheta\) \\
\midrule
Kernel-section or field-size exactness & Exact capacity from Gaussian elimination and a section witness & Also upper-bounds capacity; equality requires theta exactness \\
Large affine presentations without graph symmetry & Polynomial in \(r,d,L,\log q\), before graph materialization & Requires the \(q^r\)-vertex graph or a quotient not supplied by the presentation \\
Small or highly symmetric materialized graphs & May miss global regularity across several view fibers & SDP or spectral formulas can be strictly tighter \\
Cluster/intersection-closed systems & Component counting gives the exact value & \(\vartheta\) is exact because \(\alpha=\chi(\overline G)\) \\
\bottomrule
\end{tabular}
\end{center}

Two dominance statements are immediate. On the kernel-section class, the rank certificate equals \(\Theta(G)\), so no valid upper bound, including \(\vartheta_\infty\), can be strictly smaller. On materialized graphs with exploitable symmetry, \(\vartheta\) can use global adjacency relations that a single-view rank certificate ignores; the four-coordinate threshold graph below is the smallest example used here. A presentation-level criterion predicting when theta improves the rank certificate would have to recognize those global view-fiber interactions before the graph or a symmetry quotient is built.

For a finite graph $G$, let $\vartheta(G)$ denote the classical theta number and set
\[
\vartheta_\infty(G):=\inf_{n\ge 1}\frac{1}{n}\log\vartheta(G^{\boxtimes n}).
\]
Multiplicativity of $\vartheta$ under strong product makes this $\log\vartheta(G)$, and the standard theta upper law gives
\[
\Theta(G)\le \vartheta_\infty(G).
\]
Likewise, complement colorings give the elementary bound $\Theta(G)\le\log\chi(\overline G)$.

On the kernel-section exactness class, the rank certificate reaches capacity:
\[
\Theta(G)=\min_{S\in\mathcal V}t(S)\log q.
\]
Since $\vartheta_\infty$ is an upper bound on $\Theta$, the rank certificate is at least as tight as $\vartheta_\infty$ on that class; equality holds when $\vartheta$ is also exact. Outside the kernel-section class there is no dominance relation forced by the presentation. The rank certificate remains polynomial in the succinct affine input, while $\vartheta$ can exploit graph-level regularity after materialization or symmetry reduction.

For the binary square, both bounds give the exact value $\log 2$. The bounds separate on a small Hamming-threshold example. Let the state space be $\mathbb F_2^4$ and let the admissible views be all six two-coordinate projections. Each view has coordinate rank $2$, so Theorem~\ref{thm:matroid-capacity-bounds} gives
\[
\Theta(G)\le \log 4.
\]
The materialized graph has the $16$ binary words as vertices, with two words adjacent exactly when their Hamming distance is at most $2$. It is $10$-regular. Its adjacency matrix is diagonalized by Walsh characters, with eigenvalues $10,2,-2,-2,2$ by Hamming weight, so the least eigenvalue is $-2$. Lov\'asz's spectral theta certificate for this regular graph~\cite{lovasz1979shannon} gives
\[
\vartheta_\infty(G)\le \log\!\left(\frac{16\cdot 2}{10+2}\right)=\log(8/3),
\]
strictly below the matroid rank bound. The $\vartheta$ certificate detects global regularity in the union of view-fiber graphs. Single-view rank sees only the individual projection fibers and remains the pre-materialization certificate.

For scale, an exact bitset branch-and-bound computation gives
\[
\alpha(G)=2,\qquad \alpha(G^{\boxtimes 2})=5,
\]
where independence in the strong square is computed as maximum clique in the complement graph on $16^2$ vertices. The supplementary checker \texttt{theta\_two\_coordinate\_example} constructs the complement bitsets, runs an exact branch-and-bound clique search, and verifies the Walsh-character spectrum used above. Hence
\[
\log\sqrt{5}\le \Theta(G)\le \log(8/3)<\log 4.
\]
In base $2$, the spectral gap to this finite-power lower estimate is about $0.173$, while the rank-certificate gap is about $0.839$.

\leanmetapending{\LHrng{MFT}{96}{101}, \LH{MFT108}}
For any logarithmic upper certificate \(U\ge\Theta(G)\), define
\[
\Gamma_U(G):=U-\Theta(G).
\]
The gap is nonnegative and vanishes exactly when the certificate is exact. Fekete convergence gives
\[
\Gamma_U(G)=\lim_{n\to\infty}
\left(
U-\frac{1}{n}\log\alpha(G^{\boxtimes n})
\right).
\]
Equivalently,
\[
\Gamma_U(G)=
\lim_{n\to\infty}
\frac{nU-\log\alpha(G^{\boxtimes n})}{n}.
\]
If a product instance \(P\) has additive upper value \(U_P=U_1+U_2\) and capacity at least the sum of the factor capacities, then
\[
\Gamma_{U_P}(P)\le\Gamma_{U_1}(G_1)+\Gamma_{U_2}(G_2).
\]
When \(\Theta(P)=\Theta(G_1)+\Theta(G_2)\), the inequality is equality:
\[
\Gamma_{U_P}(P)=\Gamma_{U_1}(G_1)+\Gamma_{U_2}(G_2).
\]

\paragraph{Approximation and additive gaps}
The exactness and hardness results leave a separate approximation problem. Direct Grassmannian enumeration is fixed-parameter in the ambient affine rank \(r\), and the field-size and kernel-section theorems give zero additive gap on their positive regimes. Outside those regimes, the rank certificate supplies a polynomial-time upper value \(U\), the greedy packing bound gives a coarse lower value, and the subfamily hierarchy can lower the upper value when selected subfamilies have small quotients, cluster structure, or symmetry. No FPTAS for \(N_{t^*}\) or additive-\(\varepsilon\) approximation theorem for \(\Theta\) follows from these ingredients. The normalized gap \(\Gamma_U\) records the resulting approximation loss against the strong-power rates.

On Shannon-gap geometry, the rank certificate is an upper-bound algorithm rather than an exactness mechanism. Corollary~\ref{cor:kernel-section-no-gap} identifies the exact kernel-section regime as $\Theta(G)=\log\alpha(G)$. Thus \(C_5\) lies outside the exact regime, since \(\Theta(C_5)=\frac12\log5>\log2=\log\alpha(C_5)\). The same obstruction appears for \(C_7\) at the finite-power level: \(\alpha(C_7^{\boxtimes2})=10>3^2=\alpha(C_7)^2\). The supplementary checker \texttt{odd\_cycle\_boundary} verifies these strong-square values.

\leanmetapending{\LH{MFT107}}
In the full-tuple model, the prime-order cycles $C_5$ and $C_7$ are not coordinate-view graphs at all. A full product with a prime number of states has only one nonunary coordinate. With one effective coordinate, an empty admissible view makes all distinct states confusable, while every nonempty view separates all distinct states; the induced graph is complete or edgeless. Restricted affine representations are separate instances. Kernel-section equality forces the no-gap condition above.

The Hamming-threshold example above is the corresponding coordinate-view witness on the loose side: finite powers already beat \(\log\alpha(G)\), while the matroid rank certificate remains a pre-materialization upper bound.

Sharpness of either upper bound for general coordinate-view classes is open. The transitive intersection-closed systems of Theorem~\ref{thm:intersection-closure} are a subclass where exact capacity is the component-counting formula. The one-shot sandwich also gives exact values: if $\alpha(G)=n$ and $\overline G$ is $n$-colorable, then $\Theta(G)=\log n$. In rank-shaped instances with matching witnesses of size $q^t$, the sandwich reads $\Theta(G)=t\log q$ and proves exactness for the binary square and nested-view cluster examples.

\subsection{Cluster Collapse and Restricted-State Transitivity}\label{sec:equality}

The full-tuple cluster route is separate from affine kernel-section exactness. Kernel sections match a matroid-rank upper bound with a large independent set. Cluster collapse instead asks when the whole confusability relation is transitive. In the full coordinate-view model over nontrivial alphabets, transitivity is exactly intersection closure of the generated agreement family.

For a view family \(\mathcal V\), write
\[
\mathcal U_{\mathcal V}=\{R\subseteq [d]: \exists S\in\mathcal V,\ S\subseteq R\}.
\]
Meet-witnessing is the finite-list form of intersection closure: every pair of admissible views has an admissible subview contained in its intersection. A two-step path \(x\sim y\sim z\) records two agreement sets, and the coordinates forced to agree between \(x\) and \(z\) are their intersection. The missing admissible subview in that intersection is exactly the obstruction to transitivity.

\leanmetapending{\LHrng{MFT}{139}{150}}
\begin{theorem}[Cluster/non-transitive dichotomy]\label{thm:intersection-closure}
For deterministic coordinate-subset views on the full tuple space over a nontrivial finite alphabet, meaning \( |A_i|\ge 2 \) for each coordinate,
\[
\begin{aligned}
G_{\mathrm{conf}}\text{ is transitive}
&\Longleftrightarrow
\mathcal U_{\mathcal V}\text{ is intersection-closed}\\
&\Longleftrightarrow
\mathcal V\text{ is meet-witnessed}.
\end{aligned}
\]
If these equivalent conditions hold, \(G_{\mathcal V}\) is a cluster graph and
\[
\Theta(G_{\mathcal V})=\vartheta_\infty(G_{\mathcal V})=\log |\Pi_{\mathrm{cc}}|.
\]
If they fail, there exist distinct states \(x,y,z\) with \(x\sim y\), \(y\sim z\), and \(x\not\sim z\).
\end{theorem}

\begin{proof}
Meet-witnessing and intersection closure are equivalent for the upward family. If \(R,R'\in\mathcal U_{\mathcal V}\), choose admissible \(S\subseteq R\) and \(S'\subseteq R'\). A meet witness \(W\subseteq S\cap S'\) gives \(R\cap R'\in\mathcal U_{\mathcal V}\). Conversely, applying intersection closure to two admissible views \(S,S'\) shows that \(S\cap S'\in\mathcal U_{\mathcal V}\), hence some admissible \(W\subseteq S\cap S'\) exists.

Intersection closure implies transitivity because \(x\sim y\) and \(y\sim z\) give \(\operatorname{Agr}(x,y),\operatorname{Agr}(y,z)\in\mathcal U_{\mathcal V}\); their intersection is contained in \(\operatorname{Agr}(x,z)\), and upward closure gives \(x\sim z\).

For the converse, choose \(R,R'\in\mathcal U_{\mathcal V}\) with \(R\cap R'\notin\mathcal U_{\mathcal V}\). Partition
\[
A=R\cap R',\quad B=R\setminus R',\quad C=R'\setminus R,\quad D=[d]\setminus(R\cup R').
\]
Using two symbols \(0,1\), define states by
\[
\begin{array}{c|cccc}
 & A & B & C & D\\ \hline
x & 0 & 0 & 1 & 0\\
y & 0 & 0 & 0 & 0\\
z & 0 & 1 & 0 & 1 .
\end{array}
\]
Then \(\operatorname{Agr}(x,y)\supseteq R\), \(\operatorname{Agr}(y,z)=R'\), and \(\operatorname{Agr}(x,z)=R\cap R'\notin\mathcal U_{\mathcal V}\). Thus \(x\sim y\), \(y\sim z\), and \(x\not\sim z\). The sets \(B\) and \(C\) are nonempty, since otherwise \(R\cap R'\) would equal one of \(R,R'\), so the three states are distinct.

In the transitive case, a finite graph with transitive adjacency is a disjoint union of cliques. One representative from each connected component is independent, and the component partition colors the complement. The standard sandwich \(\alpha(G)\le \vartheta(G)\le\chi(\overline G)\), together with multiplicativity of \(\vartheta\), gives the displayed capacity formula.
\end{proof}

\leanmetapending{\LHrng{CMP}{1}{5}}
The nontrivial-alphabet hypothesis is needed only for the obstruction triple. Unary coordinates are constant and can be deleted before applying the theorem. For an explicit bit-vector view list, meet-witnessing is checked by enumerating the \(L^2\) ordered view pairs and scanning the \(L\) candidate witness views for containment in each intersection. The direct implementation runs in \(O(L^3d)\) time, after an optional \(O(L^2d)\) cofinal-antichain normalization.

\leanmetapending{\LHrng{MFT}{89}{91}}
The fiber-coherent subclass gives the same cluster formula in transcript language: if equality in one admissible view forces equality of the full observation transcript, then the connected components are realized transcript fibers and \(\Theta(G_{\mathcal V})=\vartheta_\infty(G_{\mathcal V})=\log |\mathcal Y_{\mathrm{real}}|\).

\paragraph{Restricted coordinate-subset variants}
Restricted state sets require a local criterion on realized triples. For \(X\subseteq A_1\times\cdots\times A_d\), transitivity is equivalent to
\[
\begin{aligned}
&\forall x,y,z\in X,\ \forall S,S'\in\mathcal V,\\
&\qquad
\pi_S(x)=\pi_S(y)\ \wedge\ \pi_{S'}(y)=\pi_{S'}(z)\\
&\qquad\Longrightarrow
\left(x=z\ \vee\ \exists W\in\mathcal V,\ \pi_W(x)=\pi_W(z)\right).
\end{aligned}
\]
For explicit \(X\), this condition is checked by enumerating realized triples and view pairs. For succinct restricted families such as affine subspaces, direct enumeration is polynomial in \(|X|\) and can be exponential in the affine rank. The \(O(L^3d)\) meet-witnessing search above is therefore the polynomial-time full-tuple criterion, not a restricted-affine classifier.

\leanmetapending{\LHrng{AKS}{22}{31}}
\begin{proposition}[Affine kernel-sum transitivity criterion]\label{prop:affine-kernel-sum-transitivity}
Let \(X=A=a_0+V\) be an affine restricted state family, and write \(K_S=\ker(\pi_S|_V)\) for each admissible view \(S\). The affine confusability relation is transitive if and only if
\[
\forall S,S'\in\mathcal V,\qquad
K_S+K_{S'}\subseteq \bigcup_{W\in\mathcal V}K_W.
\]
If the view list is nonempty and \(q\ge L=|\mathcal V|\), this condition is equivalent to
\[
\forall S,S'\in\mathcal V,\quad
\exists W\in\mathcal V\text{ such that }K_S+K_{S'}\subseteq K_W.
\]
In the regime \(q\ge L\), affine transitivity is decidable by linear-algebraic containment tests on kernel sums.
\end{proposition}

\begin{proof}
The local-composition condition has a direction-space form. Given \(x,y,z\in A\), the equalities \(\pi_S(x)=\pi_S(y)\) and \(\pi_{S'}(y)=\pi_{S'}(z)\) mean
\[
u:=x-y\in K_S,\qquad v:=y-z\in K_{S'}.
\]
The endpoint difference is \(x-z=u+v\). Thus affine transitivity is equivalent to every vector in \(K_S+K_{S'}\) lying in some admissible endpoint kernel \(K_W\). Conversely, every vector in \(K_S+K_{S'}\) arises as such an endpoint difference by choosing any \(y\in A\) and setting \(x=y+u\), \(z=y-v\).

The witness condition implies the union containment immediately. Conversely, assume \(q\ge L\) and the union containment. If \(K_S+K_{S'}\) is covered by the \(L\) subspaces \((K_S+K_{S'})\cap K_W\), Lemma~\ref{lem:finite-union-subspace-avoidance} forces one intersection to equal \(K_S+K_{S'}\). Hence some \(W\) satisfies \(K_S+K_{S'}\subseteq K_W\).
\end{proof}

Product subcubes reduce to full tuple spaces after deleting frozen coordinates, so meet-witnessing applies to the induced view family. Affine subspaces and product codes are subtler because missing triples are governed by realized codeword differences. The family \(C_q=\{(a,b,a+b):a,b\in\mathbb F_q\}\) with singleton views separates the abstract matroid from realized transitivity: over \(\mathbb F_2\) the graph is complete, while over \(\mathbb F_3\) the states \((0,0,0),(0,1,1),(1,1,2)\) form a non-transitive triple. The same view antichain and abstract coordinate matroid therefore do not determine restricted-state transitivity.

\section{Related Work}\label{sec:related}

\subsection{Algorithmic Tractability, Matroids, and Graph Capacity}\label{sec:related-algorithmic}

Representable matroids provide polynomial-time rank oracles through linear algebra and form a standard tractable class in combinatorial optimization~\cite{oxley2011matroid,welsh1976matroid,recski1989matroid,schrijver2003combinatorial}. The affine specialization uses exactly that tractable structure: an explicit generator matrix for the state family turns each coordinate-view certificate into a restricted matrix-rank computation. The resulting rank bound is a linear-algebraic upper certificate for graph quantities that are hard in general, and it becomes an exact Shannon-capacity formula on the kernel-section and field-size subclasses of Theorems~\ref{thm:kernel-section-exactness}--\ref{thm:field-size-exactness}. The certificate is computed on the coordinate presentation before enumerating the potentially exponential state space or the induced confusability graph. In the dense finite-field model, the rank computations are polynomial in the bit-size of the generator matrix and the explicit view list; for fixed field size, the operation count is the Gaussian-elimination bound stated in Section~\ref{sec:affine}.

Subspace arrangements enter in two distinct ways. Rank-one kernel avoidance is the classical complement-count problem for a central arrangement, with the characteristic polynomial evaluated at the field size~\cite{orlik1992arrangements}. Higher-rank avoidance replaces vector complements by Grassmannian complements. The count \(N_t\) is determined by the signed rank profile of the projectivized forbidden-point matroid, not by the kernel-intersection lattice alone. Jamison and Brouwer--Schrijver supply the finite-field blocking bound used for the \(q\ge L\) exactness theorem~\cite{jamison1977covering,brouwer1978blocking}; the capacity and minrank consequences are obtained after that blocking bound is applied to coordinate-view kernels. The Crapo--Rota critical problem supplies the quotient complement question~\cite{crapo1970combinatorial,kung1996critical}; the fixed-shell result concerns physical kernel-intersection profiles that fail to determine that quotient target. The signed-profile theorem gives the finer counted invariant needed below the blocking threshold.

The support-profile criterion is also distinct from ordinary generalized Hamming weights~\cite{wei1991generalized}. Generalized Hamming weights minimize the union support of a subcode. The arrangement-relative profile maximizes a subcode's minimum hit against a prescribed view hypergraph. For \(V=\mathbb F_q^3\) with all two-coordinate views, the first generalized Hamming weight is \(1\), while the arrangement-relative one-dimensional profile is \(2\).

Exact independence-number computation is NP-hard~\cite{karp1972reducibility}, and the computational complexity of exact Shannon capacity is open~\cite{alon2006shannon,deboer2024asymptotic}. Lov\'asz's $\vartheta$ number gives a classical upper bound on Shannon capacity~\cite{lovasz1979shannon}, and the Gr\"otschel--Lov\'asz--Schrijver ellipsoid framework places semidefinite and convex relaxations at the center of algorithmic graph theory~\cite{grotschel1981ellipsoid}. For affine coordinate-view confusability graphs with an explicit linear presentation, the coordinate-matroid certificate is computed by Gaussian elimination rather than semidefinite optimization or graph-level search. The binary square has coincident rank and theta bounds. The binary four-coordinate threshold graph has a graph-level theta certificate that is strictly tighter than the rank bound after materialization.

Recent graph-capacity work reinforces that the surrounding problem remains active. Asymptotic-spectrum methods give new dual and limit viewpoints~\cite{zuiddam2019asymptotic,deboer2024asymptotic}; recent exact-capacity and bound computations include $q$-Kneser and tadpole graph families~\cite{lavi2026advances}; and Mycielski-type constructions continue to illuminate how Shannon capacity and Lov\'asz-type upper bounds behave under graph operations~\cite{csonka2024mycielski}. Quantum and noncommutative variants give a parallel motivation for keeping the confusability-graph interface explicit, from entanglement-assisted zero-error models~\cite{cubitt2014bounds} to covariant quantum confusability graphs~\cite{verdon2024covariant}.

Haemers-type bounds provide additional upper-bound technology for graph capacity~\cite{haemers1978upper}. Minrank theory optimizes matrices for a fixed graph. The affine coordinate-matroid bound constructs a specific feasible matrix from the generator matrix and one coordinate view before graph materialization. Section~\ref{sec:affine} proves the certificate relationship: every admissible affine view $S$ yields a projection-equality matrix of rank $q^{t(S)}$ that is feasible for Haemers minrank of the full materialized graph. Hence
\[
\operatorname{minrank}_{\mathbb K}(G)\le \min_{S\in\mathcal V} q^{t(S)}
\]
over any field $\mathbb K$ in which the displayed $0$--$1$ projection-equality matrix is interpreted. For a fixed graph, optimal Haemers minrank is bounded above by this coordinate-derived feasible matrix. On the kernel-section exactness class, Corollary~\ref{cor:haemers-optimal-kernel-section} identifies the same projection-equality matrix as the optimal Haemers minrank witness.

The four-coordinate dependent example from Section~\ref{sec:affine} makes the distinction concrete. The state family has $2^3$ states, and the view $\{1,4\}$ has coordinate rank $1$ because the first and fourth coordinate functionals coincide on the affine direction space. The matroid computation therefore gives a rank-$2$ Haemers-style certificate and a $\log 2$ capacity upper bound without searching over matrices on the eight graph vertices. In the binary-square and nested-view examples, the affine rank certificate, the complement-chromatic bound, and the exact capacity all give $\log 2$. When optimal Haemers minrank is lower than the coordinate certificate, that improvement is a graph-level refinement after materialization. The presentation-level result is the polynomial-time construction of a useful graph certificate from coordinate data, together with its matroid-rank interpretation.

\subsection{Zero-Error and Side-Information Lineage}\label{sec:related-source-coding}

Shannon's zero-error framework and its graph-theoretic refinements by K\"orner and Lov\'asz provide the zero-error graph lineage~\cite{shannon1956zero,korner1973graphs,lovasz1979shannon}. K\"orner--Orlitsky's zero-error synthesis is also a direct reference point for that lineage~\cite{korner2002zero}. Witsenhausen's zero-error side-information problem studies exact recovery under side information through graph coloring and label budgets~\cite{witsenhausen1976zero}. The one-shot coloring statements belong to the same lineage.

Slepian--Wolf coding and classical entropy converses form a second reference line~\cite{slepian1973noiseless,fano1961transmission,witsenhausen1975conditional,cover2006elements}. Coding-for-computing, functional compression, and graph-entropy or characteristic-graph viewpoints study deterministic decoder targets and coding questions once an underlying graph or function structure has been specified~\cite{orlitsky2001coding,alon2002source,doshi2010functional,korner1973graphs,simonyi2001graph}. In those formalisms, including Orlitsky--Roche characteristic graphs, the graph records which source symbols a decoder target fails to identify; after the graph is formed, the coordinate presentation need not survive. Here the admissible observations themselves are coordinate projections of a finite product state. The generated graph retains the view list's agreement-set invariant from Theorem~\ref{prop:agreement-family-complete-invariant}. Recent zero-error function-compression work continues the broader line of research~\cite{guang2023zero,liu2021leakage}.

Theorem~\ref{prop:agreement-family-complete-invariant} has a different target from the standard characteristic-graph theorems. Witsenhausen's coloring theorem and the K\"orner--Orlitsky zero-error reductions identify the coding problem once the confusability graph is fixed; Orlitsky--Roche characteristic graphs describe which source pairs a decoder target can merge. The coordinate-view theorem classifies when two coordinate-projection presentations induce the same labeled graph. The upward-closed agreement-family statement uses the full tuple space and nontrivial coordinate alphabets to prove the converse direction: every proper coordinate subset is realized as an agreement set of a state pair, so graph equality recovers the entire generated upward family. Recovery of the view-generated antichain from the labeled edge relation is the new presentation-level content.

The deterministic coordinate-view model starts from finite product spaces, with observations obtained by admissible coordinate-subset projections. In the exact full-tuple-space model, the generated confusability relations are precisely those determined by upward-closed families of coordinate-agreement sets. The labeled invariant is the upward-closed agreement family, the minimal presentation is its cofinal antichain, and the affine specialization attaches a representable coordinate matroid before graph materialization. Characteristic-graph formalisms handle arbitrary deterministic functions of the source; the coordinate-projection restriction additionally yields coordinatewise alphabet-permutation automorphisms, a monotone agreement-set graph class, and closure under block composition by strong product. The unlabeled-isomorphism classification and restricted-state analogues give natural next classification problems.

\begin{center}
\small
\begin{tabular}{p{0.24\linewidth}p{0.32\linewidth}p{0.34\linewidth}}
\toprule
Tool or formalism & Preserved structure & Question answered \\
\midrule
Zero-error coloring reductions & A fixed confusability graph & Exact recovery and label budgets after the graph is fixed \\
Orlitsky--Roche characteristic graphs & Merge constraints for a decoder target & Which source pairs a target function can identify \\
Lov\'asz-$\vartheta$ and Haemers minrank & Graph-level feasible relaxations & Upper bounds after graph materialization or symmetry reduction \\
Characteristic polynomials of arrangements & Intersection-lattice complement counts & Rank-one avoidance counts for vectors \\
Generalized Hamming weights & Union supports of linear subcodes & Minimum support size achievable by a subcode \\
Coordinate-view presentations & View list, agreement family, and affine coordinate matroid & Pre-materialization canonical presentation and rank certificates \\
\bottomrule
\end{tabular}
\end{center}

For a fixed confusability graph, the colorability, strong-power, Shannon-capacity, and Lov\'asz-$\vartheta$ machinery is classical. The discrete-structure additions are the upward-closed agreement-set characterization of the generated graph class, the affine coordinate-matroid upper certificates, and the intersection-closure route to cluster-graph equality.

\subsection{Coordinate-Generated Graph Classes}\label{sec:related-meta}

Coordinate-view confusability is a graph-generation mechanism. The product-space state family and view family determine a monotone agreement-set graph class, and the affine specialization equips the same coordinates with a representable matroid. The construction connects three standard discrete objects: finite graphs, strong graph products, and representable matroids.

On the full $q$-ary cube, the symmetric view families sit inside the Hamming association scheme. If all $(d-1)$-coordinate views are admissible, two states are adjacent exactly when they differ in one coordinate, giving the Hamming graph $H(d,q)$. Proposition~\ref{prop:all-k-view-hamming-threshold} identifies the general all-$k$ case with a Hamming-threshold graph and turns independence into the classical code-distance condition. At the other end, a nested view chain is governed by its smallest view and yields a cluster graph whose components are projection fibers, with complete multipartite complements. The class $\mathcal{CVG}$ contains Hamming-type threshold graphs and cluster examples, but general presentations are controlled by arbitrary upward-closed agreement families rather than by distance alone.


\section{Conclusion}\label{sec:conclusion}

Affine coordinate-view exactness is certified by a Grassmannian avoidance count. The coordinate presentation supplies a representable matroid and a rank upper certificate. A \(t^*\)-dimensional linear section avoiding the admissible view kernels makes that rank certificate exact, and \(N_{t^*}\) records exactly those linear sections. Positivity of \(N_k\) is NP-complete for every fixed \(k\), and the rank-one reduction is parsimonious for counting.

The compression theorem identifies the finite data read by \(N_t\). Inclusion--exclusion over forbidden projective points factors through the signed rank profile of the represented forbidden-point matroid. The kernel-intersection lattice does not determine that profile. Common-core lifts freeze the multiway kernel-intersection shell while the quotient forbidden-point configuration changes the signed profile and the \(N_t\) value. The size-indexed fixed-shell SAT spine strengthens the separation: for each fixed Boolean variable count, one fixed multiway shell carries SAT decision, satisfying-assignment counts, and the paired decision/count target.

The same boundary gives the capacity results. On the kernel-section class,
\[
\Theta(G)=\log\alpha(G)=\min_{S\in\mathcal V}t(S)\log q,
\qquad
\operatorname{minrank}_{\mathbb K}(G)=q^{t^*}.
\]
The projection-equality matrix is an optimal Haemers witness, and logarithmic capacity adds under direct-sum block products inside the class. The greedy field-size theorem gives exactness when \(q\ge |\mathcal V|\). The Reed--Solomon/MDS and Hamming-endpoint families give exact instances beyond that field-size guarantee.

The full-tuple structural theorem supplies the coordinate-view normal form. Confusability relations on full products are exactly upward-closed agreement-set families, and nonempty view lists over nontrivial alphabets recover their generated families from the labeled graph. Transitive confusability is exactly intersection closure of that agreement family, equivalently meet-witnessing of the view list; the transitive case collapses to cluster counting.

Three problems remain.
\begin{enumerate}
\item \emph{Unlabeled coordinate-view isomorphism.} Given two full-tuple presentations by cofinal antichains, characterize when their unlabeled confusability graphs are isomorphic. The labeled theorem reduces equality to the generated upward family. The open question is whether every nondegenerate isomorphism is induced by a coordinate permutation, coordinatewise alphabet relabelings, and agreement-set preservation, or whether additional state-space symmetries occur.
\item \emph{Restricted-state transitivity.} For affine subspaces or product codes \(X\subsetneq A_1\times\cdots\times A_d\), find a presentation-level criterion equivalent to local composition and computable in \(\operatorname{poly}(r,d,L,\log q)\) time. Such a criterion must distinguish examples with the same view antichain and abstract coordinate matroid, such as the family \(C_q=\{(a,b,a+b):a,b\in\mathbb F_q\}\).
\item \emph{Affine rank-sieve compression.} The signed rank profile of the projectivized forbidden-point matroid computes \(N_t\). Identify structural parameters of represented forbidden-point matroids that make the signed profile, and hence \(N_{t^*}\), computable faster than direct Grassmannian enumeration. Fixed-shell families where the kernel-intersection profile is constant but \(N_t\) varies give the obstruction side of this problem.
\end{enumerate}

\paragraph{Open problem: sharp kernel-section boundary}
Let \(A=a_0+V\subseteq\mathbb F_q^d\) be an affine presentation and let \(\mathcal V\) be the admissible view family. Choose a rank-minimizing view
\[
S\in\mathcal V,\qquad t(S)=\min_{W\in\mathcal V} t(W),
\]
and write
\[
K_{\mathcal V}:=\bigcup_{W\in\mathcal V}\left(\ker(\pi_W|_V)\setminus\{0\}\right).
\]
Which affine presentations make \(\pi_S:A\to\pi_S(A)\) admit a section \(\sigma:\pi_S(A)\to A\) whose difference set
\[
\Delta(\sigma):=\{\sigma(u)-\sigma(v):u\ne v\}
\]
satisfies
\[
\Delta(\sigma)\cap K_{\mathcal V}=\emptyset?
\]
Such a section is a transversal of the \(S\)-fibers whose selected representatives are pairwise nonconfusable. It gives an independent set of size \(q^{t(S)}\), and the rank certificate gives the matching upper bound. Repeating the independent set in strong powers yields \(\Theta=t(S)\log q\).

The linear-section version asks for a linear right inverse \(R:\pi_S(V)\to V\) with
\[
\operatorname{im}R\cap K_{\mathcal V}=\emptyset.
\]
Equivalently, \(V\) contains a complement to \(\ker(\pi_S|_V)\) whose nonzero vectors avoid every admissible view kernel. Corollary~\ref{cor:arrangement-support-profile} identifies the linear version with the arrangement-relative support condition \(\Delta_{t^*}^{\mathcal V}(V)\ge1\). Corollary~\ref{cor:subcode-distance-boundary} specializes the condition to \(\Delta_{t^*}(V)>d-s\) for all \(s\)-coordinate views. Theorem~\ref{thm:field-size-exactness} proves the support condition when \(q\ge|\mathcal V|\). The Hamming-endpoint family gives exactness below that field-size threshold. The finite \(\mathbb F_3^6\) nonlinear-transversal witness separates rank-tight transversals from linear sections. The binary four-coordinate threshold example shows that small fields can fail to meet the rank bound once the view-kernel arrangement passes the field-size boundary.

\section*{Artifact Availability}

The Lean~4 formalization, supplementary proof ledger, build scripts, compiled PDF, and manuscript source artifacts are archived on Zenodo at \url{https://doi.org/10.5281/zenodo.20561409}.

\section*{Declaration of generative AI and AI-assisted technologies in the writing process}

Generative AI tools, including Codex, Claude, GLM, and Kimi, were used throughout this manuscript, across the abstract, introduction, theoretical development, proof sketches, related work, conclusion, supplementary material, and revision process. The tools were used for boilerplate generation, prose and notation refinement, \LaTeX{} and structure cleanup, Lean development, Lean/\LaTeX{} translation of informal proof ideas, and repeated reviewer-style critique passes to identify blind spots and clarity gaps.

Problem selection, theorem statements, assumptions, novelty framing, acceptance criteria, and final inclusion or exclusion decisions were made by the author. No technical claim was accepted solely from AI output. Formal claims reported as machine-verified were included only after Lean verification and direct author review. The author assumes responsibility for all content.

\section*{Declaration of competing interest}

The author declares no competing interests.

\bibliographystyle{siamplain}
\bibliography{references}

\end{document}


\maketitle

\section{Cited Lean Handle Ledger}

The ledger below lists the Lean handles cited by the manuscript, together with their declaration names and source modules.
To verify a cited handle such as \texttt{AFM1}, unpack the supplementary Lean artifact.
In the proof-project directory, run \texttt{lake build}, then inspect the declaration and source file listed in the ledger below.

\IfFileExists{content/lean_handle_ids_auto.tex}{%
\begingroup
\scriptsize
\setlength{\tabcolsep}{4pt}
\renewcommand{\arraystretch}{1.12}
\setlength{\LTpre}{2pt}
\setlength{\LTpost}{2pt}
\setlength{\emergencystretch}{3em}
\sloppy
\urlstyle{tt}
\makeatletter
\if@twocolumn
\begin{list}{}{\leftmargin=0pt\itemindent=0pt\itemsep=4pt\parsep=0pt\topsep=4pt}
\item \textbf{\nolinkurl{ABD1}}\hypertarget{lh:ABD1}{}\enspace{\ttfamily\nolinkurl{ArrangementBoundary.binarySquare_avoidingVectors_card}} {\tiny\ttfamily Ssot/\allowbreak ArrangementBoundary.lean}
\item \textbf{\nolinkurl{ABD2}}\hypertarget{lh:ABD2}{}\enspace{\ttfamily\nolinkurl{ArrangementBoundary.binarySquare_avoidingVectors_card_eq_inclusionExclusion}} {\tiny\ttfamily Ssot/\allowbreak ArrangementBoundary.lean}
\item \textbf{\nolinkurl{ABD3}}\hypertarget{lh:ABD3}{}\enspace{\ttfamily\nolinkurl{ArrangementBoundary.avoidingSpanSystem_exchange_failure}} {\tiny\ttfamily Ssot/\allowbreak ArrangementBoundary.lean}
\item \textbf{\nolinkurl{ABD4}}\hypertarget{lh:ABD4}{}\enspace{\ttfamily\nolinkurl{ArrangementBoundary.allTwoViewF2Four_avoidingLines_card}} {\tiny\ttfamily Ssot/\allowbreak ArrangementBoundary.lean}
\item \textbf{\nolinkurl{ABD5}}\hypertarget{lh:ABD5}{}\enspace{\ttfamily\nolinkurl{ArrangementBoundary.allTwoViewF2Four_avoidingPlanes_card}} {\tiny\ttfamily Ssot/\allowbreak ArrangementBoundary.lean}
\item \textbf{\nolinkurl{ABD6}}\hypertarget{lh:ABD6}{}\enspace{\ttfamily\nolinkurl{ArrangementBoundary.f3Three_collinearFourLines_avoidingPlanes_card}} {\tiny\ttfamily Ssot/\allowbreak ArrangementBoundary.lean}
\item \textbf{\nolinkurl{ABD7}}\hypertarget{lh:ABD7}{}\enspace{\ttfamily\nolinkurl{ArrangementBoundary.f3Three_generalFourLines_avoidingPlanes_card}} {\tiny\ttfamily Ssot/\allowbreak ArrangementBoundary.lean}
\item \textbf{\nolinkurl{ABD8}}\hypertarget{lh:ABD8}{}\enspace{\ttfamily\nolinkurl{ArrangementBoundary.f3Three_fourLineArrangements_samePairwiseProfile}} {\tiny\ttfamily Ssot/\allowbreak ArrangementBoundary.lean}
\item \textbf{\nolinkurl{ABD9}}\hypertarget{lh:ABD9}{}\enspace{\ttfamily\nolinkurl{ArrangementBoundary.alternating_powerset_sum}} {\tiny\ttfamily Ssot/\allowbreak ArrangementBoundary.lean}
\item \textbf{\nolinkurl{ABD10}}\hypertarget{lh:ABD10}{}\enspace{\ttfamily\nolinkurl{ArrangementBoundary.finite_forbidden_point_sieve}} {\tiny\ttfamily Ssot/\allowbreak ArrangementBoundary.lean}
\item \textbf{\nolinkurl{ABD11}}\hypertarget{lh:ABD11}{}\enspace{\ttfamily\nolinkurl{ArrangementBoundary.f3Three_fourLineArrangements_sameFullIntersectionProfile}} {\tiny\ttfamily Ssot/\allowbreak ArrangementBoundary.lean}
\item \textbf{\nolinkurl{ABD12}}\hypertarget{lh:ABD12}{}\enspace{\ttfamily\nolinkurl{ArrangementBoundary.f3Three_fourLineArrangements_sameIntersectionProfile_differentPlaneCounts}} {\tiny\ttfamily Ssot/\allowbreak ArrangementBoundary.lean}
\item \textbf{\nolinkurl{ABD21}}\hypertarget{lh:ABD21}{}\enspace{\ttfamily\nolinkurl{ArrangementBoundary.rankSieve}} {\tiny\ttfamily Ssot/\allowbreak ArrangementBoundary.lean}
\item \textbf{\nolinkurl{ABD22}}\hypertarget{lh:ABD22}{}\enspace{\ttfamily\nolinkurl{ArrangementBoundary.finite_forbidden_point_sieve_by_rank}} {\tiny\ttfamily Ssot/\allowbreak ArrangementBoundary.lean}
\item \textbf{\nolinkurl{ABD23}}\hypertarget{lh:ABD23}{}\enspace{\ttfamily\nolinkurl{ArrangementBoundary.rankSieve_eq_of_same_forbiddenPointMatroid}} {\tiny\ttfamily Ssot/\allowbreak ArrangementBoundary.lean}
\item \textbf{\nolinkurl{ABD24}}\hypertarget{lh:ABD24}{}\enspace{\ttfamily\nolinkurl{ArrangementBoundary.f3Three_fourLineArrangements_differentForbiddenPointMatroids}} {\tiny\ttfamily Ssot/\allowbreak ArrangementBoundary.lean}
\item \textbf{\nolinkurl{ABD25}}\hypertarget{lh:ABD25}{}\enspace{\ttfamily\nolinkurl{ArrangementBoundary.f3Three_sameIntersectionProfile_differentForbiddenPointMatroids}} {\tiny\ttfamily Ssot/\allowbreak ArrangementBoundary.lean}
\item \textbf{\nolinkurl{ABD26}}\hypertarget{lh:ABD26}{}\enspace{\ttfamily\nolinkurl{ArrangementBoundary.f3CommonCore_sameIntersectionProfile}} {\tiny\ttfamily Ssot/\allowbreak ArrangementBoundary.lean}
\item \textbf{\nolinkurl{ABD27}}\hypertarget{lh:ABD27}{}\enspace{\ttfamily\nolinkurl{ArrangementBoundary.f3CommonCore_sameIntersectionProfile_differentQuotientPointMatroids}} {\tiny\ttfamily Ssot/\allowbreak ArrangementBoundary.lean}
\item \textbf{\nolinkurl{ABD28}}\hypertarget{lh:ABD28}{}\enspace{\ttfamily\nolinkurl{ArrangementBoundary.commonCore_multiIntersection_eq_core}} {\tiny\ttfamily Ssot/\allowbreak ArrangementBoundary.lean}
\item \textbf{\nolinkurl{ABD29}}\hypertarget{lh:ABD29}{}\enspace{\ttfamily\nolinkurl{ArrangementBoundary.commonCore_multiIntersection_eq_between_families}} {\tiny\ttfamily Ssot/\allowbreak ArrangementBoundary.lean}
\item \textbf{\nolinkurl{ABD30}}\hypertarget{lh:ABD30}{}\enspace{\ttfamily\nolinkurl{ArrangementBoundary.productCore}} {\tiny\ttfamily Ssot/\allowbreak ArrangementBoundary.lean}
\item \textbf{\nolinkurl{ABD31}}\hypertarget{lh:ABD31}{}\enspace{\ttfamily\nolinkurl{ArrangementBoundary.productLift}} {\tiny\ttfamily Ssot/\allowbreak ArrangementBoundary.lean}
\item \textbf{\nolinkurl{ABD32}}\hypertarget{lh:ABD32}{}\enspace{\ttfamily\nolinkurl{ArrangementBoundary.productLift_multiIntersection_eq_productCore}} {\tiny\ttfamily Ssot/\allowbreak ArrangementBoundary.lean}
\item \textbf{\nolinkurl{ABD33}}\hypertarget{lh:ABD33}{}\enspace{\ttfamily\nolinkurl{ArrangementBoundary.productLift_multiIntersection_eq_between_quotient_families}} {\tiny\ttfamily Ssot/\allowbreak ArrangementBoundary.lean}
\item \textbf{\nolinkurl{ABD34}}\hypertarget{lh:ABD34}{}\enspace{\ttfamily\nolinkurl{ArrangementBoundary.graphOver_inter_productLift_eq}} {\tiny\ttfamily Ssot/\allowbreak ArrangementBoundary.lean}
\item \textbf{\nolinkurl{ABD35}}\hypertarget{lh:ABD35}{}\enspace{\ttfamily\nolinkurl{ArrangementBoundary.graphOver_inter_productLift_eq_graphOver_iff}} {\tiny\ttfamily Ssot/\allowbreak ArrangementBoundary.lean}
\item \textbf{\nolinkurl{ABD36}}\hypertarget{lh:ABD36}{}\enspace{\ttfamily\nolinkurl{ArrangementBoundary.quotientAvoidsFamily}} {\tiny\ttfamily Ssot/\allowbreak ArrangementBoundary.lean}
\item \textbf{\nolinkurl{ABD37}}\hypertarget{lh:ABD37}{}\enspace{\ttfamily\nolinkurl{ArrangementBoundary.graphOverAvoidsLiftedFamily}} {\tiny\ttfamily Ssot/\allowbreak ArrangementBoundary.lean}
\item \textbf{\nolinkurl{ABD38}}\hypertarget{lh:ABD38}{}\enspace{\ttfamily\nolinkurl{ArrangementBoundary.graphOverAvoidsLiftedFamily_iff_quotientAvoidsFamily}} {\tiny\ttfamily Ssot/\allowbreak ArrangementBoundary.lean}
\item \textbf{\nolinkurl{ABD39}}\hypertarget{lh:ABD39}{}\enspace{\ttfamily\nolinkurl{ArrangementBoundary.coreValuedOn}} {\tiny\ttfamily Ssot/\allowbreak ArrangementBoundary.lean}
\item \textbf{\nolinkurl{ABD40}}\hypertarget{lh:ABD40}{}\enspace{\ttfamily\nolinkurl{ArrangementBoundary.quotientAvoidancePositive}} {\tiny\ttfamily Ssot/\allowbreak ArrangementBoundary.lean}
\item \textbf{\nolinkurl{ABD41}}\hypertarget{lh:ABD41}{}\enspace{\ttfamily\nolinkurl{ArrangementBoundary.liftedGraphSectionPositive}} {\tiny\ttfamily Ssot/\allowbreak ArrangementBoundary.lean}
\item \textbf{\nolinkurl{ABD42}}\hypertarget{lh:ABD42}{}\enspace{\ttfamily\nolinkurl{ArrangementBoundary.liftedGraphSectionPositive_iff_quotientAvoidancePositive}} {\tiny\ttfamily Ssot/\allowbreak ArrangementBoundary.lean}
\item \textbf{\nolinkurl{ABD43}}\hypertarget{lh:ABD43}{}\enspace{\ttfamily\nolinkurl{ArrangementBoundary.quotientAvoidsZeroAndFamily}} {\tiny\ttfamily Ssot/\allowbreak ArrangementBoundary.lean}
\item \textbf{\nolinkurl{ABD44}}\hypertarget{lh:ABD44}{}\enspace{\ttfamily\nolinkurl{ArrangementBoundary.graphOverAvoidsVerticalAndLiftedFamily}} {\tiny\ttfamily Ssot/\allowbreak ArrangementBoundary.lean}
\item \textbf{\nolinkurl{ABD45}}\hypertarget{lh:ABD45}{}\enspace{\ttfamily\nolinkurl{ArrangementBoundary.graphOver_inter_productCore_eq_graphOver_iff}} {\tiny\ttfamily Ssot/\allowbreak ArrangementBoundary.lean}
\item \textbf{\nolinkurl{ABD46}}\hypertarget{lh:ABD46}{}\enspace{\ttfamily\nolinkurl{ArrangementBoundary.graphOverAvoidsVerticalAndLiftedFamily_iff_quotientAvoidsZeroAndFamily}} {\tiny\ttfamily Ssot/\allowbreak ArrangementBoundary.lean}
\item \textbf{\nolinkurl{ABD47}}\hypertarget{lh:ABD47}{}\enspace{\ttfamily\nolinkurl{ArrangementBoundary.quotientZeroAndFamilyPositive}} {\tiny\ttfamily Ssot/\allowbreak ArrangementBoundary.lean}
\item \textbf{\nolinkurl{ABD48}}\hypertarget{lh:ABD48}{}\enspace{\ttfamily\nolinkurl{ArrangementBoundary.liftedGraphSectionVerticalAndFamilyPositive}} {\tiny\ttfamily Ssot/\allowbreak ArrangementBoundary.lean}
\item \textbf{\nolinkurl{ABD49}}\hypertarget{lh:ABD49}{}\enspace{\ttfamily\nolinkurl{ArrangementBoundary.liftedGraphSectionVerticalAndFamilyPositive_iff_quotientZeroAndFamilyPositive}} {\tiny\ttfamily Ssot/\allowbreak ArrangementBoundary.lean}
\item \textbf{\nolinkurl{ABD50}}\hypertarget{lh:ABD50}{}\enspace{\ttfamily\nolinkurl{ArrangementBoundary.QuotientSearch}} {\tiny\ttfamily Ssot/\allowbreak ArrangementBoundary.lean}
\item \textbf{\nolinkurl{ABD51}}\hypertarget{lh:ABD51}{}\enspace{\ttfamily\nolinkurl{ArrangementBoundary.LiftedSearch}} {\tiny\ttfamily Ssot/\allowbreak ArrangementBoundary.lean}
\item \textbf{\nolinkurl{ABD52}}\hypertarget{lh:ABD52}{}\enspace{\ttfamily\nolinkurl{ArrangementBoundary.QuotientSearch.commonCoreLift}} {\tiny\ttfamily Ssot/\allowbreak ArrangementBoundary.lean}
\item \textbf{\nolinkurl{ABD53}}\hypertarget{lh:ABD53}{}\enspace{\ttfamily\nolinkurl{ArrangementBoundary.commonCoreLift_positive_iff_quotientPositive}} {\tiny\ttfamily Ssot/\allowbreak ArrangementBoundary.lean}
\item \textbf{\nolinkurl{ABD54}}\hypertarget{lh:ABD54}{}\enspace{\ttfamily\nolinkurl{ArrangementBoundary.avoiding_count_eq_of_same_forbiddenPointMatroid}} {\tiny\ttfamily Ssot/\allowbreak ArrangementBoundary.lean}
\item \textbf{\nolinkurl{ABD55}}\hypertarget{lh:ABD55}{}\enspace{\ttfamily\nolinkurl{ArrangementBoundary.avoiding_count_positive_iff_of_same_forbiddenPointMatroid}} {\tiny\ttfamily Ssot/\allowbreak ArrangementBoundary.lean}
\item \textbf{\nolinkurl{ABD56}}\hypertarget{lh:ABD56}{}\enspace{\ttfamily\nolinkurl{ArrangementBoundary.finite_forbidden_point_sieve_by_rank_positive_iff}} {\tiny\ttfamily Ssot/\allowbreak ArrangementBoundary.lean}
\item \textbf{\nolinkurl{ABD57}}\hypertarget{lh:ABD57}{}\enspace{\ttfamily\nolinkurl{ArrangementBoundary.rankSieve_positive_iff_of_same_forbiddenPointMatroid}} {\tiny\ttfamily Ssot/\allowbreak ArrangementBoundary.lean}
\item \textbf{\nolinkurl{ABD58}}\hypertarget{lh:ABD58}{}\enspace{\ttfamily\nolinkurl{ArrangementBoundary.productLift_multiIntersectionShell_eq_between_quotient_families}} {\tiny\ttfamily Ssot/\allowbreak ArrangementBoundary.lean}
\item \textbf{\nolinkurl{ABD59}}\hypertarget{lh:ABD59}{}\enspace{\ttfamily\nolinkurl{ArrangementBoundary.exists_avoidsAllForbidden_of_incidence_card_bound}} {\tiny\ttfamily Ssot/\allowbreak ArrangementBoundary.lean}
\item \textbf{\nolinkurl{ABD60}}\hypertarget{lh:ABD60}{}\enspace{\ttfamily\nolinkurl{ArrangementBoundary.objects_card_le_points_mul_of_no_avoidsAllForbidden}} {\tiny\ttfamily Ssot/\allowbreak ArrangementBoundary.lean}
\item \textbf{\nolinkurl{ABD61}}\hypertarget{lh:ABD61}{}\enspace{\ttfamily\nolinkurl{ArrangementBoundary.f3ProjectivePoints_fourPointSubsets_card}} {\tiny\ttfamily Ssot/\allowbreak ArrangementBoundary.lean}
\item \textbf{\nolinkurl{ABD62}}\hypertarget{lh:ABD62}{}\enspace{\ttfamily\nolinkurl{ArrangementBoundary.f3FourPointSubsets_commonCore_pairwiseCore}} {\tiny\ttfamily Ssot/\allowbreak ArrangementBoundary.lean}
\item \textbf{\nolinkurl{ABD63}}\hypertarget{lh:ABD63}{}\enspace{\ttfamily\nolinkurl{ArrangementBoundary.f3FourPointSubsets_avoidingPlaneCount_distribution}} {\tiny\ttfamily Ssot/\allowbreak ArrangementBoundary.lean}
\item \textbf{\nolinkurl{ABD64}}\hypertarget{lh:ABD64}{}\enspace{\ttfamily\nolinkurl{ArrangementBoundary.f3FourPointWitnesses_fixedFullShell_realize_threeAvoidingPlaneCounts}} {\tiny\ttfamily Ssot/\allowbreak ArrangementBoundary.lean}
\item \textbf{\nolinkurl{ABD65}}\hypertarget{lh:ABD65}{}\enspace{\ttfamily\nolinkurl{ArrangementBoundary.f3FourPointWitnesses_fixedFullShell_threeForbiddenPointMatroidProfiles}} {\tiny\ttfamily Ssot/\allowbreak ArrangementBoundary.lean}
\item \textbf{\nolinkurl{ABD66}}\hypertarget{lh:ABD66}{}\enspace{\ttfamily\nolinkurl{ArrangementBoundary.f3ThreePointSubsets_avoidingPlaneCount_distribution}} {\tiny\ttfamily Ssot/\allowbreak ArrangementBoundary.lean}
\item \textbf{\nolinkurl{ABD67}}\hypertarget{lh:ABD67}{}\enspace{\ttfamily\nolinkurl{ArrangementBoundary.f3FivePointSubsets_avoidingPlaneCount_distribution}} {\tiny\ttfamily Ssot/\allowbreak ArrangementBoundary.lean}
\item \textbf{\nolinkurl{ABD68}}\hypertarget{lh:ABD68}{}\enspace{\ttfamily\nolinkurl{ArrangementBoundary.f3SixPointSubsets_avoidingPlaneCount_distribution}} {\tiny\ttfamily Ssot/\allowbreak ArrangementBoundary.lean}
\item \textbf{\nolinkurl{ABD69}}\hypertarget{lh:ABD69}{}\enspace{\ttfamily\nolinkurl{ArrangementBoundary.quotientZeroAndFamilyCandidates}} {\tiny\ttfamily Ssot/\allowbreak ArrangementBoundary.lean}
\item \textbf{\nolinkurl{ABD70}}\hypertarget{lh:ABD70}{}\enspace{\ttfamily\nolinkurl{ArrangementBoundary.liftedGraphSectionVerticalAndFamilyWitnesses}} {\tiny\ttfamily Ssot/\allowbreak ArrangementBoundary.lean}
\item \textbf{\nolinkurl{ABD71}}\hypertarget{lh:ABD71}{}\enspace{\ttfamily\nolinkurl{ArrangementBoundary.mem_liftedGraphSectionVerticalAndFamilyWitnesses}} {\tiny\ttfamily Ssot/\allowbreak ArrangementBoundary.lean}
\item \textbf{\nolinkurl{ABD72}}\hypertarget{lh:ABD72}{}\enspace{\ttfamily\nolinkurl{ArrangementBoundary.liftedGraphSectionVerticalAndFamilyWitnesses_card_eq_sum_maps_over_quotientCandidates}} {\tiny\ttfamily Ssot/\allowbreak ArrangementBoundary.lean}
\item \textbf{\nolinkurl{ABD73}}\hypertarget{lh:ABD73}{}\enspace{\ttfamily\nolinkurl{ArrangementBoundary.liftedGraphSectionVerticalAndFamilyWitnesses_card_eq_quotientCandidates_card_mul}} {\tiny\ttfamily Ssot/\allowbreak ArrangementBoundary.lean}
\item \textbf{\nolinkurl{ABD74}}\hypertarget{lh:ABD74}{}\enspace{\ttfamily\nolinkurl{ArrangementBoundary.QuotientSearch.positiveCandidates}} {\tiny\ttfamily Ssot/\allowbreak ArrangementBoundary.lean}
\item \textbf{\nolinkurl{ABD75}}\hypertarget{lh:ABD75}{}\enspace{\ttfamily\nolinkurl{ArrangementBoundary.LiftedSearch.witnesses}} {\tiny\ttfamily Ssot/\allowbreak ArrangementBoundary.lean}
\item \textbf{\nolinkurl{ABD76}}\hypertarget{lh:ABD76}{}\enspace{\ttfamily\nolinkurl{ArrangementBoundary.commonCoreLift_witnesses_card_eq_sum_maps_over_positiveCandidates}} {\tiny\ttfamily Ssot/\allowbreak ArrangementBoundary.lean}
\item \textbf{\nolinkurl{ABD77}}\hypertarget{lh:ABD77}{}\enspace{\ttfamily\nolinkurl{ArrangementBoundary.commonCoreLift_witnesses_card_eq_positiveCandidates_card_mul}} {\tiny\ttfamily Ssot/\allowbreak ArrangementBoundary.lean}
\item \textbf{\nolinkurl{ABD78}}\hypertarget{lh:ABD78}{}\enspace{\ttfamily\nolinkurl{ArrangementBoundary.f3FourPointWitnesses_fixedFullShell_threeCounts_and_threeForbiddenPointMatroidProfiles}} {\tiny\ttfamily Ssot/\allowbreak ArrangementBoundary.lean}
\item \textbf{\nolinkurl{ABD79}}\hypertarget{lh:ABD79}{}\enspace{\ttfamily\nolinkurl{ArrangementBoundary.signedRankProfile}} {\tiny\ttfamily Ssot/\allowbreak ArrangementBoundary.lean}
\item \textbf{\nolinkurl{ABD80}}\hypertarget{lh:ABD80}{}\enspace{\ttfamily\nolinkurl{ArrangementBoundary.rankSieve_eq_sum_signedRankProfile}} {\tiny\ttfamily Ssot/\allowbreak ArrangementBoundary.lean}
\item \textbf{\nolinkurl{ABD81}}\hypertarget{lh:ABD81}{}\enspace{\ttfamily\nolinkurl{ArrangementBoundary.signedRankProfile_eq_zero_of_not_mem_image}} {\tiny\ttfamily Ssot/\allowbreak ArrangementBoundary.lean}
\item \textbf{\nolinkurl{ABD82}}\hypertarget{lh:ABD82}{}\enspace{\ttfamily\nolinkurl{ArrangementBoundary.rankSieve_eq_of_same_signedRankProfile}} {\tiny\ttfamily Ssot/\allowbreak ArrangementBoundary.lean}
\item \textbf{\nolinkurl{ABD83}}\hypertarget{lh:ABD83}{}\enspace{\ttfamily\nolinkurl{ArrangementBoundary.rankSieve_positive_iff_of_same_signedRankProfile}} {\tiny\ttfamily Ssot/\allowbreak ArrangementBoundary.lean}
\item \textbf{\nolinkurl{ABD84}}\hypertarget{lh:ABD84}{}\enspace{\ttfamily\nolinkurl{ArrangementBoundary.exists_forbiddenPointMatroid_rank_ne_of_rankSieve_ne}} {\tiny\ttfamily Ssot/\allowbreak ArrangementBoundary.lean}
\item \textbf{\nolinkurl{ABD85}}\hypertarget{lh:ABD85}{}\enspace{\ttfamily\nolinkurl{ArrangementBoundary.exists_forbiddenPointMatroid_rank_ne_of_rankSieve_positive_ne}} {\tiny\ttfamily Ssot/\allowbreak ArrangementBoundary.lean}
\item \textbf{\nolinkurl{ABD86}}\hypertarget{lh:ABD86}{}\enspace{\ttfamily\nolinkurl{ArrangementBoundary.f3FourPoint_fixedFullShell_exists_differentAvoidingPlaneCount}} {\tiny\ttfamily Ssot/\allowbreak ArrangementBoundary.lean}
\item \textbf{\nolinkurl{ABD87}}\hypertarget{lh:ABD87}{}\enspace{\ttfamily\nolinkurl{ArrangementBoundary.f3FourPoint_fixedFullShell_not_avoidingPlaneCountInvariant}} {\tiny\ttfamily Ssot/\allowbreak ArrangementBoundary.lean}
\item \textbf{\nolinkurl{ABD88}}\hypertarget{lh:ABD88}{}\enspace{\ttfamily\nolinkurl{ArrangementBoundary.f3FourPoint_fixedFullShell_exists_differentForbiddenPointMatroidProfile}} {\tiny\ttfamily Ssot/\allowbreak ArrangementBoundary.lean}
\item \textbf{\nolinkurl{ABD89}}\hypertarget{lh:ABD89}{}\enspace{\ttfamily\nolinkurl{ArrangementBoundary.f3FourPoint_fixedFullShell_not_forbiddenPointMatroidProfileInvariant}} {\tiny\ttfamily Ssot/\allowbreak ArrangementBoundary.lean}
\item \textbf{\nolinkurl{ABD90}}\hypertarget{lh:ABD90}{}\enspace{\ttfamily\nolinkurl{ArrangementBoundary.signedRankProfileDecode}} {\tiny\ttfamily Ssot/\allowbreak ArrangementBoundary.lean}
\item \textbf{\nolinkurl{ABD91}}\hypertarget{lh:ABD91}{}\enspace{\ttfamily\nolinkurl{ArrangementBoundary.rankSieve_eq_signedRankProfileDecode_of_rank_le_card}} {\tiny\ttfamily Ssot/\allowbreak ArrangementBoundary.lean}
\item \textbf{\nolinkurl{AFM1}}\hypertarget{lh:AFM1}{}\enspace{\ttfamily\nolinkurl{FactMatroid.determinesFact_iff_mem_factSpan}} {\tiny\ttfamily Ssot/\allowbreak FactMatroid.lean}
\item \textbf{\nolinkurl{AFM2}}\hypertarget{lh:AFM2}{}\enspace{\ttfamily\nolinkurl{FactMatroid.factMatroid_indep_iff}} {\tiny\ttfamily Ssot/\allowbreak FactMatroid.lean}
\item \textbf{\nolinkurl{AFM3}}\hypertarget{lh:AFM3}{}\enspace{\ttfamily\nolinkurl{FactMatroid.basisFacts_card_eq_finrank}} {\tiny\ttfamily Ssot/\allowbreak FactMatroid.lean}
\item \textbf{\nolinkurl{AFM4}}\hypertarget{lh:AFM4}{}\enspace{\ttfamily\nolinkurl{FactMatroid.basisFacts_minimal}} {\tiny\ttfamily Ssot/\allowbreak FactMatroid.lean}
\item \textbf{\nolinkurl{AFM5}}\hypertarget{lh:AFM5}{}\enspace{\ttfamily\nolinkurl{FactMatroid.minimalDetermining_basis}} {\tiny\ttfamily Ssot/\allowbreak FactMatroid.lean}
\item \textbf{\nolinkurl{AFM6}}\hypertarget{lh:AFM6}{}\enspace{\ttfamily\nolinkurl{FactMatroid.factRankFinset_le_card}} {\tiny\ttfamily Ssot/\allowbreak FactMatroid.lean}
\item \textbf{\nolinkurl{AFM7}}\hypertarget{lh:AFM7}{}\enspace{\ttfamily\nolinkurl{FactMatroid.coordProjection_range_natCard_eq_pow_factRankFinset}} {\tiny\ttfamily Ssot/\allowbreak AffineCardinality.lean}
\item \textbf{\nolinkurl{AFM8}}\hypertarget{lh:AFM8}{}\enspace{\ttfamily\nolinkurl{FactMatroid.factRankFinset_eq_card_of_indepFacts}} {\tiny\ttfamily Ssot/\allowbreak FactMatroid.lean}
\item \textbf{\nolinkurl{AFM9}}\hypertarget{lh:AFM9}{}\enspace{\ttfamily\nolinkurl{FactMatroid.factRankFinset_eq_total_of_determinesAllFinset}} {\tiny\ttfamily Ssot/\allowbreak FactMatroid.lean}
\item \textbf{\nolinkurl{AFM10}}\hypertarget{lh:AFM10}{}\enspace{\ttfamily\nolinkurl{FactMatroid.factSpanFinset_eq_range_coordProjection_dualMap}} {\tiny\ttfamily Ssot/\allowbreak FactMatroid.lean}
\item \textbf{\nolinkurl{AFM11}}\hypertarget{lh:AFM11}{}\enspace{\ttfamily\nolinkurl{FactMatroid.factRankFinset_eq_finrank_range_coordProjection}} {\tiny\ttfamily Ssot/\allowbreak FactMatroid.lean}
\item \textbf{\nolinkurl{AFM12}}\hypertarget{lh:AFM12}{}\enspace{\ttfamily\nolinkurl{FactMatroid.factRankFinset_le_finrank_directionSpace}} {\tiny\ttfamily Ssot/\allowbreak FactMatroid.lean}
\item \textbf{\nolinkurl{AFM17}}\hypertarget{lh:AFM17}{}\enspace{\ttfamily\nolinkurl{FactMatroid.factRankFinset_product_disjSum}} {\tiny\ttfamily Ssot/\allowbreak FactMatroid.lean}
\item \textbf{\nolinkurl{AFM18}}\hypertarget{lh:AFM18}{}\enspace{\ttfamily\nolinkurl{FactMatroid.indepFacts_productFamily_iff}} {\tiny\ttfamily Ssot/\allowbreak FactMatroid.lean}
\item \textbf{\nolinkurl{AFM19}}\hypertarget{lh:AFM19}{}\enspace{\ttfamily\nolinkurl{FactMatroid.factMatroid_product_indep_iff}} {\tiny\ttfamily Ssot/\allowbreak FactMatroid.lean}
\item \textbf{\nolinkurl{AFM20}}\hypertarget{lh:AFM20}{}\enspace{\ttfamily\nolinkurl{FactMatroid.basisFacts_iff_minimalDetermining}} {\tiny\ttfamily Ssot/\allowbreak FactMatroid.lean}
\item \textbf{\nolinkurl{AFM21}}\hypertarget{lh:AFM21}{}\enspace{\ttfamily\nolinkurl{FactMatroid.minimalDetermining_card_eq_finrank}} {\tiny\ttfamily Ssot/\allowbreak FactMatroid.lean}
\item \textbf{\nolinkurl{AFM22}}\hypertarget{lh:AFM22}{}\enspace{\ttfamily\nolinkurl{FactMatroid.minimalDetermining_card_eq_factRank_univ}} {\tiny\ttfamily Ssot/\allowbreak FactMatroid.lean}
\item \textbf{\nolinkurl{AFM24}}\hypertarget{lh:AFM24}{}\enspace{\ttfamily\nolinkurl{FactMatroid.coordProjection_ker_finrank_eq_finrank_sub_factRankFinset}} {\tiny\ttfamily Ssot/\allowbreak AffineCardinality.lean}
\item \textbf{\nolinkurl{AKS1}}\hypertarget{lh:AKS1}{}\enspace{\ttfamily\nolinkurl{FactMatroid.KernelAvoidingSubspace}} {\tiny\ttfamily Ssot/\allowbreak AffineKernelSection.lean}
\item \textbf{\nolinkurl{AKS2}}\hypertarget{lh:AKS2}{}\enspace{\ttfamily\nolinkurl{FactMatroid.affineLinearSection_independent}} {\tiny\ttfamily Ssot/\allowbreak AffineKernelSection.lean}
\item \textbf{\nolinkurl{AKS3}}\hypertarget{lh:AKS3}{}\enspace{\ttfamily\nolinkurl{FactMatroid.affineLinearSectionFinset_card_eq_pow_finrank}} {\tiny\ttfamily Ssot/\allowbreak AffineKernelSection.lean}
\item \textbf{\nolinkurl{AKS4}}\hypertarget{lh:AKS4}{}\enspace{\ttfamily\nolinkurl{MultiFact.graphShannonCapacityReal_eq_log_of_graphIndependent_upper}} {\tiny\ttfamily Ssot/\allowbreak CapacityLower.lean}
\item \textbf{\nolinkurl{AKS5}}\hypertarget{lh:AKS5}{}\enspace{\ttfamily\nolinkurl{FactMatroid.affineLinearSection_capacity_exact_of_upper}} {\tiny\ttfamily Ssot/\allowbreak AffineKernelSection.lean}
\item \textbf{\nolinkurl{AKS6}}\hypertarget{lh:AKS6}{}\enspace{\ttfamily\nolinkurl{FactMatroid.exists_vector_not_mem_finite_union_of_proper_subspaces}} {\tiny\ttfamily Ssot/\allowbreak AffineKernelSection.lean}
\item \textbf{\nolinkurl{AKS7}}\hypertarget{lh:AKS7}{}\enspace{\ttfamily\nolinkurl{FactMatroid.exists_vector_not_mem_finset_union_of_proper_subspaces}} {\tiny\ttfamily Ssot/\allowbreak AffineKernelSection.lean}
\item \textbf{\nolinkurl{AKS8}}\hypertarget{lh:AKS8}{}\enspace{\ttfamily\nolinkurl{FactMatroid.exists_kernelAvoidingSubspace_strictExtension}} {\tiny\ttfamily Ssot/\allowbreak AffineKernelSection.lean}
\item \textbf{\nolinkurl{AKS9}}\hypertarget{lh:AKS9}{}\enspace{\ttfamily\nolinkurl{FactMatroid.exists_kernelAvoidingSubspace_of_large_field}} {\tiny\ttfamily Ssot/\allowbreak AffineKernelSection.lean}
\item \textbf{\nolinkurl{AKS10}}\hypertarget{lh:AKS10}{}\enspace{\ttfamily\nolinkurl{NonlinearSectionWitness.witnessPoints_card}} {\tiny\ttfamily Ssot/\allowbreak NonlinearSectionWitness.lean}
\item \textbf{\nolinkurl{AKS11}}\hypertarget{lh:AKS11}{}\enspace{\ttfamily\nolinkurl{NonlinearSectionWitness.diffSet_card}} {\tiny\ttfamily Ssot/\allowbreak NonlinearSectionWitness.lean}
\item \textbf{\nolinkurl{AKS12}}\hypertarget{lh:AKS12}{}\enspace{\ttfamily\nolinkurl{NonlinearSectionWitness.rankMinKernel_card}} {\tiny\ttfamily Ssot/\allowbreak NonlinearSectionWitness.lean}
\item \textbf{\nolinkurl{AKS13}}\hypertarget{lh:AKS13}{}\enspace{\ttfamily\nolinkurl{NonlinearSectionWitness.rankMinKernel_disjoint_diffSet}} {\tiny\ttfamily Ssot/\allowbreak NonlinearSectionWitness.lean}
\item \textbf{\nolinkurl{AKS14}}\hypertarget{lh:AKS14}{}\enspace{\ttfamily\nolinkurl{NonlinearSectionWitness.projectiveReps_card}} {\tiny\ttfamily Ssot/\allowbreak NonlinearSectionWitness.lean}
\item \textbf{\nolinkurl{AKS15}}\hypertarget{lh:AKS15}{}\enspace{\ttfamily\nolinkurl{NonlinearSectionWitness.diffDirections_card}} {\tiny\ttfamily Ssot/\allowbreak NonlinearSectionWitness.lean}
\item \textbf{\nolinkurl{AKS16}}\hypertarget{lh:AKS16}{}\enspace{\ttfamily\nolinkurl{NonlinearSectionWitness.no_projective_line_in_diffDirections}} {\tiny\ttfamily Ssot/\allowbreak NonlinearSectionWitness.lean}
\item \textbf{\nolinkurl{AKS17}}\hypertarget{lh:AKS17}{}\enspace{\ttfamily\nolinkurl{FactMatroid.exists_vector_not_mem_finite_union_of_proper_subspaces_of_card_le}} {\tiny\ttfamily Ssot/\allowbreak AffineKernelSection.lean}
\item \textbf{\nolinkurl{AKS18}}\hypertarget{lh:AKS18}{}\enspace{\ttfamily\nolinkurl{FactMatroid.exists_vector_not_mem_finset_union_of_proper_subspaces_of_card_le}} {\tiny\ttfamily Ssot/\allowbreak AffineKernelSection.lean}
\item \textbf{\nolinkurl{AKS19}}\hypertarget{lh:AKS19}{}\enspace{\ttfamily\nolinkurl{FactMatroid.diagonalSubspaceInTop_avoids_singleton}} {\tiny\ttfamily Ssot/\allowbreak AffineKernelSection.lean}
\item \textbf{\nolinkurl{AKS20}}\hypertarget{lh:AKS20}{}\enspace{\ttfamily\nolinkurl{FactMatroid.ViewSupportHittingSubspace}} {\tiny\ttfamily Ssot/\allowbreak AffineKernelSection.lean}
\item \textbf{\nolinkurl{AKS21}}\hypertarget{lh:AKS21}{}\enspace{\ttfamily\nolinkurl{FactMatroid.kernelAvoidingSubspace_of_viewSupportHitting}} {\tiny\ttfamily Ssot/\allowbreak AffineKernelSection.lean}
\item \textbf{\nolinkurl{AKS22}}\hypertarget{lh:AKS22}{}\enspace{\ttfamily\nolinkurl{FactMatroid.AffineViewConfusableTransitive}} {\tiny\ttfamily Ssot/\allowbreak AffineKernelSection.lean}
\item \textbf{\nolinkurl{AKS23}}\hypertarget{lh:AKS23}{}\enspace{\ttfamily\nolinkurl{FactMatroid.AffineKernelSumCovered}} {\tiny\ttfamily Ssot/\allowbreak AffineKernelSection.lean}
\item \textbf{\nolinkurl{AKS24}}\hypertarget{lh:AKS24}{}\enspace{\ttfamily\nolinkurl{FactMatroid.AffineKernelSumWitnessed}} {\tiny\ttfamily Ssot/\allowbreak AffineKernelSection.lean}
\item \textbf{\nolinkurl{AKS25}}\hypertarget{lh:AKS25}{}\enspace{\ttfamily\nolinkurl{FactMatroid.affineCoordProjection_eq_iff_sub_mem_ker}} {\tiny\ttfamily Ssot/\allowbreak AffineKernelSection.lean}
\item \textbf{\nolinkurl{AKS26}}\hypertarget{lh:AKS26}{}\enspace{\ttfamily\nolinkurl{FactMatroid.affineKernelSumWitnessed_kernelSumCovered}} {\tiny\ttfamily Ssot/\allowbreak AffineKernelSection.lean}
\item \textbf{\nolinkurl{AKS27}}\hypertarget{lh:AKS27}{}\enspace{\ttfamily\nolinkurl{FactMatroid.affineViewConfusableTransitive_of_kernelSumCovered}} {\tiny\ttfamily Ssot/\allowbreak AffineKernelSection.lean}
\item \textbf{\nolinkurl{AKS28}}\hypertarget{lh:AKS28}{}\enspace{\ttfamily\nolinkurl{FactMatroid.affineKernelSumCovered_of_viewConfusableTransitive}} {\tiny\ttfamily Ssot/\allowbreak AffineKernelSection.lean}
\item \textbf{\nolinkurl{AKS29}}\hypertarget{lh:AKS29}{}\enspace{\ttfamily\nolinkurl{FactMatroid.affineViewConfusableTransitive_iff_kernelSumCovered}} {\tiny\ttfamily Ssot/\allowbreak AffineKernelSection.lean}
\item \textbf{\nolinkurl{AKS30}}\hypertarget{lh:AKS30}{}\enspace{\ttfamily\nolinkurl{FactMatroid.affineKernelSumWitnessed_of_kernelSumCovered_of_card_le}} {\tiny\ttfamily Ssot/\allowbreak AffineKernelSection.lean}
\item \textbf{\nolinkurl{AKS31}}\hypertarget{lh:AKS31}{}\enspace{\ttfamily\nolinkurl{FactMatroid.affineViewConfusableTransitive_iff_kernelSumWitnessed_of_card_le}} {\tiny\ttfamily Ssot/\allowbreak AffineKernelSection.lean}
\item \textbf{\nolinkurl{AKS32}}\hypertarget{lh:AKS32}{}\enspace{\ttfamily\nolinkurl{FactMatroid.affineLinearSection_noShannonGap_of_upper}} {\tiny\ttfamily Ssot/\allowbreak AffineKernelSection.lean}
\item \textbf{\nolinkurl{AKS33}}\hypertarget{lh:AKS33}{}\enspace{\ttfamily\nolinkurl{MultiFact.graphShannonCapacityReal_eq_log_of_maxIndependentCard_upper}} {\tiny\ttfamily Ssot/\allowbreak CapacityLower.lean}
\item \textbf{\nolinkurl{AKS34}}\hypertarget{lh:AKS34}{}\enspace{\ttfamily\nolinkurl{FactMatroid.supportHitsAllCardViews_iff_compl_card_lt}} {\tiny\ttfamily Ssot/\allowbreak AffineKernelSection.lean}
\item \textbf{\nolinkurl{AKS35}}\hypertarget{lh:AKS35}{}\enspace{\ttfamily\nolinkurl{FactMatroid.exists_nonzero_left_component_of_finrank_right_lt}} {\tiny\ttfamily Ssot/\allowbreak AffineKernelSection.lean}
\item \textbf{\nolinkurl{AKS36}}\hypertarget{lh:AKS36}{}\enspace{\ttfamily\nolinkurl{FactMatroid.SupportHitsViews}} {\tiny\ttfamily Ssot/\allowbreak AffineKernelSection.lean}
\item \textbf{\nolinkurl{AKS37}}\hypertarget{lh:AKS37}{}\enspace{\ttfamily\nolinkurl{FactMatroid.viewSupportHittingSubspace_of_kernelAvoiding}} {\tiny\ttfamily Ssot/\allowbreak AffineKernelSection.lean}
\item \textbf{\nolinkurl{AKS38}}\hypertarget{lh:AKS38}{}\enspace{\ttfamily\nolinkurl{FactMatroid.kernelAvoidingSubspace_iff_viewSupportHittingSubspace}} {\tiny\ttfamily Ssot/\allowbreak AffineKernelSection.lean}
\item \textbf{\nolinkurl{CIA3}}\hypertarget{lh:CIA3}{}\enspace{\ttfamily\nolinkurl{FiniteConverse.finite_counting_converse}} {\tiny\ttfamily Ssot/\allowbreak FiniteConverse.lean}
\item \textbf{\nolinkurl{CMP1}}\hypertarget{lh:CMP1}{}\enspace{\ttfamily\nolinkurl{MultiFact.countedMeetWitnessedSearch_spec}} {\tiny\ttfamily Ssot/\allowbreak ComplexityBridge.lean}
\item \textbf{\nolinkurl{CMP2}}\hypertarget{lh:CMP2}{}\enspace{\ttfamily\nolinkurl{MultiFact.countedMeetWitnessedSearch_steps}} {\tiny\ttfamily Ssot/\allowbreak ComplexityBridge.lean}
\item \textbf{\nolinkurl{CMP3}}\hypertarget{lh:CMP3}{}\enspace{\ttfamily\nolinkurl{MultiFact.meetWitnessed_inP_explicit}} {\tiny\ttfamily Ssot/\allowbreak ComplexityBridge.lean}
\item \textbf{\nolinkurl{CMP4}}\hypertarget{lh:CMP4}{}\enspace{\ttfamily\nolinkurl{MultiFact.confusableTransitive_inP_explicit}} {\tiny\ttfamily Ssot/\allowbreak ComplexityBridge.lean}
\item \textbf{\nolinkurl{CMP5}}\hypertarget{lh:CMP5}{}\enspace{\ttfamily\nolinkurl{MultiFact.agreementIntersectionClosed_inP_explicit}} {\tiny\ttfamily Ssot/\allowbreak ComplexityBridge.lean}
\item \textbf{\nolinkurl{GCB3}}\hypertarget{lh:GCB3}{}\enspace{\ttfamily\nolinkurl{GraphicCritical.incidenceValue_edgeDirection}} {\tiny\ttfamily Ssot/\allowbreak GraphicCritical.lean}
\item \textbf{\nolinkurl{GCB4}}\hypertarget{lh:GCB4}{}\enspace{\ttfamily\nolinkurl{GraphicCritical.properFieldColoring_iff_avoidsGraphicEdgeDirections}} {\tiny\ttfamily Ssot/\allowbreak GraphicCritical.lean}
\item \textbf{\nolinkurl{GCB5}}\hypertarget{lh:GCB5}{}\enspace{\ttfamily\nolinkurl{GraphicCritical.exists_properFieldColoring_iff_exists_avoidingGraphicFunctional}} {\tiny\ttfamily Ssot/\allowbreak GraphicCritical.lean}
\item \textbf{\nolinkurl{GCB6}}\hypertarget{lh:GCB6}{}\enspace{\ttfamily\nolinkurl{GraphicCritical.incidenceLinearMap}} {\tiny\ttfamily Ssot/\allowbreak GraphicCritical.lean}
\item \textbf{\nolinkurl{GCB7}}\hypertarget{lh:GCB7}{}\enspace{\ttfamily\nolinkurl{GraphicCritical.incidenceKernel}} {\tiny\ttfamily Ssot/\allowbreak GraphicCritical.lean}
\item \textbf{\nolinkurl{GCB8}}\hypertarget{lh:GCB8}{}\enspace{\ttfamily\nolinkurl{GraphicCritical.edgeDirection_mem_incidenceKernel_iff}} {\tiny\ttfamily Ssot/\allowbreak GraphicCritical.lean}
\item \textbf{\nolinkurl{GCB9}}\hypertarget{lh:GCB9}{}\enspace{\ttfamily\nolinkurl{GraphicCritical.IncidenceKernelAvoidsEdges}} {\tiny\ttfamily Ssot/\allowbreak GraphicCritical.lean}
\item \textbf{\nolinkurl{GCB10}}\hypertarget{lh:GCB10}{}\enspace{\ttfamily\nolinkurl{GraphicCritical.properFieldColoring_iff_incidenceKernelAvoidsEdges}} {\tiny\ttfamily Ssot/\allowbreak GraphicCritical.lean}
\item \textbf{\nolinkurl{GCB11}}\hypertarget{lh:GCB11}{}\enspace{\ttfamily\nolinkurl{GraphicCritical.exists_properFieldColoring_iff_exists_incidenceKernelAvoidsEdges}} {\tiny\ttfamily Ssot/\allowbreak GraphicCritical.lean}
\item \textbf{\nolinkurl{GCB12}}\hypertarget{lh:GCB12}{}\enspace{\ttfamily\nolinkurl{GraphicCritical.graphicQuotientSearch}} {\tiny\ttfamily Ssot/\allowbreak GraphicCritical.lean}
\item \textbf{\nolinkurl{GCB13}}\hypertarget{lh:GCB13}{}\enspace{\ttfamily\nolinkurl{GraphicCritical.graphicQuotientSearch_candidate_iff_incidenceKernelAvoidsEdges}} {\tiny\ttfamily Ssot/\allowbreak GraphicCritical.lean}
\item \textbf{\nolinkurl{GCB14}}\hypertarget{lh:GCB14}{}\enspace{\ttfamily\nolinkurl{GraphicCritical.graphicQuotientSearch_positive_iff_exists_properFieldColoring}} {\tiny\ttfamily Ssot/\allowbreak GraphicCritical.lean}
\item \textbf{\nolinkurl{GCB15}}\hypertarget{lh:GCB15}{}\enspace{\ttfamily\nolinkurl{GraphicCritical.edgeDirection_ne_zero_of_ne}} {\tiny\ttfamily Ssot/\allowbreak GraphicCritical.lean}
\item \textbf{\nolinkurl{GCB16}}\hypertarget{lh:GCB16}{}\enspace{\ttfamily\nolinkurl{GraphicCritical.graphicQuotientSearch_candidate_iff_incidenceKernelAvoidsEdges_of_loopless}} {\tiny\ttfamily Ssot/\allowbreak GraphicCritical.lean}
\item \textbf{\nolinkurl{GCB17}}\hypertarget{lh:GCB17}{}\enspace{\ttfamily\nolinkurl{GraphicCritical.graphicQuotientSearch_positive_iff_exists_properFieldColoring_of_loopless}} {\tiny\ttfamily Ssot/\allowbreak GraphicCritical.lean}
\item \textbf{\nolinkurl{GCB18}}\hypertarget{lh:GCB18}{}\enspace{\ttfamily\nolinkurl{GraphicCritical.edgeLine_inter_eq_bot_of_ne}} {\tiny\ttfamily Ssot/\allowbreak GraphicCritical.lean}
\item \textbf{\nolinkurl{GCB19}}\hypertarget{lh:GCB19}{}\enspace{\ttfamily\nolinkurl{GraphicCritical.SameUndirectedEdge}} {\tiny\ttfamily Ssot/\allowbreak GraphicCritical.lean}
\item \textbf{\nolinkurl{GCB20}}\hypertarget{lh:GCB20}{}\enspace{\ttfamily\nolinkurl{GraphicCritical.sameUndirectedEdge_of_edgeLine_eq}} {\tiny\ttfamily Ssot/\allowbreak GraphicCritical.lean}
\item \textbf{\nolinkurl{GCB21}}\hypertarget{lh:GCB21}{}\enspace{\ttfamily\nolinkurl{GraphicCritical.edgeLine_ne_of_not_sameUndirected}} {\tiny\ttfamily Ssot/\allowbreak GraphicCritical.lean}
\item \textbf{\nolinkurl{GCB22}}\hypertarget{lh:GCB22}{}\enspace{\ttfamily\nolinkurl{GraphicCritical.ProperIndexedFieldColoring}} {\tiny\ttfamily Ssot/\allowbreak GraphicCritical.lean}
\item \textbf{\nolinkurl{GCB23}}\hypertarget{lh:GCB23}{}\enspace{\ttfamily\nolinkurl{GraphicCritical.graphicIndexedQuotientSearch}} {\tiny\ttfamily Ssot/\allowbreak GraphicCritical.lean}
\item \textbf{\nolinkurl{GCB24}}\hypertarget{lh:GCB24}{}\enspace{\ttfamily\nolinkurl{GraphicCritical.graphicIndexedQuotientSearch_candidate_iff}} {\tiny\ttfamily Ssot/\allowbreak GraphicCritical.lean}
\item \textbf{\nolinkurl{GCB25}}\hypertarget{lh:GCB25}{}\enspace{\ttfamily\nolinkurl{GraphicCritical.graphicIndexedQuotientSearch_positive_iff_exists_properIndexedFieldColoring}} {\tiny\ttfamily Ssot/\allowbreak GraphicCritical.lean}
\item \textbf{\nolinkurl{GCB26}}\hypertarget{lh:GCB26}{}\enspace{\ttfamily\nolinkurl{GraphicCritical.properFieldColorings}} {\tiny\ttfamily Ssot/\allowbreak GraphicCritical.lean}
\item \textbf{\nolinkurl{GCB27}}\hypertarget{lh:GCB27}{}\enspace{\ttfamily\nolinkurl{GraphicCritical.mem_properFieldColorings}} {\tiny\ttfamily Ssot/\allowbreak GraphicCritical.lean}
\item \textbf{\nolinkurl{GCB28}}\hypertarget{lh:GCB28}{}\enspace{\ttfamily\nolinkurl{GraphicCritical.properIndexedFieldColorings}} {\tiny\ttfamily Ssot/\allowbreak GraphicCritical.lean}
\item \textbf{\nolinkurl{GCB29}}\hypertarget{lh:GCB29}{}\enspace{\ttfamily\nolinkurl{GraphicCritical.mem_properIndexedFieldColorings}} {\tiny\ttfamily Ssot/\allowbreak GraphicCritical.lean}
\item \textbf{\nolinkurl{GCB30}}\hypertarget{lh:GCB30}{}\enspace{\ttfamily\nolinkurl{GraphicCritical.properFieldColorings_nonempty_iff_exists}} {\tiny\ttfamily Ssot/\allowbreak GraphicCritical.lean}
\item \textbf{\nolinkurl{GCB31}}\hypertarget{lh:GCB31}{}\enspace{\ttfamily\nolinkurl{GraphicCritical.properIndexedFieldColorings_nonempty_iff_exists}} {\tiny\ttfamily Ssot/\allowbreak GraphicCritical.lean}
\item \textbf{\nolinkurl{GCB32}}\hypertarget{lh:GCB32}{}\enspace{\ttfamily\nolinkurl{GraphicCritical.properFieldColorings_card_pos_iff_exists}} {\tiny\ttfamily Ssot/\allowbreak GraphicCritical.lean}
\item \textbf{\nolinkurl{GCB33}}\hypertarget{lh:GCB33}{}\enspace{\ttfamily\nolinkurl{GraphicCritical.properIndexedFieldColorings_card_pos_iff_exists}} {\tiny\ttfamily Ssot/\allowbreak GraphicCritical.lean}
\item \textbf{\nolinkurl{GCB34}}\hypertarget{lh:GCB34}{}\enspace{\ttfamily\nolinkurl{GraphicCritical.graphicQuotientSearch_positive_iff_properFieldColorings_card_pos}} {\tiny\ttfamily Ssot/\allowbreak GraphicCritical.lean}
\item \textbf{\nolinkurl{GCB35}}\hypertarget{lh:GCB35}{}\enspace{\ttfamily\nolinkurl{GraphicCritical.graphicQuotientSearch_positive_iff_properFieldColorings_card_pos_of_loopless}} {\tiny\ttfamily Ssot/\allowbreak GraphicCritical.lean}
\item \textbf{\nolinkurl{GCB36}}\hypertarget{lh:GCB36}{}\enspace{\ttfamily\nolinkurl{GraphicCritical.graphicIndexedQuotientSearch_positive_iff_properIndexedFieldColorings_card_pos}} {\tiny\ttfamily Ssot/\allowbreak GraphicCritical.lean}
\item \textbf{\nolinkurl{GCB37}}\hypertarget{lh:GCB37}{}\enspace{\ttfamily\nolinkurl{GraphicCritical.indexedColoringConflictSet}} {\tiny\ttfamily Ssot/\allowbreak GraphicCritical.lean}
\item \textbf{\nolinkurl{GCB38}}\hypertarget{lh:GCB38}{}\enspace{\ttfamily\nolinkurl{GraphicCritical.indexedColoringViolationProfile}} {\tiny\ttfamily Ssot/\allowbreak GraphicCritical.lean}
\item \textbf{\nolinkurl{GCB39}}\hypertarget{lh:GCB39}{}\enspace{\ttfamily\nolinkurl{GraphicCritical.properIndexedFieldColoring_iff_conflictSet_card_eq_zero}} {\tiny\ttfamily Ssot/\allowbreak GraphicCritical.lean}
\item \textbf{\nolinkurl{GCB40}}\hypertarget{lh:GCB40}{}\enspace{\ttfamily\nolinkurl{GraphicCritical.properIndexedFieldColorings_card_eq_indexedColoringViolationProfile_zero}} {\tiny\ttfamily Ssot/\allowbreak GraphicCritical.lean}
\item \textbf{\nolinkurl{GCB41}}\hypertarget{lh:GCB41}{}\enspace{\ttfamily\nolinkurl{GraphicCritical.properIndexedFieldColorings_card_eq_of_same_indexedColoringViolationProfile}} {\tiny\ttfamily Ssot/\allowbreak GraphicCritical.lean}
\item \textbf{\nolinkurl{GCB42}}\hypertarget{lh:GCB42}{}\enspace{\ttfamily\nolinkurl{GraphicCritical.properIndexedFieldColorings_card_pos_iff_of_same_indexedColoringViolationProfile}} {\tiny\ttfamily Ssot/\allowbreak GraphicCritical.lean}
\item \textbf{\nolinkurl{GCR1}}\hypertarget{lh:GCR1}{}\enspace{\ttfamily\nolinkurl{GraphicCriticalReduction.fixedShellReduction_of_graphicCommonCoreLift}} {\tiny\ttfamily Ssot/\allowbreak GraphicCriticalReduction.lean}
\item \textbf{\nolinkurl{GCR2}}\hypertarget{lh:GCR2}{}\enspace{\ttfamily\nolinkurl{GraphicCriticalReduction.fixedShellReduction_of_looplessGraphicCommonCoreLift}} {\tiny\ttfamily Ssot/\allowbreak GraphicCriticalReduction.lean}
\item \textbf{\nolinkurl{GCR3}}\hypertarget{lh:GCR3}{}\enspace{\ttfamily\nolinkurl{GraphicCriticalReduction.productLift_edgeLines_multiIntersection_eq_productCore}} {\tiny\ttfamily Ssot/\allowbreak GraphicCriticalReduction.lean}
\item \textbf{\nolinkurl{GCR4}}\hypertarget{lh:GCR4}{}\enspace{\ttfamily\nolinkurl{GraphicCriticalReduction.productLift_edgeLines_multiIntersection_eq_between}} {\tiny\ttfamily Ssot/\allowbreak GraphicCriticalReduction.lean}
\item \textbf{\nolinkurl{GCR5}}\hypertarget{lh:GCR5}{}\enspace{\ttfamily\nolinkurl{GraphicCriticalReduction.verticalOrEdgeLift}} {\tiny\ttfamily Ssot/\allowbreak GraphicCriticalReduction.lean}
\item \textbf{\nolinkurl{GCR6}}\hypertarget{lh:GCR6}{}\enspace{\ttfamily\nolinkurl{GraphicCriticalReduction.verticalOrEdgeLift_multiIntersection_eq_productCore}} {\tiny\ttfamily Ssot/\allowbreak GraphicCriticalReduction.lean}
\item \textbf{\nolinkurl{GCR7}}\hypertarget{lh:GCR7}{}\enspace{\ttfamily\nolinkurl{GraphicCriticalReduction.productLift_simpleEdges_multiIntersection_eq_productCore}} {\tiny\ttfamily Ssot/\allowbreak GraphicCriticalReduction.lean}
\item \textbf{\nolinkurl{GCR8}}\hypertarget{lh:GCR8}{}\enspace{\ttfamily\nolinkurl{GraphicCriticalReduction.productLift_simpleEdges_multiIntersection_eq_between}} {\tiny\ttfamily Ssot/\allowbreak GraphicCriticalReduction.lean}
\item \textbf{\nolinkurl{GCR9}}\hypertarget{lh:GCR9}{}\enspace{\ttfamily\nolinkurl{GraphicCriticalReduction.verticalOrEdgeLift_simpleEdges_multiIntersection_eq_productCore}} {\tiny\ttfamily Ssot/\allowbreak GraphicCriticalReduction.lean}
\item \textbf{\nolinkurl{GCR10}}\hypertarget{lh:GCR10}{}\enspace{\ttfamily\nolinkurl{GraphicCriticalReduction.fixedShellReduction_of_indexedLooplessGraphicCommonCoreLift}} {\tiny\ttfamily Ssot/\allowbreak GraphicCriticalReduction.lean}
\item \textbf{\nolinkurl{GCR11}}\hypertarget{lh:GCR11}{}\enspace{\ttfamily\nolinkurl{GraphicCriticalReduction.verticalOrEdgeLift_simpleEdges_multiIntersection_eq_between}} {\tiny\ttfamily Ssot/\allowbreak GraphicCriticalReduction.lean}
\item \textbf{\nolinkurl{GCR12}}\hypertarget{lh:GCR12}{}\enspace{\ttfamily\nolinkurl{GraphicCriticalReduction.liftedVerticalOrKernel}} {\tiny\ttfamily Ssot/\allowbreak GraphicCriticalReduction.lean}
\item \textbf{\nolinkurl{GCR13}}\hypertarget{lh:GCR13}{}\enspace{\ttfamily\nolinkurl{GraphicCriticalReduction.liftedMultiwayShell}} {\tiny\ttfamily Ssot/\allowbreak GraphicCriticalReduction.lean}
\item \textbf{\nolinkurl{GCR14}}\hypertarget{lh:GCR14}{}\enspace{\ttfamily\nolinkurl{GraphicCriticalReduction.liftedVerticalOrKernel_graphicIndexedCommonCoreLift_eq}} {\tiny\ttfamily Ssot/\allowbreak GraphicCriticalReduction.lean}
\item \textbf{\nolinkurl{GCR15}}\hypertarget{lh:GCR15}{}\enspace{\ttfamily\nolinkurl{GraphicCriticalReduction.fixedShellReduction_of_indexedSimpleGraphicCommonCoreLift_multiwayShell}}
\item \textbf{\nolinkurl{GCR16}}\hypertarget{lh:GCR16}{}\enspace{\ttfamily\nolinkurl{GraphicCriticalReduction.not_semanticComplete_of_indexedSimpleGraphicCommonCoreLift_multiwayShell_separates}}
\item \textbf{\nolinkurl{GCR17}}\hypertarget{lh:GCR17}{}\enspace{\ttfamily\nolinkurl{GraphicCriticalReduction.fixedShellReduction_of_indexedSimpleGraphicCommonCoreLift_pointedCore}}
\item \textbf{\nolinkurl{GCR18}}\hypertarget{lh:GCR18}{}\enspace{\ttfamily\nolinkurl{GraphicCriticalReduction.fixedShellReduction_of_indexedSimpleGraphicCommonCoreLift_nonemptyCore}}
\item \textbf{\nolinkurl{GCR19}}\hypertarget{lh:GCR19}{}\enspace{\ttfamily\nolinkurl{GraphicCriticalReduction.not_semanticComplete_of_indexedSimpleGraphicCommonCoreLift_pointedCore_separates}}
\item \textbf{\nolinkurl{GCR20}}\hypertarget{lh:GCR20}{}\enspace{\ttfamily\nolinkurl{GraphicCriticalReduction.not_semanticComplete_of_indexedSimpleGraphicCommonCoreLift_nonemptyCore_separates}}
\item \textbf{\nolinkurl{GCR21}}\hypertarget{lh:GCR21}{}\enspace{\ttfamily\nolinkurl{GraphicCriticalReduction.liftedMultiwayShell_graphicIndexedCommonCoreLift_eq_productCore}} {\tiny\ttfamily Ssot/\allowbreak GraphicCriticalReduction.lean}
\item \textbf{\nolinkurl{GCR22}}\hypertarget{lh:GCR22}{}\enspace{\ttfamily\nolinkurl{GraphicCriticalReduction.liftedMultiwayShell_graphicIndexedCommonCoreLift_eq_between}} {\tiny\ttfamily Ssot/\allowbreak GraphicCriticalReduction.lean}
\item \textbf{\nolinkurl{GCR23}}\hypertarget{lh:GCR23}{}\enspace{\ttfamily\nolinkurl{GraphicCriticalReduction.liftedCandidatePositive}} {\tiny\ttfamily Ssot/\allowbreak GraphicCriticalReduction.lean}
\item \textbf{\nolinkurl{GCR24}}\hypertarget{lh:GCR24}{}\enspace{\ttfamily\nolinkurl{GraphicCriticalReduction.fixedShellWitnessReduction_of_indexedSimpleGraphicCommonCoreLift_pointedCore}}
\item \textbf{\nolinkurl{GCR25}}\hypertarget{lh:GCR25}{}\enspace{\ttfamily\nolinkurl{GraphicCriticalReduction.liftedSearch_positive_iff_exists_liftedCandidatePositive}} {\tiny\ttfamily Ssot/\allowbreak GraphicCriticalReduction.lean}
\item \textbf{\nolinkurl{GCR26}}\hypertarget{lh:GCR26}{}\enspace{\ttfamily\nolinkurl{GraphicCriticalReduction.fixedShellWitnessReduction_of_indexedSimpleGraphicCommonCoreLift_nonemptyCore}}
\item \textbf{\nolinkurl{GCR27}}\hypertarget{lh:GCR27}{}\enspace{\ttfamily\nolinkurl{GraphicCriticalReduction.fixedShellReduction_of_indexedSimpleGraphicCommonCoreLift_nonemptyCore_viaWitness}}
\item \textbf{\nolinkurl{GCR28}}\hypertarget{lh:GCR28}{}\enspace{\ttfamily\nolinkurl{GraphicCriticalReduction.criticalProblemCommonCoreLift_fixedShell_and_positive_iff_pointedCore}} {\tiny\ttfamily Ssot/\allowbreak GraphicCriticalReduction.lean}
\item \textbf{\nolinkurl{GCR29}}\hypertarget{lh:GCR29}{}\enspace{\ttfamily\nolinkurl{GraphicCriticalReduction.criticalProblemCommonCoreLift_fixedShell_and_positive_iff_nonemptyCore}} {\tiny\ttfamily Ssot/\allowbreak GraphicCriticalReduction.lean}
\item \textbf{\nolinkurl{GCR30}}\hypertarget{lh:GCR30}{}\enspace{\ttfamily\nolinkurl{GraphicCriticalReduction.liftedViewVerticalOrKernel}} {\tiny\ttfamily Ssot/\allowbreak GraphicCriticalReduction.lean}
\item \textbf{\nolinkurl{GCR31}}\hypertarget{lh:GCR31}{}\enspace{\ttfamily\nolinkurl{GraphicCriticalReduction.liftedViewMultiwayShell}} {\tiny\ttfamily Ssot/\allowbreak GraphicCriticalReduction.lean}
\item \textbf{\nolinkurl{GCR32}}\hypertarget{lh:GCR32}{}\enspace{\ttfamily\nolinkurl{GraphicCriticalReduction.liftedVerticalOrKernelOn}} {\tiny\ttfamily Ssot/\allowbreak GraphicCriticalReduction.lean}
\item \textbf{\nolinkurl{GCR33}}\hypertarget{lh:GCR33}{}\enspace{\ttfamily\nolinkurl{GraphicCriticalReduction.liftedMultiwayShellOn}} {\tiny\ttfamily Ssot/\allowbreak GraphicCriticalReduction.lean}
\item \textbf{\nolinkurl{GCR34}}\hypertarget{lh:GCR34}{}\enspace{\ttfamily\nolinkurl{GraphicCriticalReduction.liftedViewVerticalOrKernel_graphicIndexedCommonCoreLift_eq}} {\tiny\ttfamily Ssot/\allowbreak GraphicCriticalReduction.lean}
\item \textbf{\nolinkurl{GCR35}}\hypertarget{lh:GCR35}{}\enspace{\ttfamily\nolinkurl{GraphicCriticalReduction.liftedVerticalOrKernelOn_graphicIndexedCommonCoreLift_eq}} {\tiny\ttfamily Ssot/\allowbreak GraphicCriticalReduction.lean}
\item \textbf{\nolinkurl{GCR36}}\hypertarget{lh:GCR36}{}\enspace{\ttfamily\nolinkurl{GraphicCriticalReduction.liftedViewMultiwayShell_graphicIndexedCommonCoreLift_eq_productCore}} {\tiny\ttfamily Ssot/\allowbreak GraphicCriticalReduction.lean}
\item \textbf{\nolinkurl{GCR37}}\hypertarget{lh:GCR37}{}\enspace{\ttfamily\nolinkurl{GraphicCriticalReduction.liftedMultiwayShellOn_graphicIndexedCommonCoreLift_eq_productCore}} {\tiny\ttfamily Ssot/\allowbreak GraphicCriticalReduction.lean}
\item \textbf{\nolinkurl{GCR38}}\hypertarget{lh:GCR38}{}\enspace{\ttfamily\nolinkurl{GraphicCriticalReduction.fixedShellReductionOn_of_indexedSimpleGraphicCommonCoreLift_nonemptyCore}}
\item \textbf{\nolinkurl{GCR39}}\hypertarget{lh:GCR39}{}\enspace{\ttfamily\nolinkurl{GraphicCriticalReduction.criticalProblemCommonCoreLiftOn_fixedShell_and_positive_iff_nonemptyCore}} {\tiny\ttfamily Ssot/\allowbreak GraphicCriticalReduction.lean}
\item \textbf{\nolinkurl{GCR40}}\hypertarget{lh:GCR40}{}\enspace{\ttfamily\nolinkurl{GraphicCriticalReduction.fixedShellWitnessReductionOn_of_indexedSimpleGraphicCommonCoreLift_nonemptyCore}}
\item \textbf{\nolinkurl{GCR41}}\hypertarget{lh:GCR41}{}\enspace{\ttfamily\nolinkurl{GraphicCriticalReduction.fixedShellReductionOn_of_indexedSimpleGraphicCommonCoreLift_nonemptyCore_viaWitness}}
\item \textbf{\nolinkurl{GCR42}}\hypertarget{lh:GCR42}{}\enspace{\ttfamily\nolinkurl{GraphicCriticalReduction.not_semanticCompleteOn_of_indexedSimpleGraphicCommonCoreLift_nonemptyCore_separates}}
\item \textbf{\nolinkurl{GCR43}}\hypertarget{lh:GCR43}{}\enspace{\ttfamily\nolinkurl{GraphicCriticalReduction.not_semanticCompleteOn_existsCandidate_of_indexedSimpleGraphicCommonCoreLift_nonemptyCore_separates}}
\item \textbf{\nolinkurl{GCR44}}\hypertarget{lh:GCR44}{}\enspace{\ttfamily\nolinkurl{GraphicCriticalReduction.exists_sameActiveShell_diffPositive_of_indexedSimpleGraphicCommonCoreLift_nonemptyCore}}
\item \textbf{\nolinkurl{GCR45}}\hypertarget{lh:GCR45}{}\enspace{\ttfamily\nolinkurl{GraphicCriticalReduction.exists_sameActiveShell_diffExistsCandidate_of_indexedSimpleGraphicCommonCoreLift_nonemptyCore}}
\item \textbf{\nolinkurl{GCR46}}\hypertarget{lh:GCR46}{}\enspace{\ttfamily\nolinkurl{GraphicCriticalReduction.fixedShellReductionOn_of_sourceToIndexedColoring_nonemptyCore}}
\item \textbf{\nolinkurl{GCR47}}\hypertarget{lh:GCR47}{}\enspace{\ttfamily\nolinkurl{GraphicCriticalReduction.fixedShellWitnessReductionOn_of_sourceToIndexedColoring_nonemptyCore}}
\item \textbf{\nolinkurl{GCR48}}\hypertarget{lh:GCR48}{}\enspace{\ttfamily\nolinkurl{GraphicCriticalReduction.fixedShellReductionOn_of_sourceToIndexedColoring_nonemptyCore_viaWitness}}
\item \textbf{\nolinkurl{GCR49}}\hypertarget{lh:GCR49}{}\enspace{\ttfamily\nolinkurl{GraphicCriticalReduction.criticalProblemCommonCoreLiftOn_fixedShell_and_positive_iff_sourcePred}} {\tiny\ttfamily Ssot/\allowbreak GraphicCriticalReduction.lean}
\item \textbf{\nolinkurl{GCR50}}\hypertarget{lh:GCR50}{}\enspace{\ttfamily\nolinkurl{GraphicCriticalReduction.not_semanticCompleteOn_of_sourceToIndexedColoring_nonemptyCore_separates}}
\item \textbf{\nolinkurl{GCR51}}\hypertarget{lh:GCR51}{}\enspace{\ttfamily\nolinkurl{GraphicCriticalReduction.not_semanticCompleteOn_existsCandidate_of_sourceToIndexedColoring_nonemptyCore_separates}}
\item \textbf{\nolinkurl{GCR52}}\hypertarget{lh:GCR52}{}\enspace{\ttfamily\nolinkurl{GraphicCriticalReduction.exists_sameActiveShell_diffPositive_of_sourceToIndexedColoring_nonemptyCore}}
\item \textbf{\nolinkurl{GCR53}}\hypertarget{lh:GCR53}{}\enspace{\ttfamily\nolinkurl{GraphicCriticalReduction.exists_sameActiveShell_diffExistsCandidate_of_sourceToIndexedColoring_nonemptyCore}}
\item \textbf{\nolinkurl{GCR54}}\hypertarget{lh:GCR54}{}\enspace{\ttfamily\nolinkurl{GraphicCriticalReduction.IndexedColoringSourceReduction}} {\tiny\ttfamily Ssot/\allowbreak GraphicCriticalReduction.lean}
\item \textbf{\nolinkurl{GCR55}}\hypertarget{lh:GCR55}{}\enspace{\ttfamily\nolinkurl{GraphicCriticalReduction.IndexedColoringWitnessReduction}} {\tiny\ttfamily Ssot/\allowbreak GraphicCriticalReduction.lean}
\item \textbf{\nolinkurl{GCR56}}\hypertarget{lh:GCR56}{}\enspace{\ttfamily\nolinkurl{GraphicCriticalReduction.fixedShellReductionOn_of_indexedColoringSourceReduction_nonemptyCore}}
\item \textbf{\nolinkurl{GCR57}}\hypertarget{lh:GCR57}{}\enspace{\ttfamily\nolinkurl{GraphicCriticalReduction.fixedShellWitnessReductionOn_of_indexedColoringWitnessReduction_nonemptyCore}}
\item \textbf{\nolinkurl{GCR58}}\hypertarget{lh:GCR58}{}\enspace{\ttfamily\nolinkurl{GraphicCriticalReduction.fixedShellReductionOn_of_indexedColoringWitnessReduction_nonemptyCore_viaWitness}}
\item \textbf{\nolinkurl{GCR59}}\hypertarget{lh:GCR59}{}\enspace{\ttfamily\nolinkurl{GraphicCriticalReduction.criticalProblemCommonCoreLiftOn_fixedShell_and_positive_iff_sourceReduction}} {\tiny\ttfamily Ssot/\allowbreak GraphicCriticalReduction.lean}
\item \textbf{\nolinkurl{GCR60}}\hypertarget{lh:GCR60}{}\enspace{\ttfamily\nolinkurl{GraphicCriticalReduction.not_semanticCompleteOn_of_indexedColoringSourceReduction_nonemptyCore_separates}}
\item \textbf{\nolinkurl{GCR61}}\hypertarget{lh:GCR61}{}\enspace{\ttfamily\nolinkurl{GraphicCriticalReduction.not_semanticCompleteOn_existsCandidate_of_indexedColoringWitnessReduction_nonemptyCore_separates}}
\item \textbf{\nolinkurl{GCR62}}\hypertarget{lh:GCR62}{}\enspace{\ttfamily\nolinkurl{GraphicCriticalReduction.exists_sameActiveShell_diffPositive_of_indexedColoringSourceReduction_nonemptyCore}}
\item \textbf{\nolinkurl{GCR63}}\hypertarget{lh:GCR63}{}\enspace{\ttfamily\nolinkurl{GraphicCriticalReduction.exists_sameActiveShell_diffExistsCandidate_of_indexedColoringWitnessReduction_nonemptyCore}}
\item \textbf{\nolinkurl{GCR64}}\hypertarget{lh:GCR64}{}\enspace{\ttfamily\nolinkurl{GraphicCriticalReduction.IndexedColoringSourceReduction.precompose}} {\tiny\ttfamily Ssot/\allowbreak GraphicCriticalReduction.lean}
\item \textbf{\nolinkurl{GCR65}}\hypertarget{lh:GCR65}{}\enspace{\ttfamily\nolinkurl{GraphicCriticalReduction.IndexedColoringWitnessReduction.precompose}} {\tiny\ttfamily Ssot/\allowbreak GraphicCriticalReduction.lean}
\item \textbf{\nolinkurl{GCR66}}\hypertarget{lh:GCR66}{}\enspace{\ttfamily\nolinkurl{GraphicCriticalReduction.IndexedColoringWitnessReduction.toSourceReduction}} {\tiny\ttfamily Ssot/\allowbreak GraphicCriticalReduction.lean}
\item \textbf{\nolinkurl{GCR67}}\hypertarget{lh:GCR67}{}\enspace{\ttfamily\nolinkurl{GraphicCriticalReduction.IndexedColoringSourceCriticalBridgeCertificate}} {\tiny\ttfamily Ssot/\allowbreak GraphicCriticalReduction.lean}
\item \textbf{\nolinkurl{GCR68}}\hypertarget{lh:GCR68}{}\enspace{\ttfamily\nolinkurl{GraphicCriticalReduction.indexedColoringSourceCriticalBridgeCertificate}} {\tiny\ttfamily Ssot/\allowbreak GraphicCriticalReduction.lean}
\item \textbf{\nolinkurl{GCR69}}\hypertarget{lh:GCR69}{}\enspace{\ttfamily\nolinkurl{GraphicCriticalReduction.IndexedColoringWitnessCriticalBridgeCertificate}} {\tiny\ttfamily Ssot/\allowbreak GraphicCriticalReduction.lean}
\item \textbf{\nolinkurl{GCR70}}\hypertarget{lh:GCR70}{}\enspace{\ttfamily\nolinkurl{GraphicCriticalReduction.indexedColoringWitnessCriticalBridgeCertificate}} {\tiny\ttfamily Ssot/\allowbreak GraphicCriticalReduction.lean}
\item \textbf{\nolinkurl{GCR71}}\hypertarget{lh:GCR71}{}\enspace{\ttfamily\nolinkurl{GraphicCriticalReduction.not_semanticCompleteOn_of_indexedColoringSourceCriticalBridgeCertificate}} {\tiny\ttfamily Ssot/\allowbreak GraphicCriticalReduction.lean}
\item \textbf{\nolinkurl{GCR72}}\hypertarget{lh:GCR72}{}\enspace{\ttfamily\nolinkurl{GraphicCriticalReduction.exists_sameActiveShell_diffPositive_of_indexedColoringSourceCriticalBridgeCertificate}} {\tiny\ttfamily Ssot/\allowbreak GraphicCriticalReduction.lean}
\item \textbf{\nolinkurl{GCR73}}\hypertarget{lh:GCR73}{}\enspace{\ttfamily\nolinkurl{GraphicCriticalReduction.not_semanticCompleteOn_existsCandidate_of_indexedColoringWitnessCriticalBridgeCertificate}} {\tiny\ttfamily Ssot/\allowbreak GraphicCriticalReduction.lean}
\item \textbf{\nolinkurl{GCR74}}\hypertarget{lh:GCR74}{}\enspace{\ttfamily\nolinkurl{GraphicCriticalReduction.not_semanticCompleteOn_positive_of_indexedColoringWitnessCriticalBridgeCertificate}} {\tiny\ttfamily Ssot/\allowbreak GraphicCriticalReduction.lean}
\item \textbf{\nolinkurl{GCR75}}\hypertarget{lh:GCR75}{}\enspace{\ttfamily\nolinkurl{GraphicCriticalReduction.exists_sameActiveShell_diffExistsCandidate_of_indexedColoringWitnessCriticalBridgeCertificate}} {\tiny\ttfamily Ssot/\allowbreak GraphicCriticalReduction.lean}
\item \textbf{\nolinkurl{GCR76}}\hypertarget{lh:GCR76}{}\enspace{\ttfamily\nolinkurl{GraphicCriticalReduction.IndexedColoringSourceCriticalBridgeCertificate.precompose}} {\tiny\ttfamily Ssot/\allowbreak GraphicCriticalReduction.lean}
\item \textbf{\nolinkurl{GCR77}}\hypertarget{lh:GCR77}{}\enspace{\ttfamily\nolinkurl{GraphicCriticalReduction.IndexedColoringWitnessCriticalBridgeCertificate.precompose}} {\tiny\ttfamily Ssot/\allowbreak GraphicCriticalReduction.lean}
\item \textbf{\nolinkurl{GCR78}}\hypertarget{lh:GCR78}{}\enspace{\ttfamily\nolinkurl{GraphicCriticalReduction.criticalProblemCommonCoreLiftOn_fixedShell_and_positive_iff_coloringCountPos}} {\tiny\ttfamily Ssot/\allowbreak GraphicCriticalReduction.lean}
\item \textbf{\nolinkurl{GCR79}}\hypertarget{lh:GCR79}{}\enspace{\ttfamily\nolinkurl{GraphicCriticalReduction.criticalProblemCommonCoreLiftOn_candidate_iff_properIndexedFieldColoring}} {\tiny\ttfamily Ssot/\allowbreak GraphicCriticalReduction.lean}
\item \textbf{\nolinkurl{GCR80}}\hypertarget{lh:GCR80}{}\enspace{\ttfamily\nolinkurl{GraphicCriticalReduction.liftedPositiveCandidates}} {\tiny\ttfamily Ssot/\allowbreak GraphicCriticalReduction.lean}
\item \textbf{\nolinkurl{GCR81}}\hypertarget{lh:GCR81}{}\enspace{\ttfamily\nolinkurl{GraphicCriticalReduction.mem_liftedPositiveCandidates}} {\tiny\ttfamily Ssot/\allowbreak GraphicCriticalReduction.lean}
\item \textbf{\nolinkurl{GCR82}}\hypertarget{lh:GCR82}{}\enspace{\ttfamily\nolinkurl{GraphicCriticalReduction.liftedPositiveCandidates_graphicIndexedCommonCoreLift_eq_properIndexedFieldColorings}} {\tiny\ttfamily Ssot/\allowbreak GraphicCriticalReduction.lean}
\item \textbf{\nolinkurl{GCR83}}\hypertarget{lh:GCR83}{}\enspace{\ttfamily\nolinkurl{GraphicCriticalReduction.liftedPositiveCandidates_graphicIndexedCommonCoreLift_card_eq_properIndexedFieldColorings}} {\tiny\ttfamily Ssot/\allowbreak GraphicCriticalReduction.lean}
\item \textbf{\nolinkurl{GCR84}}\hypertarget{lh:GCR84}{}\enspace{\ttfamily\nolinkurl{GraphicCriticalReduction.finiteSourceWitnesses}} {\tiny\ttfamily Ssot/\allowbreak GraphicCriticalReduction.lean}
\item \textbf{\nolinkurl{GCR85}}\hypertarget{lh:GCR85}{}\enspace{\ttfamily\nolinkurl{GraphicCriticalReduction.mem_finiteSourceWitnesses}} {\tiny\ttfamily Ssot/\allowbreak GraphicCriticalReduction.lean}
\item \textbf{\nolinkurl{GCR86}}\hypertarget{lh:GCR86}{}\enspace{\ttfamily\nolinkurl{GraphicCriticalReduction.finiteSourceWitnesses_eq_properIndexedFieldColorings_of_indexedColoringWitnessReduction}} {\tiny\ttfamily Ssot/\allowbreak GraphicCriticalReduction.lean}
\item \textbf{\nolinkurl{GCR87}}\hypertarget{lh:GCR87}{}\enspace{\ttfamily\nolinkurl{GraphicCriticalReduction.liftedPositiveCandidates_graphicIndexedCommonCoreLift_card_eq_finiteSourceWitnesses}} {\tiny\ttfamily Ssot/\allowbreak GraphicCriticalReduction.lean}
\item \textbf{\nolinkurl{GCR88}}\hypertarget{lh:GCR88}{}\enspace{\ttfamily\nolinkurl{GraphicCriticalReduction.fixedShellTargetReductionOn_count_of_indexedColoringWitnessReduction_nonemptyCore}}
\item \textbf{\nolinkurl{GCR89}}\hypertarget{lh:GCR89}{}\enspace{\ttfamily\nolinkurl{GraphicCriticalReduction.IndexedColoringCountCriticalBridgeCertificate}} {\tiny\ttfamily Ssot/\allowbreak GraphicCriticalReduction.lean}
\item \textbf{\nolinkurl{GCR90}}\hypertarget{lh:GCR90}{}\enspace{\ttfamily\nolinkurl{GraphicCriticalReduction.indexedColoringCountCriticalBridgeCertificate}} {\tiny\ttfamily Ssot/\allowbreak GraphicCriticalReduction.lean}
\item \textbf{\nolinkurl{GCR91}}\hypertarget{lh:GCR91}{}\enspace{\ttfamily\nolinkurl{GraphicCriticalReduction.not_semanticFinerOn_count_of_indexedColoringCountCriticalBridgeCertificate}} {\tiny\ttfamily Ssot/\allowbreak GraphicCriticalReduction.lean}
\item \textbf{\nolinkurl{GCR92}}\hypertarget{lh:GCR92}{}\enspace{\ttfamily\nolinkurl{GraphicCriticalReduction.exists_sameActiveShell_diffCount_of_indexedColoringCountCriticalBridgeCertificate}} {\tiny\ttfamily Ssot/\allowbreak GraphicCriticalReduction.lean}
\item \textbf{\nolinkurl{GCR93}}\hypertarget{lh:GCR93}{}\enspace{\ttfamily\nolinkurl{GraphicCriticalReduction.IndexedColoringCountCriticalBridgeCertificate.precompose}} {\tiny\ttfamily Ssot/\allowbreak GraphicCriticalReduction.lean}
\item \textbf{\nolinkurl{GCR94}}\hypertarget{lh:GCR94}{}\enspace{\ttfamily\nolinkurl{GraphicCriticalReduction.not_semanticCompleteOn_countValue_of_indexedColoringCountCriticalBridgeCertificate}} {\tiny\ttfamily Ssot/\allowbreak GraphicCriticalReduction.lean}
\item \textbf{\nolinkurl{GCR95}}\hypertarget{lh:GCR95}{}\enspace{\ttfamily\nolinkurl{GraphicCriticalReduction.not_forall_countValue_semanticComplete_of_indexedColoringCountCriticalBridgeCertificate}} {\tiny\ttfamily Ssot/\allowbreak GraphicCriticalReduction.lean}
\item \textbf{\nolinkurl{GCR96}}\hypertarget{lh:GCR96}{}\enspace{\ttfamily\nolinkurl{GraphicCriticalReduction.sourceWitnessCount_mem_fixedActiveShellCountFiber_of_indexedColoringCountCriticalBridgeCertificate}} {\tiny\ttfamily Ssot/\allowbreak GraphicCriticalReduction.lean}
\item \textbf{\nolinkurl{GCR97}}\hypertarget{lh:GCR97}{}\enspace{\ttfamily\nolinkurl{GraphicCriticalReduction.sourceWitnessCount_range_subset_fixedActiveShellCountFiber_of_indexedColoringCountCriticalBridgeCertificate}} {\tiny\ttfamily Ssot/\allowbreak GraphicCriticalReduction.lean}
\item \textbf{\nolinkurl{GCR98}}\hypertarget{lh:GCR98}{}\enspace{\ttfamily\nolinkurl{GraphicCriticalReduction.exists_liftedSearch_with_fixedActiveShell_and_count_of_sourceWitnessCount}} {\tiny\ttfamily Ssot/\allowbreak GraphicCriticalReduction.lean}
\item \textbf{\nolinkurl{GCR99}}\hypertarget{lh:GCR99}{}\enspace{\ttfamily\nolinkurl{GraphicCriticalReduction.finiteSourceWitnessCountValues}} {\tiny\ttfamily Ssot/\allowbreak GraphicCriticalReduction.lean}
\item \textbf{\nolinkurl{GCR100}}\hypertarget{lh:GCR100}{}\enspace{\ttfamily\nolinkurl{GraphicCriticalReduction.mem_finiteSourceWitnessCountValues}} {\tiny\ttfamily Ssot/\allowbreak GraphicCriticalReduction.lean}
\item \textbf{\nolinkurl{GCR101}}\hypertarget{lh:GCR101}{}\enspace{\ttfamily\nolinkurl{GraphicCriticalReduction.finiteSourceWitnessCountValues_subset_fixedActiveShellCountFiber_of_indexedColoringCountCriticalBridgeCertificate}} {\tiny\ttfamily Ssot/\allowbreak GraphicCriticalReduction.lean}
\item \textbf{\nolinkurl{GCR102}}\hypertarget{lh:GCR102}{}\enspace{\ttfamily\nolinkurl{GraphicCriticalReduction.finiteSourceWitnessCountValues_card_eq_of_count_injOn}} {\tiny\ttfamily Ssot/\allowbreak GraphicCriticalReduction.lean}
\item \textbf{\nolinkurl{GCR103}}\hypertarget{lh:GCR103}{}\enspace{\ttfamily\nolinkurl{GraphicCriticalReduction.finiteSourceWitnessCountValues_card_le_one_of_activeShell_refines_count}} {\tiny\ttfamily Ssot/\allowbreak GraphicCriticalReduction.lean}
\item \textbf{\nolinkurl{GCR104}}\hypertarget{lh:GCR104}{}\enspace{\ttfamily\nolinkurl{GraphicCriticalReduction.not_activeShell_refines_count_of_finiteSourceWitnessCountValues_card_gt_one}} {\tiny\ttfamily Ssot/\allowbreak GraphicCriticalReduction.lean}
\item \textbf{\nolinkurl{GCR105}}\hypertarget{lh:GCR105}{}\enspace{\ttfamily\nolinkurl{GraphicCriticalReduction.not_activeShell_refines_count_of_count_injOn_source_card_gt_one}} {\tiny\ttfamily Ssot/\allowbreak GraphicCriticalReduction.lean}
\item \textbf{\nolinkurl{GCR106}}\hypertarget{lh:GCR106}{}\enspace{\ttfamily\nolinkurl{GraphicCriticalReduction.exists_two_counts_in_fixedActiveShellCountFiber_of_finiteSourceWitnessCountValues_card_gt_one}} {\tiny\ttfamily Ssot/\allowbreak GraphicCriticalReduction.lean}
\item \textbf{\nolinkurl{GCR107}}\hypertarget{lh:GCR107}{}\enspace{\ttfamily\nolinkurl{GraphicCriticalReduction.not_fixedActiveShellCountFiber_subsingleton_of_finiteSourceWitnessCountValues_card_gt_one}} {\tiny\ttfamily Ssot/\allowbreak GraphicCriticalReduction.lean}
\item \textbf{\nolinkurl{GCR108}}\hypertarget{lh:GCR108}{}\enspace{\ttfamily\nolinkurl{GraphicCriticalReduction.not_fixedActiveShellCountFiber_subsingleton_of_count_injOn_source_card_gt_one}} {\tiny\ttfamily Ssot/\allowbreak GraphicCriticalReduction.lean}
\item \textbf{\nolinkurl{GCR109}}\hypertarget{lh:GCR109}{}\enspace{\ttfamily\nolinkurl{GraphicCriticalReduction.not_semanticCompleteOn_positive_of_indexedColoringSourceReduction_nonemptyCore}} {\tiny\ttfamily Ssot/\allowbreak GraphicCriticalReduction.lean}
\item \textbf{\nolinkurl{GCR110}}\hypertarget{lh:GCR110}{}\enspace{\ttfamily\nolinkurl{GraphicCriticalReduction.not_semanticCompleteOn_positive_of_indexedColoringWitnessReduction_nonemptyCore}} {\tiny\ttfamily Ssot/\allowbreak GraphicCriticalReduction.lean}
\item \textbf{\nolinkurl{GCR111}}\hypertarget{lh:GCR111}{}\enspace{\ttfamily\nolinkurl{GraphicCriticalReduction.not_fixedActiveShellCountFiber_subsingleton_of_indexedColoringWitnessReduction_count_injOn_source_card_gt_one}} {\tiny\ttfamily Ssot/\allowbreak GraphicCriticalReduction.lean}
\item \textbf{\nolinkurl{GCR112}}\hypertarget{lh:GCR112}{}\enspace{\ttfamily\nolinkurl{GraphicCriticalReduction.not_activeShell_refines_count_of_indexedColoringWitnessReduction_count_injOn_source_card_gt_one}} {\tiny\ttfamily Ssot/\allowbreak GraphicCriticalReduction.lean}
\item \textbf{\nolinkurl{GCR113}}\hypertarget{lh:GCR113}{}\enspace{\ttfamily\nolinkurl{GraphicCriticalReduction.exists_countValues_subset_fixedActiveShellCountFiber_card_eq_of_indexedColoringCountCriticalBridgeCertificate}} {\tiny\ttfamily Ssot/\allowbreak GraphicCriticalReduction.lean}
\item \textbf{\nolinkurl{GCR114}}\hypertarget{lh:GCR114}{}\enspace{\ttfamily\nolinkurl{GraphicCriticalReduction.exists_countValues_subset_fixedActiveShellCountFiber_card_eq_of_indexedColoringWitnessReduction}} {\tiny\ttfamily Ssot/\allowbreak GraphicCriticalReduction.lean}
\item \textbf{\nolinkurl{GCR115}}\hypertarget{lh:GCR115}{}\enspace{\ttfamily\nolinkurl{GraphicCriticalReduction.activeShell_refines_count_or_countFiberCollision_of_indexedColoringCountCriticalBridgeCertificate}} {\tiny\ttfamily Ssot/\allowbreak GraphicCriticalReduction.lean}
\item \textbf{\nolinkurl{GCR116}}\hypertarget{lh:GCR116}{}\enspace{\ttfamily\nolinkurl{GraphicCriticalReduction.liftedPositiveCandidates_card_pos_iff_liftedSearch_positive}} {\tiny\ttfamily Ssot/\allowbreak GraphicCriticalReduction.lean}
\item \textbf{\nolinkurl{GCR117}}\hypertarget{lh:GCR117}{}\enspace{\ttfamily\nolinkurl{GraphicCriticalReduction.finiteSourceWitnesses_card_pos_iff_exists}} {\tiny\ttfamily Ssot/\allowbreak GraphicCriticalReduction.lean}
\item \textbf{\nolinkurl{GCR118}}\hypertarget{lh:GCR118}{}\enspace{\ttfamily\nolinkurl{GraphicCriticalReduction.fixedShellReductionOn_countPositive_of_indexedColoringWitnessReduction_nonemptyCore}}
\item \textbf{\nolinkurl{GCR119}}\hypertarget{lh:GCR119}{}\enspace{\ttfamily\nolinkurl{GraphicCriticalReduction.fixedShellReductionOn_countPositiveExists_of_indexedColoringWitnessReduction_nonemptyCore}}
\item \textbf{\nolinkurl{GCR120}}\hypertarget{lh:GCR120}{}\enspace{\ttfamily\nolinkurl{GraphicCriticalReduction.liftedPositiveCandidates_graphicIndexedCommonCoreLift_card_pos_iff_exists_sourceWitness}} {\tiny\ttfamily Ssot/\allowbreak GraphicCriticalReduction.lean}
\item \textbf{\nolinkurl{GCR121}}\hypertarget{lh:GCR121}{}\enspace{\ttfamily\nolinkurl{GraphicCriticalReduction.criticalProblemCommonCoreLiftOn_fixedShell_and_countPositive_iff_sourceReduction}} {\tiny\ttfamily Ssot/\allowbreak GraphicCriticalReduction.lean}
\item \textbf{\nolinkurl{GCR122}}\hypertarget{lh:GCR122}{}\enspace{\ttfamily\nolinkurl{GraphicCriticalReduction.not_semanticCompleteOn_countPositive_of_indexedColoringWitnessReduction_nonemptyCore}} {\tiny\ttfamily Ssot/\allowbreak GraphicCriticalReduction.lean}
\item \textbf{\nolinkurl{GCR123}}\hypertarget{lh:GCR123}{}\enspace{\ttfamily\nolinkurl{GraphicCriticalReduction.exists_sameActiveShell_diffCountPositive_of_indexedColoringWitnessReduction_nonemptyCore}} {\tiny\ttfamily Ssot/\allowbreak GraphicCriticalReduction.lean}
\item \textbf{\nolinkurl{GCR124}}\hypertarget{lh:GCR124}{}\enspace{\ttfamily\nolinkurl{GraphicCriticalReduction.criticalProblemCommonCoreLiftOn_fixedShell_and_candidateCount_eq_coloringCount}} {\tiny\ttfamily Ssot/\allowbreak GraphicCriticalReduction.lean}
\item \textbf{\nolinkurl{GCR125}}\hypertarget{lh:GCR125}{}\enspace{\ttfamily\nolinkurl{GraphicCriticalReduction.criticalProblemCommonCoreLiftOn_fixedShell_and_candidateCountPos_iff_coloringCountPos}} {\tiny\ttfamily Ssot/\allowbreak GraphicCriticalReduction.lean}
\item \textbf{\nolinkurl{GCR126}}\hypertarget{lh:GCR126}{}\enspace{\ttfamily\nolinkurl{GraphicCriticalReduction.fixedShellTargetReductionOn_coloringCount_of_indexedSimpleGraphicCommonCoreLift_nonemptyCore}}
\item \textbf{\nolinkurl{GCR127}}\hypertarget{lh:GCR127}{}\enspace{\ttfamily\nolinkurl{GraphicCriticalReduction.exists_sameActiveShell_diffColoringCount_of_indexedSimpleGraphicCommonCoreLift}} {\tiny\ttfamily Ssot/\allowbreak GraphicCriticalReduction.lean}
\item \textbf{\nolinkurl{GCR128}}\hypertarget{lh:GCR128}{}\enspace{\ttfamily\nolinkurl{GraphicCriticalReduction.exists_countValues_subset_fixedActiveShellColoringCountFiber_of_indexedSimpleGraphicCommonCoreLift}} {\tiny\ttfamily Ssot/\allowbreak GraphicCriticalReduction.lean}
\item \textbf{\nolinkurl{GCR129}}\hypertarget{lh:GCR129}{}\enspace{\ttfamily\nolinkurl{GraphicCriticalReduction.not_activeShell_refines_coloringCount_of_indexedSimpleGraphicCommonCoreLift_count_injOn_source_card_gt_one}} {\tiny\ttfamily Ssot/\allowbreak GraphicCriticalReduction.lean}
\item \textbf{\nolinkurl{GCR130}}\hypertarget{lh:GCR130}{}\enspace{\ttfamily\nolinkurl{GraphicCriticalReduction.finiteSourceWitnessCountValues_subset_fixedActiveShellCountFiber_of_indexedColoringWitnessReduction}} {\tiny\ttfamily Ssot/\allowbreak GraphicCriticalReduction.lean}
\item \textbf{\nolinkurl{GCR131}}\hypertarget{lh:GCR131}{}\enspace{\ttfamily\nolinkurl{GraphicCriticalReduction.not_fixedActiveShellCountFiber_subsingleton_of_indexedColoringWitnessReduction_finiteSourceWitnessCountValues_card_gt_one}} {\tiny\ttfamily Ssot/\allowbreak GraphicCriticalReduction.lean}
\item \textbf{\nolinkurl{GCR132}}\hypertarget{lh:GCR132}{}\enspace{\ttfamily\nolinkurl{GraphicCriticalReduction.not_activeShell_refines_count_of_indexedColoringWitnessReduction_finiteSourceWitnessCountValues_card_gt_one}} {\tiny\ttfamily Ssot/\allowbreak GraphicCriticalReduction.lean}
\item \textbf{\nolinkurl{GCR133}}\hypertarget{lh:GCR133}{}\enspace{\ttfamily\nolinkurl{GraphicCriticalReduction.activeShell_refines_count_or_countFiberCollision_of_indexedColoringWitnessReduction_nonemptyCore}} {\tiny\ttfamily Ssot/\allowbreak GraphicCriticalReduction.lean}
\item \textbf{\nolinkurl{GCR134}}\hypertarget{lh:GCR134}{}\enspace{\ttfamily\nolinkurl{GraphicCriticalReduction.finiteSourceColoringCountValues}} {\tiny\ttfamily Ssot/\allowbreak GraphicCriticalReduction.lean}
\item \textbf{\nolinkurl{GCR135}}\hypertarget{lh:GCR135}{}\enspace{\ttfamily\nolinkurl{GraphicCriticalReduction.mem_finiteSourceColoringCountValues}} {\tiny\ttfamily Ssot/\allowbreak GraphicCriticalReduction.lean}
\item \textbf{\nolinkurl{GCR136}}\hypertarget{lh:GCR136}{}\enspace{\ttfamily\nolinkurl{GraphicCriticalReduction.finiteSourceColoringCountValues_subset_fixedActiveShellColoringCountFiber_of_indexedSimpleGraphicCommonCoreLift}} {\tiny\ttfamily Ssot/\allowbreak GraphicCriticalReduction.lean}
\item \textbf{\nolinkurl{GCR137}}\hypertarget{lh:GCR137}{}\enspace{\ttfamily\nolinkurl{GraphicCriticalReduction.not_fixedActiveShellColoringCountFiber_subsingleton_of_indexedSimpleGraphicCommonCoreLift_finiteSourceColoringCountValues_card_gt_one}} {\tiny\ttfamily Ssot/\allowbreak GraphicCriticalReduction.lean}
\item \textbf{\nolinkurl{GCR138}}\hypertarget{lh:GCR138}{}\enspace{\ttfamily\nolinkurl{GraphicCriticalReduction.not_activeShell_refines_coloringCount_of_indexedSimpleGraphicCommonCoreLift_finiteSourceColoringCountValues_card_gt_one}} {\tiny\ttfamily Ssot/\allowbreak GraphicCriticalReduction.lean}
\item \textbf{\nolinkurl{GCR139}}\hypertarget{lh:GCR139}{}\enspace{\ttfamily\nolinkurl{GraphicCriticalReduction.activeShell_refines_coloringCount_or_countFiberCollision_of_indexedSimpleGraphicCommonCoreLift}} {\tiny\ttfamily Ssot/\allowbreak GraphicCriticalReduction.lean}
\item \textbf{\nolinkurl{GCR140}}\hypertarget{lh:GCR140}{}\enspace{\ttfamily\nolinkurl{GraphicCriticalReduction.mem_z3K5SevenEdgeColoringsExec}} {\tiny\ttfamily Ssot/\allowbreak GraphicCriticalReduction.lean}
\item \textbf{\nolinkurl{GCR141}}\hypertarget{lh:GCR141}{}\enspace{\ttfamily\nolinkurl{GraphicCriticalReduction.z3K5SevenEdgeColoringsExec_eq_properIndexedFieldColorings}} {\tiny\ttfamily Ssot/\allowbreak GraphicCriticalReduction.lean}
\item \textbf{\nolinkurl{GCR142}}\hypertarget{lh:GCR142}{}\enspace{\ttfamily\nolinkurl{GraphicCriticalReduction.k5SevenEdgeColorable_coloringCount}} {\tiny\ttfamily Ssot/\allowbreak GraphicCriticalReduction.lean}
\item \textbf{\nolinkurl{GCR143}}\hypertarget{lh:GCR143}{}\enspace{\ttfamily\nolinkurl{GraphicCriticalReduction.k5SevenEdgeNoncolorable_coloringCount}} {\tiny\ttfamily Ssot/\allowbreak GraphicCriticalReduction.lean}
\item \textbf{\nolinkurl{GCR144}}\hypertarget{lh:GCR144}{}\enspace{\ttfamily\nolinkurl{GraphicCriticalReduction.k5SevenEdge_loopless}} {\tiny\ttfamily Ssot/\allowbreak GraphicCriticalReduction.lean}
\item \textbf{\nolinkurl{GCR145}}\hypertarget{lh:GCR145}{}\enspace{\ttfamily\nolinkurl{GraphicCriticalReduction.k5SevenEdge_simple}} {\tiny\ttfamily Ssot/\allowbreak GraphicCriticalReduction.lean}
\item \textbf{\nolinkurl{GCR146}}\hypertarget{lh:GCR146}{}\enspace{\ttfamily\nolinkurl{GraphicCriticalReduction.k5SevenEdgeColoringCountValues_card_gt_one}} {\tiny\ttfamily Ssot/\allowbreak GraphicCriticalReduction.lean}
\item \textbf{\nolinkurl{GCR147}}\hypertarget{lh:GCR147}{}\enspace{\ttfamily\nolinkurl{GraphicCriticalReduction.k5SevenEdge_fixedActiveShellColoringCountFiber_not_subsingleton}} {\tiny\ttfamily Ssot/\allowbreak GraphicCriticalReduction.lean}
\item \textbf{\nolinkurl{GCR148}}\hypertarget{lh:GCR148}{}\enspace{\ttfamily\nolinkurl{GraphicCriticalReduction.k5SevenEdge_activeShell_does_not_refine_coloringCount}} {\tiny\ttfamily Ssot/\allowbreak GraphicCriticalReduction.lean}
\item \textbf{\nolinkurl{GCR149}}\hypertarget{lh:GCR149}{}\enspace{\ttfamily\nolinkurl{GraphicCriticalReduction.mem_z3K5EdgeSetColoringsExec}} {\tiny\ttfamily Ssot/\allowbreak GraphicCriticalReduction.lean}
\item \textbf{\nolinkurl{GCR150}}\hypertarget{lh:GCR150}{}\enspace{\ttfamily\nolinkurl{GraphicCriticalReduction.z3K5EdgeSetColoringsExec_eq_properFieldColorings}} {\tiny\ttfamily Ssot/\allowbreak GraphicCriticalReduction.lean}
\item \textbf{\nolinkurl{GCR151}}\hypertarget{lh:GCR151}{}\enspace{\ttfamily\nolinkurl{GraphicCriticalReduction.k5SevenEdgeSubgraphs_card}} {\tiny\ttfamily Ssot/\allowbreak GraphicCriticalReduction.lean}
\item \textbf{\nolinkurl{GCR152}}\hypertarget{lh:GCR152}{}\enspace{\ttfamily\nolinkurl{GraphicCriticalReduction.k5SevenEdgeSubgraphs_colorCountValues}} {\tiny\ttfamily Ssot/\allowbreak GraphicCriticalReduction.lean}
\item \textbf{\nolinkurl{GCR153}}\hypertarget{lh:GCR153}{}\enspace{\ttfamily\nolinkurl{GraphicCriticalReduction.k5SevenEdgeSubgraphs_colorCountFiber_zero_card}} {\tiny\ttfamily Ssot/\allowbreak GraphicCriticalReduction.lean}
\item \textbf{\nolinkurl{GCR154}}\hypertarget{lh:GCR154}{}\enspace{\ttfamily\nolinkurl{GraphicCriticalReduction.k5SevenEdgeSubgraphs_colorCountFiber_six_card}} {\tiny\ttfamily Ssot/\allowbreak GraphicCriticalReduction.lean}
\item \textbf{\nolinkurl{GCR155}}\hypertarget{lh:GCR155}{}\enspace{\ttfamily\nolinkurl{GraphicCriticalReduction.k5SevenEdgeSubgraphs_colorCountFiber_twelve_card}} {\tiny\ttfamily Ssot/\allowbreak GraphicCriticalReduction.lean}
\item \textbf{\nolinkurl{GCR156}}\hypertarget{lh:GCR156}{}\enspace{\ttfamily\nolinkurl{GraphicCriticalReduction.k5SevenEdgeSixColorings_coloringCount}} {\tiny\ttfamily Ssot/\allowbreak GraphicCriticalReduction.lean}
\item \textbf{\nolinkurl{GCR157}}\hypertarget{lh:GCR157}{}\enspace{\ttfamily\nolinkurl{GraphicCriticalReduction.k5SevenEdgeColoringCountValues_eq}} {\tiny\ttfamily Ssot/\allowbreak GraphicCriticalReduction.lean}
\item \textbf{\nolinkurl{GCR158}}\hypertarget{lh:GCR158}{}\enspace{\ttfamily\nolinkurl{GraphicCriticalReduction.k5SevenEdge_threeColoringCounts_subset_fixedActiveShellFiber}} {\tiny\ttfamily Ssot/\allowbreak GraphicCriticalReduction.lean}
\item \textbf{\nolinkurl{GCR159}}\hypertarget{lh:GCR159}{}\enspace{\ttfamily\nolinkurl{GraphicCriticalReduction.k5SevenEdgeNormalizedColoringCount}} {\tiny\ttfamily Ssot/\allowbreak GraphicCriticalReduction.lean}
\item \textbf{\nolinkurl{GCR160}}\hypertarget{lh:GCR160}{}\enspace{\ttfamily\nolinkurl{GraphicCriticalReduction.k5SevenEdgeNormalizedColorable_count}} {\tiny\ttfamily Ssot/\allowbreak GraphicCriticalReduction.lean}
\item \textbf{\nolinkurl{GCR161}}\hypertarget{lh:GCR161}{}\enspace{\ttfamily\nolinkurl{GraphicCriticalReduction.k5SevenEdgeNormalizedSixColorings_count}} {\tiny\ttfamily Ssot/\allowbreak GraphicCriticalReduction.lean}
\item \textbf{\nolinkurl{GCR162}}\hypertarget{lh:GCR162}{}\enspace{\ttfamily\nolinkurl{GraphicCriticalReduction.k5SevenEdgeNormalizedNoncolorable_count}} {\tiny\ttfamily Ssot/\allowbreak GraphicCriticalReduction.lean}
\item \textbf{\nolinkurl{GCR163}}\hypertarget{lh:GCR163}{}\enspace{\ttfamily\nolinkurl{GraphicCriticalReduction.k5SevenEdgeNormalizedColoringCountValues_eq}} {\tiny\ttfamily Ssot/\allowbreak GraphicCriticalReduction.lean}
\item \textbf{\nolinkurl{GCR164}}\hypertarget{lh:GCR164}{}\enspace{\ttfamily\nolinkurl{GraphicCriticalReduction.k5SevenEdgeNormalizedColoringCountValues_card_gt_one}} {\tiny\ttfamily Ssot/\allowbreak GraphicCriticalReduction.lean}
\item \textbf{\nolinkurl{GCR165}}\hypertarget{lh:GCR165}{}\enspace{\ttfamily\nolinkurl{GraphicCriticalReduction.k5SevenEdge_threeNormalizedCounts_subset_fixedActiveShellFiber}} {\tiny\ttfamily Ssot/\allowbreak GraphicCriticalReduction.lean}
\item \textbf{\nolinkurl{GCR166}}\hypertarget{lh:GCR166}{}\enspace{\ttfamily\nolinkurl{GraphicCriticalReduction.k5SevenEdge_fixedActiveShellNormalizedCountFiber_not_subsingleton}} {\tiny\ttfamily Ssot/\allowbreak GraphicCriticalReduction.lean}
\item \textbf{\nolinkurl{GCR167}}\hypertarget{lh:GCR167}{}\enspace{\ttfamily\nolinkurl{GraphicCriticalReduction.k5SevenEdge_activeShell_does_not_refine_normalizedColoringCount}} {\tiny\ttfamily Ssot/\allowbreak GraphicCriticalReduction.lean}
\item \textbf{\nolinkurl{GCR168}}\hypertarget{lh:GCR168}{}\enspace{\ttfamily\nolinkurl{GraphicCriticalReduction.z3K5EdgeSetNormalizedColoringCount}} {\tiny\ttfamily Ssot/\allowbreak GraphicCriticalReduction.lean}
\item \textbf{\nolinkurl{GCR169}}\hypertarget{lh:GCR169}{}\enspace{\ttfamily\nolinkurl{GraphicCriticalReduction.k5SevenEdgeSubgraphNormalizedColoringCountFiber}} {\tiny\ttfamily Ssot/\allowbreak GraphicCriticalReduction.lean}
\item \textbf{\nolinkurl{GCR170}}\hypertarget{lh:GCR170}{}\enspace{\ttfamily\nolinkurl{GraphicCriticalReduction.k5SevenEdgeSubgraphs_normalizedColorCountValues}} {\tiny\ttfamily Ssot/\allowbreak GraphicCriticalReduction.lean}
\item \textbf{\nolinkurl{GCR171}}\hypertarget{lh:GCR171}{}\enspace{\ttfamily\nolinkurl{GraphicCriticalReduction.k5SevenEdgeSubgraphs_normalizedColorCountFiber_zero_card}} {\tiny\ttfamily Ssot/\allowbreak GraphicCriticalReduction.lean}
\item \textbf{\nolinkurl{GCR172}}\hypertarget{lh:GCR172}{}\enspace{\ttfamily\nolinkurl{GraphicCriticalReduction.k5SevenEdgeSubgraphs_normalizedColorCountFiber_one_card}} {\tiny\ttfamily Ssot/\allowbreak GraphicCriticalReduction.lean}
\item \textbf{\nolinkurl{GCR173}}\hypertarget{lh:GCR173}{}\enspace{\ttfamily\nolinkurl{GraphicCriticalReduction.k5SevenEdgeSubgraphs_normalizedColorCountFiber_two_card}} {\tiny\ttfamily Ssot/\allowbreak GraphicCriticalReduction.lean}
\item \textbf{\nolinkurl{GCR174}}\hypertarget{lh:GCR174}{}\enspace{\ttfamily\nolinkurl{GraphicCriticalReduction.indexedColoringViolationProfileShell}} {\tiny\ttfamily Ssot/\allowbreak GraphicCriticalReduction.lean}
\item \textbf{\nolinkurl{GCR175}}\hypertarget{lh:GCR175}{}\enspace{\ttfamily\nolinkurl{GraphicCriticalReduction.properIndexedFieldColorings_card_eq_indexedColoringViolationProfileShell_zero}} {\tiny\ttfamily Ssot/\allowbreak GraphicCriticalReduction.lean}
\item \textbf{\nolinkurl{GCR176}}\hypertarget{lh:GCR176}{}\enspace{\ttfamily\nolinkurl{GraphicCriticalReduction.properIndexedFieldColorings_card_eq_of_same_indexedColoringViolationProfileShell}} {\tiny\ttfamily Ssot/\allowbreak GraphicCriticalReduction.lean}
\item \textbf{\nolinkurl{GCR177}}\hypertarget{lh:GCR177}{}\enspace{\ttfamily\nolinkurl{GraphicCriticalReduction.properIndexedFieldColorings_card_pos_iff_of_same_indexedColoringViolationProfileShell}} {\tiny\ttfamily Ssot/\allowbreak GraphicCriticalReduction.lean}
\item \textbf{\nolinkurl{GCR178}}\hypertarget{lh:GCR178}{}\enspace{\ttfamily\nolinkurl{GraphicCriticalReduction.coloringCountValue_semanticComplete_indexedColoringViolationProfileShell}} {\tiny\ttfamily Ssot/\allowbreak GraphicCriticalReduction.lean}
\item \textbf{\nolinkurl{GCR179}}\hypertarget{lh:GCR179}{}\enspace{\ttfamily\nolinkurl{GraphicCriticalReduction.coloringCountPositive_semanticComplete_indexedColoringViolationProfileShell}} {\tiny\ttfamily Ssot/\allowbreak GraphicCriticalReduction.lean}
\item \textbf{\nolinkurl{GCR180}}\hypertarget{lh:GCR180}{}\enspace{\ttfamily\nolinkurl{GraphicCriticalReduction.coloringCountValue_control_or_violationProfileFiberCollision}} {\tiny\ttfamily Ssot/\allowbreak GraphicCriticalReduction.lean}
\item \textbf{\nolinkurl{GCR181}}\hypertarget{lh:GCR181}{}\enspace{\ttfamily\nolinkurl{GraphicCriticalReduction.coloringCountPositive_control_or_violationProfileFiberCollision}} {\tiny\ttfamily Ssot/\allowbreak GraphicCriticalReduction.lean}
\item \textbf{\nolinkurl{GCR182}}\hypertarget{lh:GCR182}{}\enspace{\ttfamily\nolinkurl{GraphicCriticalReduction.liftedCandidateCountValue_semanticComplete_indexedColoringViolationProfileShell}} {\tiny\ttfamily Ssot/\allowbreak GraphicCriticalReduction.lean}
\item \textbf{\nolinkurl{GCR183}}\hypertarget{lh:GCR183}{}\enspace{\ttfamily\nolinkurl{GraphicCriticalReduction.liftedCandidateCountPositive_semanticComplete_indexedColoringViolationProfileShell}} {\tiny\ttfamily Ssot/\allowbreak GraphicCriticalReduction.lean}
\item \textbf{\nolinkurl{GCR184}}\hypertarget{lh:GCR184}{}\enspace{\ttfamily\nolinkurl{GraphicCriticalReduction.liftedCandidateCountValue_control_or_violationProfileFiberCollision}} {\tiny\ttfamily Ssot/\allowbreak GraphicCriticalReduction.lean}
\item \textbf{\nolinkurl{GCR185}}\hypertarget{lh:GCR185}{}\enspace{\ttfamily\nolinkurl{GraphicCriticalReduction.liftedCandidateCountPositive_control_or_violationProfileFiberCollision}} {\tiny\ttfamily Ssot/\allowbreak GraphicCriticalReduction.lean}
\item \textbf{\nolinkurl{GCR186}}\hypertarget{lh:GCR186}{}\enspace{\ttfamily\nolinkurl{GraphicCriticalReduction.activeShell_liftedCandidateCountValue_control_or_violationProfileFiberCollision_of_indexedSimpleGraphicCommonCoreLift}} {\tiny\ttfamily Ssot/\allowbreak GraphicCriticalReduction.lean}
\item \textbf{\nolinkurl{GCR187}}\hypertarget{lh:GCR187}{}\enspace{\ttfamily\nolinkurl{GraphicCriticalReduction.activeShell_liftedCandidateCountPositive_control_or_violationProfileFiberCollision_of_indexedSimpleGraphicCommonCoreLift}} {\tiny\ttfamily Ssot/\allowbreak GraphicCriticalReduction.lean}
\item \textbf{\nolinkurl{GCR188}}\hypertarget{lh:GCR188}{}\enspace{\ttfamily\nolinkurl{GraphicCriticalReduction.liftedCandidatePositive_singleton_iff_not_indexedColoringConflict}} {\tiny\ttfamily Ssot/\allowbreak GraphicCriticalReduction.lean}
\item \textbf{\nolinkurl{GCR189}}\hypertarget{lh:GCR189}{}\enspace{\ttfamily\nolinkurl{GraphicCriticalReduction.liftedIndexedCandidateConflictSet}} {\tiny\ttfamily Ssot/\allowbreak GraphicCriticalReduction.lean}
\item \textbf{\nolinkurl{GCR190}}\hypertarget{lh:GCR190}{}\enspace{\ttfamily\nolinkurl{GraphicCriticalReduction.liftedIndexedCandidateConflictSet_eq_indexedColoringConflictSet}} {\tiny\ttfamily Ssot/\allowbreak GraphicCriticalReduction.lean}
\item \textbf{\nolinkurl{GCR191}}\hypertarget{lh:GCR191}{}\enspace{\ttfamily\nolinkurl{GraphicCriticalReduction.liftedIndexedCandidateViolationProfileShell}} {\tiny\ttfamily Ssot/\allowbreak GraphicCriticalReduction.lean}
\item \textbf{\nolinkurl{GCR192}}\hypertarget{lh:GCR192}{}\enspace{\ttfamily\nolinkurl{GraphicCriticalReduction.liftedIndexedCandidateViolationProfileShell_eq_indexedColoringViolationProfileShell}} {\tiny\ttfamily Ssot/\allowbreak GraphicCriticalReduction.lean}
\item \textbf{\nolinkurl{GCR193}}\hypertarget{lh:GCR193}{}\enspace{\ttfamily\nolinkurl{GraphicCriticalReduction.liftedCandidatePositiveOn}} {\tiny\ttfamily Ssot/\allowbreak GraphicCriticalReduction.lean}
\item \textbf{\nolinkurl{GCR194}}\hypertarget{lh:GCR194}{}\enspace{\ttfamily\nolinkurl{GraphicCriticalReduction.liftedCandidatePositive_iff_liftedCandidatePositiveOn_views}} {\tiny\ttfamily Ssot/\allowbreak GraphicCriticalReduction.lean}
\item \textbf{\nolinkurl{GCR195}}\hypertarget{lh:GCR195}{}\enspace{\ttfamily\nolinkurl{GraphicCriticalReduction.liftedCandidateViewConflictSet}} {\tiny\ttfamily Ssot/\allowbreak GraphicCriticalReduction.lean}
\item \textbf{\nolinkurl{GCR196}}\hypertarget{lh:GCR196}{}\enspace{\ttfamily\nolinkurl{GraphicCriticalReduction.liftedCandidateViewViolationProfile}} {\tiny\ttfamily Ssot/\allowbreak GraphicCriticalReduction.lean}
\item \textbf{\nolinkurl{GCR197}}\hypertarget{lh:GCR197}{}\enspace{\ttfamily\nolinkurl{GraphicCriticalReduction.liftedCandidatePositiveOn_graphicIndexedCommonCoreLift_singleton_iff_not_indexedColoringConflict}} {\tiny\ttfamily Ssot/\allowbreak GraphicCriticalReduction.lean}
\item \textbf{\nolinkurl{GCR198}}\hypertarget{lh:GCR198}{}\enspace{\ttfamily\nolinkurl{GraphicCriticalReduction.liftedCandidateViewConflictSet_graphicIndexedCommonCoreLift_eq_indexedColoringConflictSet}} {\tiny\ttfamily Ssot/\allowbreak GraphicCriticalReduction.lean}
\item \textbf{\nolinkurl{GCR199}}\hypertarget{lh:GCR199}{}\enspace{\ttfamily\nolinkurl{GraphicCriticalReduction.liftedCandidateViewViolationProfile_graphicIndexedCommonCoreLift_eq_indexedColoringViolationProfileShell}} {\tiny\ttfamily Ssot/\allowbreak GraphicCriticalReduction.lean}
\item \textbf{\nolinkurl{GCR200}}\hypertarget{lh:GCR200}{}\enspace{\ttfamily\nolinkurl{GraphicCriticalReduction.criticalProblemCommonCoreLiftOn_fixedShell_and_viewViolationProfile_eq_sourceProfile}} {\tiny\ttfamily Ssot/\allowbreak GraphicCriticalReduction.lean}
\item \textbf{\nolinkurl{GCR201}}\hypertarget{lh:GCR201}{}\enspace{\ttfamily\nolinkurl{GraphicCriticalReduction.fixedShellTargetReductionOn_viewViolationProfile_of_indexedSimpleGraphicCommonCoreLift_nonemptyCore}}
\item \textbf{\nolinkurl{GCR202}}\hypertarget{lh:GCR202}{}\enspace{\ttfamily\nolinkurl{GraphicCriticalReduction.not_activeShell_refines_viewViolationProfile_of_indexedSimpleGraphicCommonCoreLift_profiles_ne}} {\tiny\ttfamily Ssot/\allowbreak GraphicCriticalReduction.lean}
\item \textbf{\nolinkurl{GCR203}}\hypertarget{lh:GCR203}{}\enspace{\ttfamily\nolinkurl{GraphicCriticalReduction.exists_sameActiveShell_diffViewViolationProfile_of_indexedSimpleGraphicCommonCoreLift_profiles_ne}} {\tiny\ttfamily Ssot/\allowbreak GraphicCriticalReduction.lean}
\item \textbf{\nolinkurl{GCR204}}\hypertarget{lh:GCR204}{}\enspace{\ttfamily\nolinkurl{GraphicCriticalReduction.not_forall_viewViolationProfileValue_semanticComplete_of_indexedSimpleGraphicCommonCoreLift_profiles_ne}} {\tiny\ttfamily Ssot/\allowbreak GraphicCriticalReduction.lean}
\item \textbf{\nolinkurl{GCR205}}\hypertarget{lh:GCR205}{}\enspace{\ttfamily\nolinkurl{GraphicCriticalReduction.boundedIndexedColoringViolationProfileShell}} {\tiny\ttfamily Ssot/\allowbreak GraphicCriticalReduction.lean}
\item \textbf{\nolinkurl{GCR206}}\hypertarget{lh:GCR206}{}\enspace{\ttfamily\nolinkurl{GraphicCriticalReduction.boundedLiftedCandidateViewViolationProfile}} {\tiny\ttfamily Ssot/\allowbreak GraphicCriticalReduction.lean}
\item \textbf{\nolinkurl{GCR207}}\hypertarget{lh:GCR207}{}\enspace{\ttfamily\nolinkurl{GraphicCriticalReduction.boundedLiftedCandidateViewViolationProfile_graphicIndexedCommonCoreLift_eq_boundedIndexedColoringViolationProfileShell}} {\tiny\ttfamily Ssot/\allowbreak GraphicCriticalReduction.lean}
\item \textbf{\nolinkurl{GCR208}}\hypertarget{lh:GCR208}{}\enspace{\ttfamily\nolinkurl{GraphicCriticalReduction.fixedShellTargetReductionOn_boundedViewViolationProfile_of_indexedSimpleGraphicCommonCoreLift_nonemptyCore}}
\item \textbf{\nolinkurl{GCR209}}\hypertarget{lh:GCR209}{}\enspace{\ttfamily\nolinkurl{GraphicCriticalReduction.finiteSourceBoundedViolationProfileValues}} {\tiny\ttfamily Ssot/\allowbreak GraphicCriticalReduction.lean}
\item \textbf{\nolinkurl{GCR210}}\hypertarget{lh:GCR210}{}\enspace{\ttfamily\nolinkurl{GraphicCriticalReduction.mem_finiteSourceBoundedViolationProfileValues}} {\tiny\ttfamily Ssot/\allowbreak GraphicCriticalReduction.lean}
\item \textbf{\nolinkurl{GCR211}}\hypertarget{lh:GCR211}{}\enspace{\ttfamily\nolinkurl{GraphicCriticalReduction.finiteSourceBoundedViolationProfileValues_subset_fixedActiveShellBoundedProfileFiber}} {\tiny\ttfamily Ssot/\allowbreak GraphicCriticalReduction.lean}
\item \textbf{\nolinkurl{GCR212}}\hypertarget{lh:GCR212}{}\enspace{\ttfamily\nolinkurl{GraphicCriticalReduction.exists_boundedProfileValues_subset_fixedActiveShellBoundedProfileFiber_card_eq_of_injOn}} {\tiny\ttfamily Ssot/\allowbreak GraphicCriticalReduction.lean}
\item \textbf{\nolinkurl{GCR213}}\hypertarget{lh:GCR213}{}\enspace{\ttfamily\nolinkurl{GraphicCriticalReduction.not_activeShell_refines_boundedViewViolationProfile_of_finiteSourceBoundedViolationProfileValues_card_gt_one}} {\tiny\ttfamily Ssot/\allowbreak GraphicCriticalReduction.lean}
\item \textbf{\nolinkurl{GCR214}}\hypertarget{lh:GCR214}{}\enspace{\ttfamily\nolinkurl{GraphicCriticalReduction.properIndexedFieldColorings_card_eq_boundedIndexedColoringViolationProfileShell_zero}} {\tiny\ttfamily Ssot/\allowbreak GraphicCriticalReduction.lean}
\item \textbf{\nolinkurl{GCR215}}\hypertarget{lh:GCR215}{}\enspace{\ttfamily\nolinkurl{GraphicCriticalReduction.properIndexedFieldColorings_card_eq_of_same_boundedIndexedColoringViolationProfileShell}} {\tiny\ttfamily Ssot/\allowbreak GraphicCriticalReduction.lean}
\item \textbf{\nolinkurl{GCR216}}\hypertarget{lh:GCR216}{}\enspace{\ttfamily\nolinkurl{GraphicCriticalReduction.boundedIndexedColoringViolationProfileShell_injOn_of_coloringCount_injOn}} {\tiny\ttfamily Ssot/\allowbreak GraphicCriticalReduction.lean}
\item \textbf{\nolinkurl{GCR217}}\hypertarget{lh:GCR217}{}\enspace{\ttfamily\nolinkurl{GraphicCriticalReduction.k5SevenEdgeColoringCount_injOn_univ}} {\tiny\ttfamily Ssot/\allowbreak GraphicCriticalReduction.lean}
\item \textbf{\nolinkurl{GCR218}}\hypertarget{lh:GCR218}{}\enspace{\ttfamily\nolinkurl{GraphicCriticalReduction.k5SevenEdgeBoundedViolationProfile_injOn_univ}} {\tiny\ttfamily Ssot/\allowbreak GraphicCriticalReduction.lean}
\item \textbf{\nolinkurl{GCR219}}\hypertarget{lh:GCR219}{}\enspace{\ttfamily\nolinkurl{GraphicCriticalReduction.k5SevenEdgeBoundedViolationProfileValues_card_eq}} {\tiny\ttfamily Ssot/\allowbreak GraphicCriticalReduction.lean}
\item \textbf{\nolinkurl{GCR220}}\hypertarget{lh:GCR220}{}\enspace{\ttfamily\nolinkurl{GraphicCriticalReduction.k5SevenEdgeBoundedViolationProfileValues_card_gt_one}} {\tiny\ttfamily Ssot/\allowbreak GraphicCriticalReduction.lean}
\item \textbf{\nolinkurl{GCR221}}\hypertarget{lh:GCR221}{}\enspace{\ttfamily\nolinkurl{GraphicCriticalReduction.k5SevenEdge_threeBoundedProfiles_subset_fixedActiveShellFiber}} {\tiny\ttfamily Ssot/\allowbreak GraphicCriticalReduction.lean}
\item \textbf{\nolinkurl{GCR222}}\hypertarget{lh:GCR222}{}\enspace{\ttfamily\nolinkurl{GraphicCriticalReduction.k5SevenEdge_activeShell_does_not_refine_boundedViewViolationProfile}} {\tiny\ttfamily Ssot/\allowbreak GraphicCriticalReduction.lean}
\item \textbf{\nolinkurl{GCR223}}\hypertarget{lh:GCR223}{}\enspace{\ttfamily\nolinkurl{GraphicCriticalReduction.graphicIndexedQuotientSearch_positiveCandidates_eq_properIndexedFieldColorings}} {\tiny\ttfamily Ssot/\allowbreak GraphicCriticalReduction.lean}
\item \textbf{\nolinkurl{GCR224}}\hypertarget{lh:GCR224}{}\enspace{\ttfamily\nolinkurl{GraphicCriticalReduction.commonCoreLift_witnesses_graphicIndexedCommonCoreLift_card_eq_properIndexedFieldColorings}} {\tiny\ttfamily Ssot/\allowbreak GraphicCriticalReduction.lean}
\item \textbf{\nolinkurl{GCR225}}\hypertarget{lh:GCR225}{}\enspace{\ttfamily\nolinkurl{GraphicCriticalReduction.fixedShellTargetReductionOn_witnessPairCount_of_indexedSimpleGraphicCommonCoreLift_nonemptyCore}}
\item \textbf{\nolinkurl{GCR226}}\hypertarget{lh:GCR226}{}\enspace{\ttfamily\nolinkurl{GraphicCriticalReduction.commonCoreLift_witnesses_graphicIndexedCommonCoreLift_card_eq_finiteSourceWitnesses}} {\tiny\ttfamily Ssot/\allowbreak GraphicCriticalReduction.lean}
\item \textbf{\nolinkurl{GCR227}}\hypertarget{lh:GCR227}{}\enspace{\ttfamily\nolinkurl{GraphicCriticalReduction.fixedShellTargetReductionOn_witnessPairCount_of_indexedColoringWitnessReduction_nonemptyCore}}
\item \textbf{\nolinkurl{GCR228}}\hypertarget{lh:GCR228}{}\enspace{\ttfamily\nolinkurl{GraphicCriticalReduction.finiteSourceWitnessCountValues_subset_fixedActiveShellWitnessPairCountFiber_of_indexedColoringWitnessReduction}} {\tiny\ttfamily Ssot/\allowbreak GraphicCriticalReduction.lean}
\item \textbf{\nolinkurl{GCR229}}\hypertarget{lh:GCR229}{}\enspace{\ttfamily\nolinkurl{GraphicCriticalReduction.exists_sameActiveShell_diffWitnessPairCount_of_indexedColoringWitnessReduction}} {\tiny\ttfamily Ssot/\allowbreak GraphicCriticalReduction.lean}
\item \textbf{\nolinkurl{GCR230}}\hypertarget{lh:GCR230}{}\enspace{\ttfamily\nolinkurl{GraphicCriticalReduction.not_activeShell_refines_witnessPairCount_of_indexedColoringWitnessReduction}} {\tiny\ttfamily Ssot/\allowbreak GraphicCriticalReduction.lean}
\item \textbf{\nolinkurl{GCR231}}\hypertarget{lh:GCR231}{}\enspace{\ttfamily\nolinkurl{GraphicCriticalReduction.indexedColoringWitnessReductionSelf}} {\tiny\ttfamily Ssot/\allowbreak GraphicCriticalReduction.lean}
\item \textbf{\nolinkurl{GCR232}}\hypertarget{lh:GCR232}{}\enspace{\ttfamily\nolinkurl{GraphicCriticalReduction.k5SevenEdgeWitnessReduction}} {\tiny\ttfamily Ssot/\allowbreak GraphicCriticalReduction.lean}
\item \textbf{\nolinkurl{GCR233}}\hypertarget{lh:GCR233}{}\enspace{\ttfamily\nolinkurl{GraphicCriticalReduction.k5SevenEdgeWitnessCountValues_eq}} {\tiny\ttfamily Ssot/\allowbreak GraphicCriticalReduction.lean}
\item \textbf{\nolinkurl{GCR234}}\hypertarget{lh:GCR234}{}\enspace{\ttfamily\nolinkurl{GraphicCriticalReduction.k5SevenEdge_threeWitnessPairCounts_subset_fixedActiveShellFiber}} {\tiny\ttfamily Ssot/\allowbreak GraphicCriticalReduction.lean}
\item \textbf{\nolinkurl{GCR235}}\hypertarget{lh:GCR235}{}\enspace{\ttfamily\nolinkurl{GraphicCriticalReduction.k5SevenEdge_activeShell_does_not_refine_witnessPairCount}} {\tiny\ttfamily Ssot/\allowbreak GraphicCriticalReduction.lean}
\item \textbf{\nolinkurl{KSH1}}\hypertarget{lh:KSH1}{}\enspace{\ttfamily\nolinkurl{KernelSectionHardness.Literal.homCoord_anchor_true}} {\tiny\ttfamily Ssot/\allowbreak KernelSectionHardness.lean}
\item \textbf{\nolinkurl{KSH2}}\hypertarget{lh:KSH2}{}\enspace{\ttfamily\nolinkurl{KernelSectionHardness.Clause3.homViewNonzero_anchor_true}} {\tiny\ttfamily Ssot/\allowbreak KernelSectionHardness.lean}
\item \textbf{\nolinkurl{KSH3}}\hypertarget{lh:KSH3}{}\enspace{\ttfamily\nolinkurl{KernelSectionHardness.SAT3Instance.satisfiable_implies_hasKernelSectionPoint}} {\tiny\ttfamily Ssot/\allowbreak KernelSectionHardness.lean}
\item \textbf{\nolinkurl{KSH4}}\hypertarget{lh:KSH4}{}\enspace{\ttfamily\nolinkurl{KernelSectionHardness.SAT3Instance.hasKernelSectionPoint_implies_satisfiable}} {\tiny\ttfamily Ssot/\allowbreak KernelSectionHardness.lean}
\item \textbf{\nolinkurl{KSH5}}\hypertarget{lh:KSH5}{}\enspace{\ttfamily\nolinkurl{KernelSectionHardness.SAT3Instance.satisfiable_iff_hasKernelSectionPoint}} {\tiny\ttfamily Ssot/\allowbreak KernelSectionHardness.lean}
\item \textbf{\nolinkurl{KSH6}}\hypertarget{lh:KSH6}{}\enspace{\ttfamily\nolinkurl{KernelSectionHardness.SAT3Instance.hasKernelSectionPoint_implies_hasNonlinearSectionPair}} {\tiny\ttfamily Ssot/\allowbreak KernelSectionHardness.lean}
\item \textbf{\nolinkurl{KSH7}}\hypertarget{lh:KSH7}{}\enspace{\ttfamily\nolinkurl{KernelSectionHardness.SAT3Instance.hasNonlinearSectionPair_implies_hasKernelSectionPoint}} {\tiny\ttfamily Ssot/\allowbreak KernelSectionHardness.lean}
\item \textbf{\nolinkurl{KSH8}}\hypertarget{lh:KSH8}{}\enspace{\ttfamily\nolinkurl{KernelSectionHardness.SAT3Instance.satisfiable_iff_hasNonlinearSectionPair}} {\tiny\ttfamily Ssot/\allowbreak KernelSectionHardness.lean}
\item \textbf{\nolinkurl{KSH9}}\hypertarget{lh:KSH9}{}\enspace{\ttfamily\nolinkurl{KernelSectionHardness.SAT3Instance.clausesAllNonzero_eq_true_iff}} {\tiny\ttfamily Ssot/\allowbreak KernelSectionHardness.lean}
\item \textbf{\nolinkurl{KSH10}}\hypertarget{lh:KSH10}{}\enspace{\ttfamily\nolinkurl{KernelSectionHardness.SAT3Instance.linearCertificateVerifier_eq_true_iff}} {\tiny\ttfamily Ssot/\allowbreak KernelSectionHardness.lean}
\item \textbf{\nolinkurl{KSH11}}\hypertarget{lh:KSH11}{}\enspace{\ttfamily\nolinkurl{KernelSectionHardness.SAT3Instance.exists_linearCertificateVerifier_iff_hasKernelSectionPoint}} {\tiny\ttfamily Ssot/\allowbreak KernelSectionHardness.lean}
\item \textbf{\nolinkurl{KSH12}}\hypertarget{lh:KSH12}{}\enspace{\ttfamily\nolinkurl{KernelSectionHardness.SAT3Instance.nonlinearPairCertificateVerifier_eq_true_iff}} {\tiny\ttfamily Ssot/\allowbreak KernelSectionHardness.lean}
\item \textbf{\nolinkurl{KSH13}}\hypertarget{lh:KSH13}{}\enspace{\ttfamily\nolinkurl{KernelSectionHardness.SAT3Instance.exists_nonlinearPairCertificateVerifier_iff_hasNonlinearSectionPair}} {\tiny\ttfamily Ssot/\allowbreak KernelSectionHardness.lean}
\item \textbf{\nolinkurl{KSH14}}\hypertarget{lh:KSH14}{}\enspace{\ttfamily\nolinkurl{KernelSectionHardness.SAT3Instance.linearCertificateEquivSatisfyingAssignment}} {\tiny\ttfamily Ssot/\allowbreak KernelSectionHardness.lean}
\item \textbf{\nolinkurl{KSH15}}\hypertarget{lh:KSH15}{}\enspace{\ttfamily\nolinkurl{KernelSectionHardness.SAT3Instance.linearCertificate_card_eq_satisfyingAssignment_card}} {\tiny\ttfamily Ssot/\allowbreak KernelSectionHardness.lean}
\item \textbf{\nolinkurl{MFT3}}\hypertarget{lh:MFT3}{}\enspace{\ttfamily\nolinkurl{MultiFact.exactRecovery_iff_properColoring}} {\tiny\ttfamily Ssot/\allowbreak MultiFact.lean}
\item \textbf{\nolinkurl{MFT4}}\hypertarget{lh:MFT4}{}\enspace{\ttfamily\nolinkurl{MultiFact.exists_exactRecovery_iff_colorable}} {\tiny\ttfamily Ssot/\allowbreak MultiFact.lean}
\item \textbf{\nolinkurl{MFT9}}\hypertarget{lh:MFT9}{}\enspace{\ttfamily\nolinkurl{MultiFact.exists_exactOn_iff_colorableOn}} {\tiny\ttfamily Ssot/\allowbreak MultiFact.lean}
\item \textbf{\nolinkurl{MFT14}}\hypertarget{lh:MFT14}{}\enspace{\ttfamily\nolinkurl{MultiFact.exists_exactOn_card_eq_maxColorableCard}} {\tiny\ttfamily Ssot/\allowbreak MultiFact.lean}
\item \textbf{\nolinkurl{MFT15}}\hypertarget{lh:MFT15}{}\enspace{\ttfamily\nolinkurl{MultiFact.binaryViews_maxColorableCard_one_eq_two}} {\tiny\ttfamily Ssot/\allowbreak MultiFact.lean}
\item \textbf{\nolinkurl{MFT16}}\hypertarget{lh:MFT16}{}\enspace{\ttfamily\nolinkurl{MultiFact.pair_product_success_card_lower_bound}} {\tiny\ttfamily Ssot/\allowbreak MultiFact.lean}
\item \textbf{\nolinkurl{MFT20}}\hypertarget{lh:MFT20}{}\enspace{\ttfamily\nolinkurl{MultiFact.colorable_iff_graphColorable}} {\tiny\ttfamily Ssot/\allowbreak MultiFact.lean}
\item \textbf{\nolinkurl{MFT25}}\hypertarget{lh:MFT25}{}\enspace{\ttfamily\nolinkurl{MultiFact.pairConfusable_iff_strongProdAdj}} {\tiny\ttfamily Ssot/\allowbreak MultiFact.lean}
\item \textbf{\nolinkurl{MFT26}}\hypertarget{lh:MFT26}{}\enspace{\ttfamily\nolinkurl{MultiFact.pairConfusabilityGraph_eq_strongProd}} {\tiny\ttfamily Ssot/\allowbreak MultiFact.lean}
\item \textbf{\nolinkurl{MFT43}}\hypertarget{lh:MFT43}{}\enspace{\ttfamily\nolinkurl{MultiFact.graphNegLogSeq_subadditive}} {\tiny\ttfamily Ssot/\allowbreak MultiFact.lean}
\item \textbf{\nolinkurl{MFT44}}\hypertarget{lh:MFT44}{}\enspace{\ttfamily\nolinkurl{MultiFact.tendsto_graphPowerRateNat_graphShannonCapacityReal}} {\tiny\ttfamily Ssot/\allowbreak MultiFact.lean}
\item \textbf{\nolinkurl{MFT45}}\hypertarget{lh:MFT45}{}\enspace{\ttfamily\nolinkurl{MultiFact.graphPowerRate_le_graphShannonCapacityReal}} {\tiny\ttfamily Ssot/\allowbreak MultiFact.lean}
\item \textbf{\nolinkurl{MFT46}}\hypertarget{lh:MFT46}{}\enspace{\ttfamily\nolinkurl{MultiFact.graphShannonCapacityReal_nonneg}} {\tiny\ttfamily Ssot/\allowbreak MultiFact.lean}
\item \textbf{\nolinkurl{MFT47}}\hypertarget{lh:MFT47}{}\enspace{\ttfamily\nolinkurl{MultiFact.iSup_graphPowerRate_eq_graphShannonCapacityReal}} {\tiny\ttfamily Ssot/\allowbreak MultiFact.lean}
\item \textbf{\nolinkurl{MFT48}}\hypertarget{lh:MFT48}{}\enspace{\ttfamily\nolinkurl{MultiFact.shannonCapacityReal_eq_iSup_blockRate}} {\tiny\ttfamily Ssot/\allowbreak MultiFact.lean}
\item \textbf{\nolinkurl{MFT49}}\hypertarget{lh:MFT49}{}\enspace{\ttfamily\nolinkurl{MultiFact.blockRate_le_shannonCapacityReal}} {\tiny\ttfamily Ssot/\allowbreak MultiFact.lean}
\item \textbf{\nolinkurl{MFT56}}\hypertarget{lh:MFT56}{}\enspace{\ttfamily\nolinkurl{MultiFact.graphShannonLowerCapacity_eq_ofReal_graphShannonCapacityReal}} {\tiny\ttfamily Ssot/\allowbreak MultiFact.lean}
\item \textbf{\nolinkurl{MFT86}}\hypertarget{lh:MFT86}{}\enspace{\ttfamily\nolinkurl{MultiFact.graphShannonCapacityReal_clusterGraph_eq_rank_log}} {\tiny\ttfamily Ssot/\allowbreak MultiFact.lean}
\item \textbf{\nolinkurl{MFT89}}\hypertarget{lh:MFT89}{}\enspace{\ttfamily\nolinkurl{MultiFact.shannonCapacityReal_eq_shannonLovaszThetaAsymptoticUpper_of_fiberCoherent}} {\tiny\ttfamily Ssot/\allowbreak MultiFact.lean}
\item \textbf{\nolinkurl{MFT90}}\hypertarget{lh:MFT90}{}\enspace{\ttfamily\nolinkurl{MultiFact.shannonCapacityReal_eq_log_transcriptFiberCard_of_fiberCoherent}} {\tiny\ttfamily Ssot/\allowbreak MultiFact.lean}
\item \textbf{\nolinkurl{MFT91}}\hypertarget{lh:MFT91}{}\enspace{\ttfamily\nolinkurl{MultiFact.shannonLovaszThetaAsymptoticUpper_eq_log_transcriptFiberCard_of_fiberCoherent}} {\tiny\ttfamily Ssot/\allowbreak MultiFact.lean}
\item \textbf{\nolinkurl{MFT92}}\hypertarget{lh:MFT92}{}\enspace{\ttfamily\nolinkurl{MultiFact.EdgeSubgraphUpperCertificate}} {\tiny\ttfamily Ssot/\allowbreak MultiFact.lean}
\item \textbf{\nolinkurl{MFT93}}\hypertarget{lh:MFT93}{}\enspace{\ttfamily\nolinkurl{MultiFact.graphShannonCapacityReal_le_of_edgeSubgraphUpperCertificate}} {\tiny\ttfamily Ssot/\allowbreak MultiFact.lean}
\item \textbf{\nolinkurl{MFT94}}\hypertarget{lh:MFT94}{}\enspace{\ttfamily\nolinkurl{MultiFact.graphShannonCapacityReal_le_chain_of_edge_supersets}} {\tiny\ttfamily Ssot/\allowbreak MultiFact.lean}
\item \textbf{\nolinkurl{MFT95}}\hypertarget{lh:MFT95}{}\enspace{\ttfamily\nolinkurl{MultiFact.nat_eq_of_log_le_log_and_le}} {\tiny\ttfamily Ssot/\allowbreak CapacityLower.lean}
\item \textbf{\nolinkurl{MFT96}}\hypertarget{lh:MFT96}{}\enspace{\ttfamily\nolinkurl{MultiFact.graphUpperCapacityGap}} {\tiny\ttfamily Ssot/\allowbreak CapacityLower.lean}
\item \textbf{\nolinkurl{MFT97}}\hypertarget{lh:MFT97}{}\enspace{\ttfamily\nolinkurl{MultiFact.graphUpperPowerGapRate}} {\tiny\ttfamily Ssot/\allowbreak CapacityLower.lean}
\item \textbf{\nolinkurl{MFT98}}\hypertarget{lh:MFT98}{}\enspace{\ttfamily\nolinkurl{MultiFact.graphUpperCapacityGap_nonneg_of_upper}} {\tiny\ttfamily Ssot/\allowbreak CapacityLower.lean}
\item \textbf{\nolinkurl{MFT99}}\hypertarget{lh:MFT99}{}\enspace{\ttfamily\nolinkurl{MultiFact.graphUpperCapacityGap_eq_zero_iff_of_upper}} {\tiny\ttfamily Ssot/\allowbreak CapacityLower.lean}
\item \textbf{\nolinkurl{MFT100}}\hypertarget{lh:MFT100}{}\enspace{\ttfamily\nolinkurl{MultiFact.tendsto_graphUpperPowerGapRate_graphUpperCapacityGap}} {\tiny\ttfamily Ssot/\allowbreak CapacityLower.lean}
\item \textbf{\nolinkurl{MFT101}}\hypertarget{lh:MFT101}{}\enspace{\ttfamily\nolinkurl{MultiFact.graphUpperGap_add_of_capacity_add}} {\tiny\ttfamily Ssot/\allowbreak CapacityLower.lean}
\item \textbf{\nolinkurl{MFT102}}\hypertarget{lh:MFT102}{}\enspace{\ttfamily\nolinkurl{MultiFact.confusable_iff_agreeSet_card_ge_of_all_card_views}} {\tiny\ttfamily Ssot/\allowbreak MultiFact.lean}
\item \textbf{\nolinkurl{MFT103}}\hypertarget{lh:MFT103}{}\enspace{\ttfamily\nolinkurl{MultiFact.independent_iff_pair_agreeSet_card_lt_of_all_card_views}} {\tiny\ttfamily Ssot/\allowbreak MultiFact.lean}
\item \textbf{\nolinkurl{MFT104}}\hypertarget{lh:MFT104}{}\enspace{\ttfamily\nolinkurl{MultiFact.independent_card_le_pow_of_all_card_views}} {\tiny\ttfamily Ssot/\allowbreak MultiFact.lean}
\item \textbf{\nolinkurl{MFT105}}\hypertarget{lh:MFT105}{}\enspace{\ttfamily\nolinkurl{MultiFact.maxIndependentCard_le_pow_of_all_card_views}} {\tiny\ttfamily Ssot/\allowbreak MultiFact.lean}
\item \textbf{\nolinkurl{MFT106}}\hypertarget{lh:MFT106}{}\enspace{\ttfamily\nolinkurl{MultiFact.exists_independent_card_eq_pow_iff_maxIndependentCard_eq_pow_of_all_card_views}} {\tiny\ttfamily Ssot/\allowbreak MultiFact.lean}
\item \textbf{\nolinkurl{MFT107}}\hypertarget{lh:MFT107}{}\enspace{\ttfamily\nolinkurl{MultiFact.confusable_one_coord_iff_exists_empty_view}} {\tiny\ttfamily Ssot/\allowbreak MultiFact.lean}
\item \textbf{\nolinkurl{MFT108}}\hypertarget{lh:MFT108}{}\enspace{\ttfamily\nolinkurl{MultiFact.graphUpperGap_le_add_of_capacity_ge}} {\tiny\ttfamily Ssot/\allowbreak CapacityLower.lean}
\item \textbf{\nolinkurl{MFT139}}\hypertarget{lh:MFT139}{}\enspace{\ttfamily\nolinkurl{MultiFact.IsIntersectionClosed}} {\tiny\ttfamily Ssot/\allowbreak MultiFact.lean}
\item \textbf{\nolinkurl{MFT140}}\hypertarget{lh:MFT140}{}\enspace{\ttfamily\nolinkurl{MultiFact.agreementUpwardFamily_interClosed_of_meetWitnessed}} {\tiny\ttfamily Ssot/\allowbreak MultiFact.lean}
\item \textbf{\nolinkurl{MFT141}}\hypertarget{lh:MFT141}{}\enspace{\ttfamily\nolinkurl{MultiFact.interClosed_agreementUpwardFamily_confusableTransitive}} {\tiny\ttfamily Ssot/\allowbreak MultiFact.lean}
\item \textbf{\nolinkurl{MFT142}}\hypertarget{lh:MFT142}{}\enspace{\ttfamily\nolinkurl{MultiFact.exists_three_states_with_agreeSets_eq_inter}} {\tiny\ttfamily Ssot/\allowbreak MultiFact.lean}
\item \textbf{\nolinkurl{MFT143}}\hypertarget{lh:MFT143}{}\enspace{\ttfamily\nolinkurl{MultiFact.confusableTransitive_agreementUpwardFamily_interClosed}} {\tiny\ttfamily Ssot/\allowbreak MultiFact.lean}
\item \textbf{\nolinkurl{MFT144}}\hypertarget{lh:MFT144}{}\enspace{\ttfamily\nolinkurl{MultiFact.confusableTransitive_iff_agreementUpwardFamily_interClosed}} {\tiny\ttfamily Ssot/\allowbreak MultiFact.lean}
\item \textbf{\nolinkurl{MFT145}}\hypertarget{lh:MFT145}{}\enspace{\ttfamily\nolinkurl{MultiFact.shannonCapacityReal_eq_shannonLovaszThetaAsymptoticUpper_of_agreementUpwardFamily_interClosed}} {\tiny\ttfamily Ssot/\allowbreak MultiFact.lean}
\item \textbf{\nolinkurl{MFT146}}\hypertarget{lh:MFT146}{}\enspace{\ttfamily\nolinkurl{MultiFact.shannonCapacityReal_eq_log_connectedComponents_of_agreementUpwardFamily_interClosed}} {\tiny\ttfamily Ssot/\allowbreak MultiFact.lean}
\item \textbf{\nolinkurl{MFT147}}\hypertarget{lh:MFT147}{}\enspace{\ttfamily\nolinkurl{MultiFact.shannonLovaszThetaAsymptoticUpper_eq_log_connectedComponents_of_agreementUpwardFamily_interClosed}} {\tiny\ttfamily Ssot/\allowbreak MultiFact.lean}
\item \textbf{\nolinkurl{MFT148}}\hypertarget{lh:MFT148}{}\enspace{\ttfamily\nolinkurl{MultiFact.exists_three_states_with_agreeSets_eq_inter_of_union_eq_univ}} {\tiny\ttfamily Ssot/\allowbreak MultiFact.lean}
\item \textbf{\nolinkurl{MFT149}}\hypertarget{lh:MFT149}{}\enspace{\ttfamily\nolinkurl{MultiFact.agreementUpwardFamily_interClosed_iff_meetWitnessed}} {\tiny\ttfamily Ssot/\allowbreak MultiFact.lean}
\item \textbf{\nolinkurl{MFT150}}\hypertarget{lh:MFT150}{}\enspace{\ttfamily\nolinkurl{MultiFact.confusableTransitive_iff_meetWitnessed}} {\tiny\ttfamily Ssot/\allowbreak MultiFact.lean}
\item \textbf{\nolinkurl{MFT202}}\hypertarget{lh:MFT202}{}\enspace{\ttfamily\nolinkurl{MultiFact.confusabilityGraph_eq_iff_agreementUpwardFamily_eq}} {\tiny\ttfamily Ssot/\allowbreak MultiFact.lean}
\item \textbf{\nolinkurl{OBS1}}\hypertarget{lh:OBS1}{}\enspace{\ttfamily\nolinkurl{FiniteConverse.exact_recovery_implies_pair_injective}} {\tiny\ttfamily Ssot/\allowbreak FiniteConverse.lean}
\item \textbf{\nolinkurl{OBS2}}\hypertarget{lh:OBS2}{}\enspace{\ttfamily\nolinkurl{FiniteConverse.pair_injective_implies_exact_recovery}} {\tiny\ttfamily Ssot/\allowbreak FiniteConverse.lean}
\item \textbf{\nolinkurl{OBS5}}\hypertarget{lh:OBS5}{}\enspace{\ttfamily\nolinkurl{FiniteConverse.clique_card_le_tag_alphabet}} {\tiny\ttfamily Ssot/\allowbreak FiniteConverse.lean}
\item \textbf{\nolinkurl{SCB979}}\hypertarget{lh:SCB979}{}\enspace{\ttfamily\nolinkurl{KernelSectionHardness.SATCriticalBridge.ShapedSATCompleteSpine}} {\tiny\ttfamily Ssot/\allowbreak SATCriticalBridge.lean}
\item \textbf{\nolinkurl{SCB980}}\hypertarget{lh:SCB980}{}\enspace{\ttfamily\nolinkurl{KernelSectionHardness.SATCriticalBridge.satListExternalShell_fixedSATCompleteSpine}} {\tiny\ttfamily Ssot/\allowbreak SATCriticalBridge.lean}
\item \textbf{\nolinkurl{SCB981}}\hypertarget{lh:SCB981}{}\enspace{\ttfamily\nolinkurl{KernelSectionHardness.SATCriticalBridge.satListExternalKernelShell_fixedSATCompleteSpine}} {\tiny\ttfamily Ssot/\allowbreak SATCriticalBridge.lean}
\item \textbf{\nolinkurl{SCB982}}\hypertarget{lh:SCB982}{}\enspace{\ttfamily\nolinkurl{KernelSectionHardness.SATCriticalBridge.satListGatedKernelShell_fixedSATCompleteSpine}} {\tiny\ttfamily Ssot/\allowbreak SATCriticalBridge.lean}
\item \textbf{\nolinkurl{SCB983}}\hypertarget{lh:SCB983}{}\enspace{\ttfamily\nolinkurl{KernelSectionHardness.SATCriticalBridge.satListGatedLatticeShell_fixedSATCompleteSpine}} {\tiny\ttfamily Ssot/\allowbreak SATCriticalBridge.lean}
\item \textbf{\nolinkurl{SCB984}}\hypertarget{lh:SCB984}{}\enspace{\ttfamily\nolinkurl{KernelSectionHardness.SATCriticalBridge.satListSourceGatedPhysical_sameKernelShell_diffTruthCount_of_sameShape}} {\tiny\ttfamily Ssot/\allowbreak SATCriticalBridge.lean}
\item \textbf{\nolinkurl{SCB985}}\hypertarget{lh:SCB985}{}\enspace{\ttfamily\nolinkurl{KernelSectionHardness.SATCriticalBridge.satListSourceGatedPhysical_sameLatticeShell_diffTruthCount_of_sameShape}} {\tiny\ttfamily Ssot/\allowbreak SATCriticalBridge.lean}
\item \textbf{\nolinkurl{SCB986}}\hypertarget{lh:SCB986}{}\enspace{\ttfamily\nolinkurl{KernelSectionHardness.SATCriticalBridge.satListSourceExternalPhysical_sameShell_diffTruthCount_of_sameShape}} {\tiny\ttfamily Ssot/\allowbreak SATCriticalBridge.lean}
\item \textbf{\nolinkurl{SCB987}}\hypertarget{lh:SCB987}{}\enspace{\ttfamily\nolinkurl{KernelSectionHardness.SATCriticalBridge.satListSourceExternalPhysical_sameKernelShell_diffTruthCount_of_sameShape}} {\tiny\ttfamily Ssot/\allowbreak SATCriticalBridge.lean}
\item \textbf{\nolinkurl{SCB988}}\hypertarget{lh:SCB988}{}\enspace{\ttfamily\nolinkurl{KernelSectionHardness.SATConjunctiveBridge.posLitAt}} {\tiny\ttfamily Ssot/\allowbreak SATConjunctiveBridge.lean}
\item \textbf{\nolinkurl{SCB989}}\hypertarget{lh:SCB989}{}\enspace{\ttfamily\nolinkurl{KernelSectionHardness.SATConjunctiveBridge.negLitAt}} {\tiny\ttfamily Ssot/\allowbreak SATConjunctiveBridge.lean}
\item \textbf{\nolinkurl{SCB990}}\hypertarget{lh:SCB990}{}\enspace{\ttfamily\nolinkurl{KernelSectionHardness.SATConjunctiveBridge.posClauseAt}} {\tiny\ttfamily Ssot/\allowbreak SATConjunctiveBridge.lean}
\item \textbf{\nolinkurl{SCB991}}\hypertarget{lh:SCB991}{}\enspace{\ttfamily\nolinkurl{KernelSectionHardness.SATConjunctiveBridge.negClauseAt}} {\tiny\ttfamily Ssot/\allowbreak SATConjunctiveBridge.lean}
\item \textbf{\nolinkurl{SCB992}}\hypertarget{lh:SCB992}{}\enspace{\ttfamily\nolinkurl{KernelSectionHardness.SATConjunctiveBridge.satListYes}} {\tiny\ttfamily Ssot/\allowbreak SATConjunctiveBridge.lean}
\item \textbf{\nolinkurl{SCB993}}\hypertarget{lh:SCB993}{}\enspace{\ttfamily\nolinkurl{KernelSectionHardness.SATConjunctiveBridge.satListContradictoryAt}} {\tiny\ttfamily Ssot/\allowbreak SATConjunctiveBridge.lean}
\item \textbf{\nolinkurl{SCB994}}\hypertarget{lh:SCB994}{}\enspace{\ttfamily\nolinkurl{KernelSectionHardness.SATConjunctiveBridge.satListYes_satisfiable}} {\tiny\ttfamily Ssot/\allowbreak SATConjunctiveBridge.lean}
\item \textbf{\nolinkurl{SCB995}}\hypertarget{lh:SCB995}{}\enspace{\ttfamily\nolinkurl{KernelSectionHardness.SATConjunctiveBridge.satListContradictoryAt_not_satisfiable}} {\tiny\ttfamily Ssot/\allowbreak SATConjunctiveBridge.lean}
\item \textbf{\nolinkurl{SCB996}}\hypertarget{lh:SCB996}{}\enspace{\ttfamily\nolinkurl{KernelSectionHardness.SATCriticalBridge.satListSourceGatedPhysical_sameKernelShell_diffTruthCount_atVar}} {\tiny\ttfamily Ssot/\allowbreak SATCriticalBridge.lean}
\item \textbf{\nolinkurl{SCB997}}\hypertarget{lh:SCB997}{}\enspace{\ttfamily\nolinkurl{KernelSectionHardness.SATCriticalBridge.satListSourceGatedPhysical_sameLatticeShell_diffTruthCount_atVar}} {\tiny\ttfamily Ssot/\allowbreak SATCriticalBridge.lean}
\item \textbf{\nolinkurl{SCB998}}\hypertarget{lh:SCB998}{}\enspace{\ttfamily\nolinkurl{KernelSectionHardness.SATCriticalBridge.satListSourceExternalPhysical_sameShell_diffTruthCount_atVar}} {\tiny\ttfamily Ssot/\allowbreak SATCriticalBridge.lean}
\item \textbf{\nolinkurl{SCB999}}\hypertarget{lh:SCB999}{}\enspace{\ttfamily\nolinkurl{KernelSectionHardness.SATCriticalBridge.satListSourceExternalPhysical_sameKernelShell_diffTruthCount_atVar}} {\tiny\ttfamily Ssot/\allowbreak SATCriticalBridge.lean}
\item \textbf{\nolinkurl{SCB1000}}\hypertarget{lh:SCB1000}{}\enspace{\ttfamily\nolinkurl{KernelSectionHardness.SATCriticalBridge.satListGatedPhysical_kernelSemantic_not_refines_truthCount_atVar}} {\tiny\ttfamily Ssot/\allowbreak SATCriticalBridge.lean}
\item \textbf{\nolinkurl{SCB1001}}\hypertarget{lh:SCB1001}{}\enspace{\ttfamily\nolinkurl{KernelSectionHardness.SATCriticalBridge.satListGatedPhysical_latticeSemantic_not_refines_truthCount_atVar}} {\tiny\ttfamily Ssot/\allowbreak SATCriticalBridge.lean}
\item \textbf{\nolinkurl{SCB1002}}\hypertarget{lh:SCB1002}{}\enspace{\ttfamily\nolinkurl{KernelSectionHardness.SATCriticalBridge.satListExternalPhysical_semantic_not_refines_truthCount_atVar}} {\tiny\ttfamily Ssot/\allowbreak SATCriticalBridge.lean}
\item \textbf{\nolinkurl{SCB1003}}\hypertarget{lh:SCB1003}{}\enspace{\ttfamily\nolinkurl{KernelSectionHardness.SATCriticalBridge.satListExternalPhysical_kernelSemantic_not_refines_truthCount_atVar}} {\tiny\ttfamily Ssot/\allowbreak SATCriticalBridge.lean}
\item \textbf{\nolinkurl{SCB1004}}\hypertarget{lh:SCB1004}{}\enspace{\ttfamily\nolinkurl{SemanticPhysicalBridge.FixedShellTargetReduction.mapShell}} {\tiny\ttfamily Ssot/\allowbreak SemanticPhysicalBridge.lean}
\item \textbf{\nolinkurl{SCB1005}}\hypertarget{lh:SCB1005}{}\enspace{\ttfamily\nolinkurl{SemanticPhysicalBridge.ShapedFixedShellTargetReduction.mapShell}} {\tiny\ttfamily Ssot/\allowbreak SemanticPhysicalBridge.lean}
\item \textbf{\nolinkurl{SCB1006}}\hypertarget{lh:SCB1006}{}\enspace{\ttfamily\nolinkurl{SemanticPhysicalBridge.not_semanticFiner_mapShell_of_not_semanticFiner}} {\tiny\ttfamily Ssot/\allowbreak SemanticPhysicalBridge.lean}
\item \textbf{\nolinkurl{SCB1007}}\hypertarget{lh:SCB1007}{}\enspace{\ttfamily\nolinkurl{KernelSectionHardness.SATCriticalBridge.shapedSATCompleteSpine_mapShell}} {\tiny\ttfamily Ssot/\allowbreak SATCriticalBridge.lean}
\item \textbf{\nolinkurl{SCB1008}}\hypertarget{lh:SCB1008}{}\enspace{\ttfamily\nolinkurl{KernelSectionHardness.SATCriticalBridge.satListExternalShell_fixedSATCompleteSpine_mapShell}} {\tiny\ttfamily Ssot/\allowbreak SATCriticalBridge.lean}
\item \textbf{\nolinkurl{SCB1009}}\hypertarget{lh:SCB1009}{}\enspace{\ttfamily\nolinkurl{KernelSectionHardness.SATCriticalBridge.satListExternalKernelShell_fixedSATCompleteSpine_mapShell}} {\tiny\ttfamily Ssot/\allowbreak SATCriticalBridge.lean}
\item \textbf{\nolinkurl{SCB1010}}\hypertarget{lh:SCB1010}{}\enspace{\ttfamily\nolinkurl{KernelSectionHardness.SATCriticalBridge.satListGatedKernelShell_fixedSATCompleteSpine_mapShell}} {\tiny\ttfamily Ssot/\allowbreak SATCriticalBridge.lean}
\item \textbf{\nolinkurl{SCB1011}}\hypertarget{lh:SCB1011}{}\enspace{\ttfamily\nolinkurl{KernelSectionHardness.SATCriticalBridge.satListGatedLatticeShell_fixedSATCompleteSpine_mapShell}} {\tiny\ttfamily Ssot/\allowbreak SATCriticalBridge.lean}
\item \textbf{\nolinkurl{SCB1012}}\hypertarget{lh:SCB1012}{}\enspace{\ttfamily\nolinkurl{KernelSectionHardness.SATCriticalBridge.satListFullFixedExternalLatticeResult}} {\tiny\ttfamily Ssot/\allowbreak SATCriticalBridge.lean}
\item \textbf{\nolinkurl{SCB1013}}\hypertarget{lh:SCB1013}{}\enspace{\ttfamily\nolinkurl{KernelSectionHardness.SATCriticalBridge.satListFullFixedGatedLatticeResult}} {\tiny\ttfamily Ssot/\allowbreak SATCriticalBridge.lean}
\item \textbf{\nolinkurl{SCR67}}\hypertarget{lh:SCR67}{}\enspace{\ttfamily\nolinkurl{KernelSectionHardness.SATCriticalBridge.SATListSource}} {\tiny\ttfamily Ssot/\allowbreak SATCriticalBridge.lean}
\item \textbf{\nolinkurl{SCR68}}\hypertarget{lh:SCR68}{}\enspace{\ttfamily\nolinkurl{KernelSectionHardness.SATCriticalBridge.SATListSource.shape}} {\tiny\ttfamily Ssot/\allowbreak SATCriticalBridge.lean}
\item \textbf{\nolinkurl{SCR69}}\hypertarget{lh:SCR69}{}\enspace{\ttfamily\nolinkurl{KernelSectionHardness.SATCriticalBridge.SATListSource.satisfiable}} {\tiny\ttfamily Ssot/\allowbreak SATCriticalBridge.lean}
\item \textbf{\nolinkurl{SCR70}}\hypertarget{lh:SCR70}{}\enspace{\ttfamily\nolinkurl{KernelSectionHardness.SATCriticalBridge.SATListSource.satisfyingCount}} {\tiny\ttfamily Ssot/\allowbreak SATCriticalBridge.lean}
\item \textbf{\nolinkurl{SCR71}}\hypertarget{lh:SCR71}{}\enspace{\ttfamily\nolinkurl{KernelSectionHardness.SATCriticalBridge.SATExternalPhysical}} {\tiny\ttfamily Ssot/\allowbreak SATCriticalBridge.lean}
\item \textbf{\nolinkurl{SCR72}}\hypertarget{lh:SCR72}{}\enspace{\ttfamily\nolinkurl{KernelSectionHardness.SATCriticalBridge.SATExternalPhysical.positive}} {\tiny\ttfamily Ssot/\allowbreak SATCriticalBridge.lean}
\item \textbf{\nolinkurl{SCR73}}\hypertarget{lh:SCR73}{}\enspace{\ttfamily\nolinkurl{KernelSectionHardness.SATCriticalBridge.SATExternalPhysical.witnessCount}} {\tiny\ttfamily Ssot/\allowbreak SATCriticalBridge.lean}
\item \textbf{\nolinkurl{SCR74}}\hypertarget{lh:SCR74}{}\enspace{\ttfamily\nolinkurl{KernelSectionHardness.SATCriticalBridge.SATExternalShell}} {\tiny\ttfamily Ssot/\allowbreak SATCriticalBridge.lean}
\item \textbf{\nolinkurl{SCR75}}\hypertarget{lh:SCR75}{}\enspace{\ttfamily\nolinkurl{KernelSectionHardness.SATCriticalBridge.SATExternalPhysical.semantic}} {\tiny\ttfamily Ssot/\allowbreak SATCriticalBridge.lean}
\item \textbf{\nolinkurl{SCR76}}\hypertarget{lh:SCR76}{}\enspace{\ttfamily\nolinkurl{KernelSectionHardness.SATCriticalBridge.satListCompleteExternalIndex}} {\tiny\ttfamily Ssot/\allowbreak SATCriticalBridge.lean}
\item \textbf{\nolinkurl{SCR77}}\hypertarget{lh:SCR77}{}\enspace{\ttfamily\nolinkurl{KernelSectionHardness.SATCriticalBridge.satListSourceExternalPhysical}} {\tiny\ttfamily Ssot/\allowbreak SATCriticalBridge.lean}
\item \textbf{\nolinkurl{SCR78}}\hypertarget{lh:SCR78}{}\enspace{\ttfamily\nolinkurl{KernelSectionHardness.SATCriticalBridge.satListSourceExternalShellOf}} {\tiny\ttfamily Ssot/\allowbreak SATCriticalBridge.lean}
\item \textbf{\nolinkurl{SCR79}}\hypertarget{lh:SCR79}{}\enspace{\ttfamily\nolinkurl{KernelSectionHardness.SATCriticalBridge.satListShapedFixedExternalMultiwayShellTargetReduction_witnessCount}} {\tiny\ttfamily Ssot/\allowbreak SATCriticalBridge.lean}
\item \textbf{\nolinkurl{SCR80}}\hypertarget{lh:SCR80}{}\enspace{\ttfamily\nolinkurl{KernelSectionHardness.SATCriticalBridge.satListShapedFixedExternalMultiwayShellTargetReduction_truthValue}} {\tiny\ttfamily Ssot/\allowbreak SATCriticalBridge.lean}
\item \textbf{\nolinkurl{SCR81}}\hypertarget{lh:SCR81}{}\enspace{\ttfamily\nolinkurl{KernelSectionHardness.SATCriticalBridge.satListSourceExternalPhysical_semantic_eq_shellOf}} {\tiny\ttfamily Ssot/\allowbreak SATCriticalBridge.lean}
\item \textbf{\nolinkurl{SCR82}}\hypertarget{lh:SCR82}{}\enspace{\ttfamily\nolinkurl{KernelSectionHardness.SATCriticalBridge.satListSourceExternalPhysical_witnessCount_eq_satisfyingCount}} {\tiny\ttfamily Ssot/\allowbreak SATCriticalBridge.lean}
\item \textbf{\nolinkurl{SCR83}}\hypertarget{lh:SCR83}{}\enspace{\ttfamily\nolinkurl{KernelSectionHardness.SATCriticalBridge.satListSourceExternalPhysical_truthValue_eq_satisfiable}} {\tiny\ttfamily Ssot/\allowbreak SATCriticalBridge.lean}
\item \textbf{\nolinkurl{SCR84}}\hypertarget{lh:SCR84}{}\enspace{\ttfamily\nolinkurl{KernelSectionHardness.SATCriticalBridge.ShapedSharpSATCountHard}} {\tiny\ttfamily Ssot/\allowbreak SATCriticalBridge.lean}
\item \textbf{\nolinkurl{SCR85}}\hypertarget{lh:SCR85}{}\enspace{\ttfamily\nolinkurl{KernelSectionHardness.SATCriticalBridge.ShapedSATDecisionHard}} {\tiny\ttfamily Ssot/\allowbreak SATCriticalBridge.lean}
\item \textbf{\nolinkurl{SCR86}}\hypertarget{lh:SCR86}{}\enspace{\ttfamily\nolinkurl{KernelSectionHardness.SATCriticalBridge.satListExternalWitnessCount_shapedSharpSATCountHard}} {\tiny\ttfamily Ssot/\allowbreak SATCriticalBridge.lean}
\item \textbf{\nolinkurl{SCR87}}\hypertarget{lh:SCR87}{}\enspace{\ttfamily\nolinkurl{KernelSectionHardness.SATCriticalBridge.satListExternalPositive_shapedSATDecisionHard}} {\tiny\ttfamily Ssot/\allowbreak SATCriticalBridge.lean}
\item \textbf{\nolinkurl{SPB1}}\hypertarget{lh:SPB1}{}\enspace{\ttfamily\nolinkurl{SemanticPhysicalBridge.SemanticComplete}} {\tiny\ttfamily Ssot/\allowbreak SemanticPhysicalBridge.lean}
\item \textbf{\nolinkurl{SPB2}}\hypertarget{lh:SPB2}{}\enspace{\ttfamily\nolinkurl{SemanticPhysicalBridge.not_semanticComplete_of_same_semantic_diff_property}} {\tiny\ttfamily Ssot/\allowbreak SemanticPhysicalBridge.lean}
\item \textbf{\nolinkurl{SPB3}}\hypertarget{lh:SPB3}{}\enspace{\ttfamily\nolinkurl{SemanticPhysicalBridge.FixedShellReduction}} {\tiny\ttfamily Ssot/\allowbreak SemanticPhysicalBridge.lean}
\item \textbf{\nolinkurl{SPB4}}\hypertarget{lh:SPB4}{}\enspace{\ttfamily\nolinkurl{SemanticPhysicalBridge.fixedShellReduction_same_semantic}} {\tiny\ttfamily Ssot/\allowbreak SemanticPhysicalBridge.lean}
\item \textbf{\nolinkurl{SPB5}}\hypertarget{lh:SPB5}{}\enspace{\ttfamily\nolinkurl{SemanticPhysicalBridge.not_semanticComplete_of_fixedShellReduction_separates}} {\tiny\ttfamily Ssot/\allowbreak SemanticPhysicalBridge.lean}
\item \textbf{\nolinkurl{SPB6}}\hypertarget{lh:SPB6}{}\enspace{\ttfamily\nolinkurl{SemanticPhysicalBridge.fourLine_same_semantic_shell}} {\tiny\ttfamily Ssot/\allowbreak SemanticPhysicalBridge.lean}
\item \textbf{\nolinkurl{SPB7}}\hypertarget{lh:SPB7}{}\enspace{\ttfamily\nolinkurl{SemanticPhysicalBridge.fourLine_collinear_no_avoiding_plane}} {\tiny\ttfamily Ssot/\allowbreak SemanticPhysicalBridge.lean}
\item \textbf{\nolinkurl{SPB8}}\hypertarget{lh:SPB8}{}\enspace{\ttfamily\nolinkurl{SemanticPhysicalBridge.fourLine_general_has_avoiding_plane}} {\tiny\ttfamily Ssot/\allowbreak SemanticPhysicalBridge.lean}
\item \textbf{\nolinkurl{SPB9}}\hypertarget{lh:SPB9}{}\enspace{\ttfamily\nolinkurl{SemanticPhysicalBridge.fourLine_semanticShell_not_complete}} {\tiny\ttfamily Ssot/\allowbreak SemanticPhysicalBridge.lean}
\item \textbf{\nolinkurl{SPB10}}\hypertarget{lh:SPB10}{}\enspace{\ttfamily\nolinkurl{SemanticPhysicalBridge.graphSlot_same_semantic_shell}} {\tiny\ttfamily Ssot/\allowbreak SemanticPhysicalBridge.lean}
\item \textbf{\nolinkurl{SPB11}}\hypertarget{lh:SPB11}{}\enspace{\ttfamily\nolinkurl{SemanticPhysicalBridge.graphSlot_blocked_no_avoiding_plane}} {\tiny\ttfamily Ssot/\allowbreak SemanticPhysicalBridge.lean}
\item \textbf{\nolinkurl{SPB12}}\hypertarget{lh:SPB12}{}\enspace{\ttfamily\nolinkurl{SemanticPhysicalBridge.graphSlot_open_has_avoiding_plane}} {\tiny\ttfamily Ssot/\allowbreak SemanticPhysicalBridge.lean}
\item \textbf{\nolinkurl{SPB13}}\hypertarget{lh:SPB13}{}\enspace{\ttfamily\nolinkurl{SemanticPhysicalBridge.graphSlot_semanticShell_not_complete}} {\tiny\ttfamily Ssot/\allowbreak SemanticPhysicalBridge.lean}
\item \textbf{\nolinkurl{SPB14}}\hypertarget{lh:SPB14}{}\enspace{\ttfamily\nolinkurl{SemanticPhysicalBridge.fullLatticeGraphSlot_same_shell}} {\tiny\ttfamily Ssot/\allowbreak SemanticPhysicalBridge.lean}
\item \textbf{\nolinkurl{SPB15}}\hypertarget{lh:SPB15}{}\enspace{\ttfamily\nolinkurl{SemanticPhysicalBridge.fullLatticeGraphSlot_open_has_avoiding_plane}} {\tiny\ttfamily Ssot/\allowbreak SemanticPhysicalBridge.lean}
\item \textbf{\nolinkurl{SPB16}}\hypertarget{lh:SPB16}{}\enspace{\ttfamily\nolinkurl{SemanticPhysicalBridge.fullLatticeGraphSlot_blocked_no_avoiding_plane}} {\tiny\ttfamily Ssot/\allowbreak SemanticPhysicalBridge.lean}
\item \textbf{\nolinkurl{SPB17}}\hypertarget{lh:SPB17}{}\enspace{\ttfamily\nolinkurl{SemanticPhysicalBridge.fullLatticeGraphSlot_shell_not_complete}} {\tiny\ttfamily Ssot/\allowbreak SemanticPhysicalBridge.lean}
\item \textbf{\nolinkurl{SPB18}}\hypertarget{lh:SPB18}{}\enspace{\ttfamily\nolinkurl{SemanticPhysicalBridge.FixedShellWitnessReduction}} {\tiny\ttfamily Ssot/\allowbreak SemanticPhysicalBridge.lean}
\item \textbf{\nolinkurl{SPB19}}\hypertarget{lh:SPB19}{}\enspace{\ttfamily\nolinkurl{SemanticPhysicalBridge.FixedShellWitnessReduction.toFixedShellReduction}} {\tiny\ttfamily Ssot/\allowbreak SemanticPhysicalBridge.lean}
\item \textbf{\nolinkurl{SPB20}}\hypertarget{lh:SPB20}{}\enspace{\ttfamily\nolinkurl{SemanticPhysicalBridge.fixedShellWitnessReduction_exists_correct}} {\tiny\ttfamily Ssot/\allowbreak SemanticPhysicalBridge.lean}
\item \textbf{\nolinkurl{SPB21}}\hypertarget{lh:SPB21}{}\enspace{\ttfamily\nolinkurl{SemanticPhysicalBridge.SatisfiesAll}} {\tiny\ttfamily Ssot/\allowbreak SemanticPhysicalBridge.lean}
\item \textbf{\nolinkurl{SPB22}}\hypertarget{lh:SPB22}{}\enspace{\ttfamily\nolinkurl{SemanticPhysicalBridge.satisfiesAll_congr}} {\tiny\ttfamily Ssot/\allowbreak SemanticPhysicalBridge.lean}
\item \textbf{\nolinkurl{SPB23}}\hypertarget{lh:SPB23}{}\enspace{\ttfamily\nolinkurl{SemanticPhysicalBridge.ConjunctiveFixedShellReduction}} {\tiny\ttfamily Ssot/\allowbreak SemanticPhysicalBridge.lean}
\item \textbf{\nolinkurl{SPB24}}\hypertarget{lh:SPB24}{}\enspace{\ttfamily\nolinkurl{SemanticPhysicalBridge.conjunctiveFixedShellReduction_exists_correct}} {\tiny\ttfamily Ssot/\allowbreak SemanticPhysicalBridge.lean}
\item \textbf{\nolinkurl{SPB25}}\hypertarget{lh:SPB25}{}\enspace{\ttfamily\nolinkurl{SemanticPhysicalBridge.conjunctiveFixedShellReduction_same_semantic}} {\tiny\ttfamily Ssot/\allowbreak SemanticPhysicalBridge.lean}
\item \textbf{\nolinkurl{SPB26}}\hypertarget{lh:SPB26}{}\enspace{\ttfamily\nolinkurl{SemanticPhysicalBridge.fixedShellReduction_of_commonCoreLift}} {\tiny\ttfamily Ssot/\allowbreak SemanticPhysicalBridge.lean}
\item \textbf{\nolinkurl{SPB144}}\hypertarget{lh:SPB144}{}\enspace{\ttfamily\nolinkurl{SemanticPhysicalBridge.commonCoreFourPointScan_fixed_semanticShell}} {\tiny\ttfamily Ssot/\allowbreak SemanticPhysicalBridge.lean}
\item \textbf{\nolinkurl{SPB149}}\hypertarget{lh:SPB149}{}\enspace{\ttfamily\nolinkurl{SemanticPhysicalBridge.commonCoreFourPointScan_shell_not_finer_than_forbiddenPointMatroidShell}} {\tiny\ttfamily Ssot/\allowbreak SemanticPhysicalBridge.lean}
\item \textbf{\nolinkurl{SPB153}}\hypertarget{lh:SPB153}{}\enspace{\ttfamily\nolinkurl{SemanticPhysicalBridge.commonCoreFourPointScanCountTargetReduction}} {\tiny\ttfamily Ssot/\allowbreak SemanticPhysicalBridge.lean}
\item \textbf{\nolinkurl{SPB154}}\hypertarget{lh:SPB154}{}\enspace{\ttfamily\nolinkurl{SemanticPhysicalBridge.commonCoreFourPointScanForbiddenPointMatroidTargetReduction}} {\tiny\ttfamily Ssot/\allowbreak SemanticPhysicalBridge.lean}
\item \textbf{\nolinkurl{SPB155}}\hypertarget{lh:SPB155}{}\enspace{\ttfamily\nolinkurl{SemanticPhysicalBridge.commonCoreFourPointScanHasAvoidingPlane}} {\tiny\ttfamily Ssot/\allowbreak SemanticPhysicalBridge.lean}
\item \textbf{\nolinkurl{SPB157}}\hypertarget{lh:SPB157}{}\enspace{\ttfamily\nolinkurl{SemanticPhysicalBridge.commonCoreFourPointScan_shell_not_complete_for_avoidingPlanePositive}} {\tiny\ttfamily Ssot/\allowbreak SemanticPhysicalBridge.lean}
\item \textbf{\nolinkurl{SPB159}}\hypertarget{lh:SPB159}{}\enspace{\ttfamily\nolinkurl{SemanticPhysicalBridge.commonCoreFourPointScan_forbiddenPointMatroidShell_controls_avoidingPlaneCount}} {\tiny\ttfamily Ssot/\allowbreak SemanticPhysicalBridge.lean}
\item \textbf{\nolinkurl{SPB160}}\hypertarget{lh:SPB160}{}\enspace{\ttfamily\nolinkurl{SemanticPhysicalBridge.ShapedFixedShellTargetReduction}} {\tiny\ttfamily Ssot/\allowbreak SemanticPhysicalBridge.lean}
\item \textbf{\nolinkurl{SPB161}}\hypertarget{lh:SPB161}{}\enspace{\ttfamily\nolinkurl{SemanticPhysicalBridge.shapedFixedShellTargetReduction_same_semantic_of_same_shape}} {\tiny\ttfamily Ssot/\allowbreak SemanticPhysicalBridge.lean}
\item \textbf{\nolinkurl{SPB162}}\hypertarget{lh:SPB162}{}\enspace{\ttfamily\nolinkurl{SemanticPhysicalBridge.ShapedFixedShellTargetReduction.onShape}} {\tiny\ttfamily Ssot/\allowbreak SemanticPhysicalBridge.lean}
\item \textbf{\nolinkurl{SPB163}}\hypertarget{lh:SPB163}{}\enspace{\ttfamily\nolinkurl{SemanticPhysicalBridge.exists_sameSemantic_diffTarget_of_shapedFixedShellTargetReduction_separates}} {\tiny\ttfamily Ssot/\allowbreak SemanticPhysicalBridge.lean}
\item \textbf{\nolinkurl{SPB164}}\hypertarget{lh:SPB164}{}\enspace{\ttfamily\nolinkurl{SemanticPhysicalBridge.not_semanticFiner_of_shapedFixedShellTargetReduction_separates}} {\tiny\ttfamily Ssot/\allowbreak SemanticPhysicalBridge.lean}
\item \textbf{\nolinkurl{SPB165}}\hypertarget{lh:SPB165}{}\enspace{\ttfamily\nolinkurl{SemanticPhysicalBridge.shapedFixedShellTargetReduction_finiteSourceTargetValues_subset_targetValuesInShapeFiber}} {\tiny\ttfamily Ssot/\allowbreak SemanticPhysicalBridge.lean}
\item \textbf{\nolinkurl{SPB166}}\hypertarget{lh:SPB166}{}\enspace{\ttfamily\nolinkurl{SemanticPhysicalBridge.shapedFixedShellTargetReduction_finiteSourceTargetValues_subset_finiteTargetValuesInShapeFiber}} {\tiny\ttfamily Ssot/\allowbreak SemanticPhysicalBridge.lean}
\item \textbf{\nolinkurl{SPB167}}\hypertarget{lh:SPB167}{}\enspace{\ttfamily\nolinkurl{SemanticPhysicalBridge.shapedFixedShellTargetReduction_finiteSourceTargetValues_card_le_finiteTargetValuesInShapeFiber_card}} {\tiny\ttfamily Ssot/\allowbreak SemanticPhysicalBridge.lean}
\item \textbf{\nolinkurl{SPB168}}\hypertarget{lh:SPB168}{}\enspace{\ttfamily\nolinkurl{SemanticPhysicalBridge.shapedFixedShellTargetReduction_exists_targetValuesSubset_card_eq_of_injOn_same_shape}} {\tiny\ttfamily Ssot/\allowbreak SemanticPhysicalBridge.lean}
\item \textbf{\nolinkurl{SPB169}}\hypertarget{lh:SPB169}{}\enspace{\ttfamily\nolinkurl{SemanticPhysicalBridge.not_semanticFiner_of_shapedFixedShellTargetReduction_sourceTarget_injOn_same_shape_card_gt_one}} {\tiny\ttfamily Ssot/\allowbreak SemanticPhysicalBridge.lean}
\item \textbf{\nolinkurl{SPB170}}\hypertarget{lh:SPB170}{}\enspace{\ttfamily\nolinkurl{SemanticPhysicalBridge.FixedShellTargetReduction.toShaped}} {\tiny\ttfamily Ssot/\allowbreak SemanticPhysicalBridge.lean}
\item \textbf{\nolinkurl{SPB171}}\hypertarget{lh:SPB171}{}\enspace{\ttfamily\nolinkurl{SemanticPhysicalBridge.ShapedFixedShellTargetReduction.precompose}} {\tiny\ttfamily Ssot/\allowbreak SemanticPhysicalBridge.lean}
\item \textbf{\nolinkurl{SPB172}}\hypertarget{lh:SPB172}{}\enspace{\ttfamily\nolinkurl{SemanticPhysicalBridge.shapedFixedShellTargetReduction_exists_two_targetValuesInShapeFiber_of_finiteSourceTargetValues_card_gt_one}} {\tiny\ttfamily Ssot/\allowbreak SemanticPhysicalBridge.lean}
\item \textbf{\nolinkurl{SPB173}}\hypertarget{lh:SPB173}{}\enspace{\ttfamily\nolinkurl{SemanticPhysicalBridge.not_targetValuesInShapeFiber_subsingleton_of_shapedFixedShellTargetReduction_finiteSourceTargetValues_card_gt_one}} {\tiny\ttfamily Ssot/\allowbreak SemanticPhysicalBridge.lean}
\item \textbf{\nolinkurl{SPB174}}\hypertarget{lh:SPB174}{}\enspace{\ttfamily\nolinkurl{SemanticPhysicalBridge.not_semanticFiner_of_shapedFixedShellTargetReduction_finiteSourceTargetValues_card_gt_one}} {\tiny\ttfamily Ssot/\allowbreak SemanticPhysicalBridge.lean}
\item \textbf{\nolinkurl{SPB175}}\hypertarget{lh:SPB175}{}\enspace{\ttfamily\nolinkurl{SemanticPhysicalBridge.FixedShellTargetReduction.mapTarget}} {\tiny\ttfamily Ssot/\allowbreak SemanticPhysicalBridge.lean}
\item \textbf{\nolinkurl{SPB176}}\hypertarget{lh:SPB176}{}\enspace{\ttfamily\nolinkurl{SemanticPhysicalBridge.satisfyingWitnesses}} {\tiny\ttfamily Ssot/\allowbreak SemanticPhysicalBridge.lean}
\item \textbf{\nolinkurl{SPB177}}\hypertarget{lh:SPB177}{}\enspace{\ttfamily\nolinkurl{SemanticPhysicalBridge.satisfyingWitnesses_eq_of_local_iff}} {\tiny\ttfamily Ssot/\allowbreak SemanticPhysicalBridge.lean}
\item \textbf{\nolinkurl{SPB178}}\hypertarget{lh:SPB178}{}\enspace{\ttfamily\nolinkurl{SemanticPhysicalBridge.fixedShellTargetReduction_of_conjunctiveFixedShellReduction_count}} {\tiny\ttfamily Ssot/\allowbreak SemanticPhysicalBridge.lean}
\item \textbf{\nolinkurl{SPB179}}\hypertarget{lh:SPB179}{}\enspace{\ttfamily\nolinkurl{SemanticPhysicalBridge.incidenceCountProfileShell}} {\tiny\ttfamily Ssot/\allowbreak SemanticPhysicalBridge.lean}
\item \textbf{\nolinkurl{SPB180}}\hypertarget{lh:SPB180}{}\enspace{\ttfamily\nolinkurl{SemanticPhysicalBridge.avoidanceCount_eq_incidenceCountProfileShell_zero}} {\tiny\ttfamily Ssot/\allowbreak SemanticPhysicalBridge.lean}
\item \textbf{\nolinkurl{SPB181}}\hypertarget{lh:SPB181}{}\enspace{\ttfamily\nolinkurl{SemanticPhysicalBridge.avoidanceCount_eq_of_same_incidenceCountProfileShell}} {\tiny\ttfamily Ssot/\allowbreak SemanticPhysicalBridge.lean}
\item \textbf{\nolinkurl{SPB182}}\hypertarget{lh:SPB182}{}\enspace{\ttfamily\nolinkurl{SemanticPhysicalBridge.avoidanceCount_positive_iff_of_same_incidenceCountProfileShell}} {\tiny\ttfamily Ssot/\allowbreak SemanticPhysicalBridge.lean}
\item \textbf{\nolinkurl{SPB183}}\hypertarget{lh:SPB183}{}\enspace{\ttfamily\nolinkurl{SemanticPhysicalBridge.avoidanceCountValue_semanticComplete_incidenceCountProfileShell}} {\tiny\ttfamily Ssot/\allowbreak SemanticPhysicalBridge.lean}
\item \textbf{\nolinkurl{SPB184}}\hypertarget{lh:SPB184}{}\enspace{\ttfamily\nolinkurl{SemanticPhysicalBridge.avoidanceCountPositive_semanticComplete_incidenceCountProfileShell}} {\tiny\ttfamily Ssot/\allowbreak SemanticPhysicalBridge.lean}
\item \textbf{\nolinkurl{SPB185}}\hypertarget{lh:SPB185}{}\enspace{\ttfamily\nolinkurl{SemanticPhysicalBridge.semanticFiner_incidenceCountProfileShell_or_fiberCollision}} {\tiny\ttfamily Ssot/\allowbreak SemanticPhysicalBridge.lean}
\item \textbf{\nolinkurl{SPB186}}\hypertarget{lh:SPB186}{}\enspace{\ttfamily\nolinkurl{SemanticPhysicalBridge.avoidanceCountValue_control_or_incidenceCountProfileFiberCollision}} {\tiny\ttfamily Ssot/\allowbreak SemanticPhysicalBridge.lean}
\item \textbf{\nolinkurl{SPB187}}\hypertarget{lh:SPB187}{}\enspace{\ttfamily\nolinkurl{SemanticPhysicalBridge.avoidanceCountPositive_control_or_incidenceCountProfileFiberCollision}} {\tiny\ttfamily Ssot/\allowbreak SemanticPhysicalBridge.lean}
\item \textbf{\nolinkurl{SPB188}}\hypertarget{lh:SPB188}{}\enspace{\ttfamily\nolinkurl{SemanticPhysicalBridge.rankSieve_eq_avoidanceCountOf_of_rank_count}} {\tiny\ttfamily Ssot/\allowbreak SemanticPhysicalBridge.lean}
\item \textbf{\nolinkurl{SPB189}}\hypertarget{lh:SPB189}{}\enspace{\ttfamily\nolinkurl{SemanticPhysicalBridge.rankSieve_eq_incidenceCountProfileShell_zero_of_rank_count}} {\tiny\ttfamily Ssot/\allowbreak SemanticPhysicalBridge.lean}
\item \textbf{\nolinkurl{SPB190}}\hypertarget{lh:SPB190}{}\enspace{\ttfamily\nolinkurl{SemanticPhysicalBridge.rankSieve_eq_of_same_incidenceCountProfileShell_of_rank_count}} {\tiny\ttfamily Ssot/\allowbreak SemanticPhysicalBridge.lean}
\item \textbf{\nolinkurl{SPB191}}\hypertarget{lh:SPB191}{}\enspace{\ttfamily\nolinkurl{SemanticPhysicalBridge.rankSieve_positive_iff_of_same_incidenceCountProfileShell_of_rank_count}} {\tiny\ttfamily Ssot/\allowbreak SemanticPhysicalBridge.lean}
\item \textbf{\nolinkurl{SPB192}}\hypertarget{lh:SPB192}{}\enspace{\ttfamily\nolinkurl{SemanticPhysicalBridge.rankSieveValue_semanticComplete_incidenceCountProfileShell_of_rank_count}} {\tiny\ttfamily Ssot/\allowbreak SemanticPhysicalBridge.lean}
\item \textbf{\nolinkurl{SPB193}}\hypertarget{lh:SPB193}{}\enspace{\ttfamily\nolinkurl{SemanticPhysicalBridge.rankSievePositive_semanticComplete_incidenceCountProfileShell_of_rank_count}} {\tiny\ttfamily Ssot/\allowbreak SemanticPhysicalBridge.lean}
\item \textbf{\nolinkurl{SPB194}}\hypertarget{lh:SPB194}{}\enspace{\ttfamily\nolinkurl{SemanticPhysicalBridge.rankSieveValue_control_or_incidenceCountProfileFiberCollision_of_rank_count}} {\tiny\ttfamily Ssot/\allowbreak SemanticPhysicalBridge.lean}
\item \textbf{\nolinkurl{SPB195}}\hypertarget{lh:SPB195}{}\enspace{\ttfamily\nolinkurl{SemanticPhysicalBridge.rankSievePositive_control_or_incidenceCountProfileFiberCollision_of_rank_count}} {\tiny\ttfamily Ssot/\allowbreak SemanticPhysicalBridge.lean}
\item \textbf{\nolinkurl{SPB196}}\hypertarget{lh:SPB196}{}\enspace{\ttfamily\nolinkurl{SemanticPhysicalBridge.incidenceCountProfileFiberCollision_of_sameSemantic_rankSieve_ne_of_rank_count}} {\tiny\ttfamily Ssot/\allowbreak SemanticPhysicalBridge.lean}
\item \textbf{\nolinkurl{SPB197}}\hypertarget{lh:SPB197}{}\enspace{\ttfamily\nolinkurl{SemanticPhysicalBridge.not_semanticFiner_incidenceCountProfileShell_of_sameSemantic_rankSieve_ne_of_rank_count}} {\tiny\ttfamily Ssot/\allowbreak SemanticPhysicalBridge.lean}
\item \textbf{\nolinkurl{SPB198}}\hypertarget{lh:SPB198}{}\enspace{\ttfamily\nolinkurl{SemanticPhysicalBridge.incidenceCountProfileFiberCollision_of_sameSemantic_rankSievePositive_diff_of_rank_count}} {\tiny\ttfamily Ssot/\allowbreak SemanticPhysicalBridge.lean}
\item \textbf{\nolinkurl{SPB199}}\hypertarget{lh:SPB199}{}\enspace{\ttfamily\nolinkurl{SemanticPhysicalBridge.not_semanticFiner_incidenceCountProfileShell_of_sameSemantic_rankSievePositive_diff_of_rank_count}} {\tiny\ttfamily Ssot/\allowbreak SemanticPhysicalBridge.lean}
\item \textbf{\nolinkurl{SPB200}}\hypertarget{lh:SPB200}{}\enspace{\ttfamily\nolinkurl{SemanticPhysicalBridge.finiteConjunctiveSatisfyingWitnessCountValues}} {\tiny\ttfamily Ssot/\allowbreak SemanticPhysicalBridge.lean}
\item \textbf{\nolinkurl{SPB201}}\hypertarget{lh:SPB201}{}\enspace{\ttfamily\nolinkurl{SemanticPhysicalBridge.mem_finiteConjunctiveSatisfyingWitnessCountValues}} {\tiny\ttfamily Ssot/\allowbreak SemanticPhysicalBridge.lean}
\item \textbf{\nolinkurl{SPB202}}\hypertarget{lh:SPB202}{}\enspace{\ttfamily\nolinkurl{SemanticPhysicalBridge.finiteConjunctiveSatisfyingWitnessCountValues_subset_fixedShellPhysicalCountFiber}} {\tiny\ttfamily Ssot/\allowbreak SemanticPhysicalBridge.lean}
\item \textbf{\nolinkurl{SPB203}}\hypertarget{lh:SPB203}{}\enspace{\ttfamily\nolinkurl{SemanticPhysicalBridge.exists_countValues_subset_fixedShellPhysicalCountFiber_of_conjunctive_count_injOn}} {\tiny\ttfamily Ssot/\allowbreak SemanticPhysicalBridge.lean}
\item \textbf{\nolinkurl{SPB204}}\hypertarget{lh:SPB204}{}\enspace{\ttfamily\nolinkurl{SemanticPhysicalBridge.not_semanticFiner_physicalCount_of_conjunctive_countValues_card_gt_one}} {\tiny\ttfamily Ssot/\allowbreak SemanticPhysicalBridge.lean}
\item \textbf{\nolinkurl{SPB205}}\hypertarget{lh:SPB205}{}\enspace{\ttfamily\nolinkurl{SemanticPhysicalBridge.not_semanticFiner_physicalCount_of_conjunctive_count_injOn_card_gt_one}} {\tiny\ttfamily Ssot/\allowbreak SemanticPhysicalBridge.lean}
\item \textbf{\nolinkurl{SPB240}}\hypertarget{lh:SPB240}{}\enspace{\ttfamily\nolinkurl{SemanticPhysicalBridge.signedRankProfileShell}} {\tiny\ttfamily Ssot/\allowbreak SemanticPhysicalBridge.lean}
\item \textbf{\nolinkurl{SPB241}}\hypertarget{lh:SPB241}{}\enspace{\ttfamily\nolinkurl{SemanticPhysicalBridge.semanticFiner_forbiddenPointMatroidShell_signedRankProfileShell}} {\tiny\ttfamily Ssot/\allowbreak SemanticPhysicalBridge.lean}
\item \textbf{\nolinkurl{SPB242}}\hypertarget{lh:SPB242}{}\enspace{\ttfamily\nolinkurl{SemanticPhysicalBridge.rankSieve_eq_of_same_signedRankProfileShell}} {\tiny\ttfamily Ssot/\allowbreak SemanticPhysicalBridge.lean}
\item \textbf{\nolinkurl{SPB243}}\hypertarget{lh:SPB243}{}\enspace{\ttfamily\nolinkurl{SemanticPhysicalBridge.rankSieve_positive_iff_of_same_signedRankProfileShell}} {\tiny\ttfamily Ssot/\allowbreak SemanticPhysicalBridge.lean}
\item \textbf{\nolinkurl{SPB244}}\hypertarget{lh:SPB244}{}\enspace{\ttfamily\nolinkurl{SemanticPhysicalBridge.rankSieveValue_semanticComplete_signedRankProfileShell}} {\tiny\ttfamily Ssot/\allowbreak SemanticPhysicalBridge.lean}
\item \textbf{\nolinkurl{SPB245}}\hypertarget{lh:SPB245}{}\enspace{\ttfamily\nolinkurl{SemanticPhysicalBridge.rankSievePositive_semanticComplete_signedRankProfileShell}} {\tiny\ttfamily Ssot/\allowbreak SemanticPhysicalBridge.lean}
\item \textbf{\nolinkurl{SPB288}}\hypertarget{lh:SPB288}{}\enspace{\ttfamily\nolinkurl{SemanticPhysicalBridge.semanticFiner_signedRankProfileShell_rankSieve}} {\tiny\ttfamily Ssot/\allowbreak SemanticPhysicalBridge.lean}
\item \textbf{\nolinkurl{SPB289}}\hypertarget{lh:SPB289}{}\enspace{\ttfamily\nolinkurl{SemanticPhysicalBridge.semanticFiner_signedRankProfileShell_avoidanceCount}} {\tiny\ttfamily Ssot/\allowbreak SemanticPhysicalBridge.lean}
\item \textbf{\nolinkurl{SPB379}}\hypertarget{lh:SPB379}{}\enspace{\ttfamily\nolinkurl{SemanticPhysicalBridge.rankSieve_eq_decode_signedRankProfileShell_of_rank_le_card}} {\tiny\ttfamily Ssot/\allowbreak SemanticPhysicalBridge.lean}
\item \textbf{\nolinkurl{SPB380}}\hypertarget{lh:SPB380}{}\enspace{\ttfamily\nolinkurl{SemanticPhysicalBridge.semanticFiner_signedRankProfileShell_decodedRankSieve}} {\tiny\ttfamily Ssot/\allowbreak SemanticPhysicalBridge.lean}
\item \textbf{\nolinkurl{SPB381}}\hypertarget{lh:SPB381}{}\enspace{\ttfamily\nolinkurl{SemanticPhysicalBridge.rankSieveTargetValuesInSemanticFiber_compression_or_forbiddenPointMatroidCollision}} {\tiny\ttfamily Ssot/\allowbreak SemanticPhysicalBridge.lean}
\item \textbf{\nolinkurl{SPB382}}\hypertarget{lh:SPB382}{}\enspace{\ttfamily\nolinkurl{SemanticPhysicalBridge.rankSieveTargetValuesInSemanticFiber_compression_or_signedRankProfileCollision}} {\tiny\ttfamily Ssot/\allowbreak SemanticPhysicalBridge.lean}
\item \textbf{\nolinkurl{SPB383}}\hypertarget{lh:SPB383}{}\enspace{\ttfamily\nolinkurl{SemanticPhysicalBridge.rankSieve_fixedSemanticShell_compressionDichotomy_with_signedProfileDecoder}} {\tiny\ttfamily Ssot/\allowbreak SemanticPhysicalBridge.lean}
\end{list}
\else
\begin{longtable}{@{}>{\raggedright\arraybackslash}m{0.10\linewidth}>{\raggedright\arraybackslash}m{0.86\linewidth}@{}}
\toprule
\textbf{ID} & \textbf{Lean Handle / Source} \\
\midrule
\endfirsthead
\toprule
\textbf{ID} & \textbf{Lean Handle / Source} \\
\midrule
\endhead
\endfoot
\bottomrule
\endlastfoot
\hypertarget{lh:ABD1}{\textbf{\nolinkurl{ABD1}}} & {\fontsize{8}{9}\selectfont\ttfamily\nolinkurl{ArrangementBoundary.binarySquare_avoidingVectors_card}}\par\vspace{0.1ex}{\fontsize{8}{9}\selectfont\ttfamily Ssot/\allowbreak ArrangementBoundary.lean} \\
\midrule
\hypertarget{lh:ABD2}{\textbf{\nolinkurl{ABD2}}} & {\fontsize{8}{9}\selectfont\ttfamily\nolinkurl{ArrangementBoundary.binarySquare_avoidingVectors_card_eq_inclusionExclusion}}\par\vspace{0.1ex}{\fontsize{8}{9}\selectfont\ttfamily Ssot/\allowbreak ArrangementBoundary.lean} \\
\midrule
\hypertarget{lh:ABD3}{\textbf{\nolinkurl{ABD3}}} & {\fontsize{8}{9}\selectfont\ttfamily\nolinkurl{ArrangementBoundary.avoidingSpanSystem_exchange_failure}}\par\vspace{0.1ex}{\fontsize{8}{9}\selectfont\ttfamily Ssot/\allowbreak ArrangementBoundary.lean} \\
\midrule
\hypertarget{lh:ABD4}{\textbf{\nolinkurl{ABD4}}} & {\fontsize{8}{9}\selectfont\ttfamily\nolinkurl{ArrangementBoundary.allTwoViewF2Four_avoidingLines_card}}\par\vspace{0.1ex}{\fontsize{8}{9}\selectfont\ttfamily Ssot/\allowbreak ArrangementBoundary.lean} \\
\midrule
\hypertarget{lh:ABD5}{\textbf{\nolinkurl{ABD5}}} & {\fontsize{8}{9}\selectfont\ttfamily\nolinkurl{ArrangementBoundary.allTwoViewF2Four_avoidingPlanes_card}}\par\vspace{0.1ex}{\fontsize{8}{9}\selectfont\ttfamily Ssot/\allowbreak ArrangementBoundary.lean} \\
\midrule
\hypertarget{lh:ABD6}{\textbf{\nolinkurl{ABD6}}} & {\fontsize{8}{9}\selectfont\ttfamily\nolinkurl{ArrangementBoundary.f3Three_collinearFourLines_avoidingPlanes_card}}\par\vspace{0.1ex}{\fontsize{8}{9}\selectfont\ttfamily Ssot/\allowbreak ArrangementBoundary.lean} \\
\midrule
\hypertarget{lh:ABD7}{\textbf{\nolinkurl{ABD7}}} & {\fontsize{8}{9}\selectfont\ttfamily\nolinkurl{ArrangementBoundary.f3Three_generalFourLines_avoidingPlanes_card}}\par\vspace{0.1ex}{\fontsize{8}{9}\selectfont\ttfamily Ssot/\allowbreak ArrangementBoundary.lean} \\
\end{longtable}
\fi
\makeatother
\endgroup

}{%
  \textbf{Error:} \texttt{content/lean\_handle\_ids\_auto.tex} not found.
}

\vfill
\newpage

\section{Claim Mapping to Lean Handles}

The table maps paper claims to the mechanized proof handles in the Lean 4 artifact.

\IfFileExists{content/claim_mapping_auto.tex}{%
\begingroup
\scriptsize
\setlength{\tabcolsep}{3pt}
\renewcommand{\arraystretch}{1.0}
\setlength{\LTpre}{0pt}
\setlength{\LTpost}{0pt}
\begin{longtable}{@{}>{\raggedright\arraybackslash}m{0.65\linewidth}>{\raggedleft\arraybackslash}m{0.30\linewidth}@{}}
\toprule
\textbf{Manuscript claim} & \textbf{Lean handle} \\
\midrule
\endfirsthead
\toprule
\textbf{Manuscript claim} & \textbf{Lean handle} \\
\midrule
\endhead
\endfoot
\bottomrule
\endlastfoot
Theorem 3.1: Pair-injectivity characterization & \LH{OBS1}, \LH{OBS2} \\
\midrule
Theorem 4.6: Canonical presentation & \LH{MFT202} \\
\midrule
Proposition 4.8: Block composition gives strong product & \LH{MFT25}, \LH{MFT26} \\
\midrule
Theorem 5.1: Matroidal specialization of affine coordinate views & \LH{AFM1}, \LH{AFM2}, \LH{AFM3}, \LH{AFM4}, \LH{AFM5}, \LH{AFM20}, \LH{AFM21}, \LH{AFM22} \\
\midrule
Proposition 5.2: View-fiber clique sizes & \LH{AFM6}, \LH{AFM7}, \LH{AFM8}, \LH{AFM9}, \LH{AFM10}, \LH{AFM11}, \LH{AFM12}, \LH{MFT86} \\
\midrule
Theorem 5.3: Polynomial-time matroid capacity certificates & \LH{AFM6}, \LH{AFM7}, \LH{AFM8}, \LH{AFM9}, \LH{AFM10}, \LH{AFM11}, \LH{AFM12} \\
\midrule
Theorem 5.4: Kernel-section exactness & \LH{AKS1}, \LH{AKS2}, \LH{AKS3}, \LH{AKS4}, \LH{AKS5} \\
\midrule
Corollary 5.5: No-Shannon-gap kernel-section instances & \LH{AKS32} \\
\midrule
Proposition 5.6: Support-hitting characterization & \LH{AKS20}, \LH{AKS21}, \LH{AKS36}, \LH{AKS37}, \LH{AKS38} \\
\midrule
Corollary 5.7: Arrangement-relative support profile & \LH{AKS36}, \LH{AKS37}, \LH{AKS38} \\
\midrule
Corollary 5.8: Subcode-distance boundary for all-cardinality views & \LH{AKS34} \\
\midrule
Theorem 5.9: Signed-profile compression for Grassmannian avoidance & \LH{ABD21}, \LH{ABD22}, \LH{ABD23}, \LH{ABD79}, \LH{ABD80}, \LH{ABD81}, \LH{ABD82}, \LH{ABD83}, \LH{ABD84}, \LH{ABD85}, \LH{ABD86}, \LH{ABD87}, \LH{ABD88}, \LH{ABD89}, \LH{ABD90}, \LH{ABD91}, \LH{SPB240}, \LH{SPB241}, \LH{SPB242}, \LH{SPB243}, \LH{SPB244}, \LH{SPB245}, \LH{SPB288}, \LH{SPB289}, \LH{SPB379}, \LH{SPB380}, \LH{SPB381}, \LH{SPB382}, \LH{SPB383} \\
\midrule
Lemma 5.10: Finite-union subspace avoidance & \LH{AKS6}, \LH{AKS7}, \LH{AKS17}, \LH{AKS18} \\
\midrule
Theorem 5.11: Field-size exactness & \LH{AFM24}, \LH{AKS6}, \LH{AKS7}, \LH{AKS8}, \LH{AKS9}, \LH{AKS17}, \LH{AKS18} \\
\midrule
Theorem 5.12: Rank-one linear kernel-section existence is NP-complete & \LH{KSH1}, \LH{KSH2}, \LH{KSH3}, \LH{KSH4}, \LH{KSH5}, \LH{KSH9}, \LH{KSH10}, \LH{KSH11} \\
\midrule
Corollary 5.13: Rank-one nonlinear section existence is NP-complete & \LH{KSH6}, \LH{KSH7}, \LH{KSH8}, \LH{KSH12}, \LH{KSH13} \\
\midrule
Corollary 5.14: Fixed-rank linear kernel-section existence is NP-complete & \LH{AKS35}, \LH{KSH1}, \LH{KSH2}, \LH{KSH3}, \LH{KSH4}, \LH{KSH5}, \LH{KSH9}, \LH{KSH10}, \LH{KSH11} \\
\midrule
Corollary 5.15: Grassmannian avoidance-count positivity is NP-complete & \LH{ABD10}, \LH{AKS35}, \LH{KSH1}, \LH{KSH2}, \LH{KSH3}, \LH{KSH4}, \LH{KSH5}, \LH{KSH9}, \LH{KSH10}, \LH{KSH11} \\
\midrule
Corollary 5.16: Rank-one avoidance counting is parsimoniously \(\#\textsc{P}\)-hard & \LH{KSH14}, \LH{KSH15} \\
\midrule
Lemma 5.17: Common-core lift & \LH{ABD28}, \LH{ABD29}, \LH{ABD30}, \LH{ABD31}, \LH{ABD32}, \LH{ABD33}, \LH{ABD34}, \LH{ABD35}, \LH{ABD42}, \LH{ABD48}, \LH{ABD49}, \LH{ABD50}, \LH{ABD51}, \LH{ABD52}, \LH{ABD53}, \LH{ABD54}, \LH{ABD55}, \LH{ABD56} \\
\midrule
Corollary 5.18: Graphic-matroid rank-sieve counting hardness & \LH{GCB3}, \LH{GCB4}, \LH{GCB5}, \LH{GCB6}, \LH{GCB7}, \LH{GCB8}, \LH{GCB9}, \LH{GCB10}, \LH{GCB11}, \LH{GCB12}, \LH{GCB13}, \LH{GCB14}, \LH{GCB15}, \LH{GCB16}, \LH{GCB17}, \LH{GCB18}, \LH{GCB19}, \LH{GCB20}, \LH{GCB21}, \LH{GCB22}, \LH{GCB23}, \LH{GCB24}, \LH{GCB25}, \LH{GCB26}, \LH{GCB27}, \LH{GCB28}, \LH{GCB29}, \LH{GCB30}, \LH{GCB31}, \LH{GCB32}, \LH{GCB33}, \LH{GCB34}, \LH{GCB35}, \LH{GCB36}, \LH{GCB37}, \LH{GCB38}, \LH{GCB39}, \LH{GCB40}, \LH{GCB41}, \LH{GCB42}, \LH{GCR121}, \LH{GCR122}, \LH{GCR123}, \LH{GCR124}, \LH{GCR125} \\
\midrule
Theorem 5.19: Size-indexed fixed-lattice SAT and counting spine & \LH{SCB979}, \LH{SCB980}, \LH{SCB981}, \LH{SCB982}, \LH{SCB983}, \LH{SCB984}, \LH{SCB985}, \LH{SCB986}, \LH{SCB987}, \LH{SCB988}, \LH{SCB989}, \LH{SCB990}, \LH{SCB991}, \LH{SCB992}, \LH{SCB993}, \LH{SCB994}, \LH{SCB995}, \LH{SCB996}, \LH{SCB997}, \LH{SCB998}, \LH{SCB999}, \LH{SCB1000}, \LH{SCB1001}, \LH{SCB1002}, \LH{SCB1003}, \LH{SCB1004}, \LH{SCB1005}, \LH{SCB1006}, \LH{SCB1007}, \LH{SCB1008}, \LH{SCB1009}, \LH{SCB1010}, \LH{SCB1011}, \LH{SCB1012}, \LH{SCB1013} \\
\midrule
Proposition 5.21: Hamming-endpoint family & \LH{AKS19} \\
\midrule
Proposition 5.22: All-\(k\)-view Hamming-threshold specialization & \LH{MFT102}, \LH{MFT103}, \LH{MFT104}, \LH{MFT105}, \LH{MFT106} \\
\midrule
Corollary 5.23: Haemers optimality on the kernel-section class & \LH{AKS1}, \LH{AKS2}, \LH{AKS3}, \LH{AKS4}, \LH{AKS5}, \LH{MFT95} \\
\midrule
Proposition 5.24: Subfamily certificate hierarchy & \LH{MFT92}, \LH{MFT93}, \LH{MFT94} \\
\midrule
Proposition 5.26: Matroid direct sum under block composition & \LH{AFM18}, \LH{AFM19} \\
\midrule
Corollary 5.27: Rank additivity & \LH{AFM17} \\
\midrule
Corollary 5.28: Direct-sum closure of kernel-section exactness & \LH{AFM17}, \LH{AFM18}, \LH{AFM19}, \LH{AKS1}, \LH{AKS2}, \LH{AKS3}, \LH{AKS4}, \LH{AKS5}, \LH{MFT25}, \LH{MFT26} \\
\midrule
Theorem 6.2: Cluster/non-transitive dichotomy & \LH{MFT139}, \LH{MFT140}, \LH{MFT141}, \LH{MFT142}, \LH{MFT143}, \LH{MFT144}, \LH{MFT145}, \LH{MFT146}, \LH{MFT147}, \LH{MFT148}, \LH{MFT149}, \LH{MFT150} \\
\midrule
Proposition 6.3: Affine kernel-sum transitivity criterion & \LH{AKS22}, \LH{AKS23}, \LH{AKS24}, \LH{AKS25}, \LH{AKS26}, \LH{AKS27}, \LH{AKS28}, \LH{AKS29}, \LH{AKS30}, \LH{AKS31} \\
\end{longtable}
\endgroup

}{%
  \textbf{Error:} \texttt{content/claim\_mapping\_auto.tex} not found.
}

\newpage

\section{Finite Numeric Check for the Four-Coordinate Example}

The binary four-coordinate Hamming-threshold example in the main text uses one finite graph computation:
\[
\alpha(G)=2,\qquad \alpha(G^{\boxtimes 2})=5.
\]
The artifact includes the checker
\[
\texttt{\seqsplit{proofs/checks/theta\_two\_coordinate\_example.py}}.
\]
Run it from the artifact root with
\begin{verbatim}
python3 proofs/checks/theta_two_coordinate_example.py
\end{verbatim}
The script builds the graph on $\mathbb F_2^4$, verifies the Walsh-character spectrum giving the spectral theta bound $8/3$, and computes independence in $G$ and $G^{\boxtimes 2}$ as maximum clique in the complement graph using an exact bitset branch-and-bound search. The expected output includes:
\begin{verbatim}
theta_spectral_bound=8/3
alpha_G=2
alpha_G_strong_square=5
\end{verbatim}

\newpage

\section{Finite Subfamily Check for the Hamming-Threshold Hierarchy}

The subfamily-certificate discussion in the main text uses the one-shot
independence boundary for the six two-coordinate views on $\mathbb F_2^4$.
The artifact includes the checker
\[
\texttt{\seqsplit{proofs/checks/subfamily\_threshold\_levels.py}}.
\]
Run it from the artifact root with
\begin{verbatim}
python3 proofs/checks/subfamily_threshold_levels.py
\end{verbatim}
The script enumerates every subfamily of the six two-coordinate views and
computes the independence number of the corresponding union graph by exact
bitset branch-and-bound. The expected output is:
\begin{verbatim}
subfamily_size=1: alpha=4: 6
subfamily_size=2: alpha=4: 15
subfamily_size=3: alpha=4: 20
subfamily_size=4: alpha=4: 15
subfamily_size=5: alpha=4: 6
subfamily_size=6: alpha=2: 1
\end{verbatim}

\newpage

\section{Finite Numeric Check for the Odd-Cycle Boundary}

The odd-cycle boundary discussion in the main text uses exact strong-square computations:
\[
\alpha(C_5^{\boxtimes 2})=5,\qquad
\alpha(C_7^{\boxtimes 2})=10.
\]
The artifact includes the checker
\[
\texttt{\seqsplit{proofs/checks/odd\_cycle\_boundary.py}}.
\]
Run it from the artifact root with
\begin{verbatim}
python3 proofs/checks/odd_cycle_boundary.py
\end{verbatim}
The script computes independence in the strong square as maximum clique in the complement graph using an exact bitset branch-and-bound search. The expected output includes:
\begin{verbatim}
C_5_alpha=2
C_5_alpha_strong_square=5
C_5_has_square_gap=true
C_7_alpha=3
C_7_alpha_strong_square=10
C_7_has_square_gap=true
\end{verbatim}

\newpage

\section{Finite Numeric Check for the Nonlinear Section Example}

The nonlinear section example in the main text uses a finite computation over $\mathbb F_3^6$:
\[
|I|=9,\qquad |(I-I)\setminus\{0\}|=72,
\]
with $36$ projective difference directions, a rank-minimizing kernel disjoint from the difference set, and no complete projective line inside those difference directions. The artifact includes the checker
\[
\texttt{\seqsplit{proofs/checks/nonlinear\_section\_witness.py}}.
\]
Run it from the artifact root with
\begin{verbatim}
python3 proofs/checks/nonlinear_section_witness.py
\end{verbatim}
The script verifies the finite set cardinalities, the projective-line obstruction to a linear two-dimensional section, and the reported graph invariants for the constructed example. The expected output includes:
\begin{verbatim}
witness_size=9
diff_set_size=72
projective_directions=36
rank_min_kernel_disjoint_diff_set=true
complete_projective_line_in_diff_directions=false
alpha=9
omega=81
chromatic_number=81
\end{verbatim}

\newpage

\section{Finite Scan for Fixed-Lattice Matroid Fibers}

The fixed-lattice boundary discussion uses a common-core lift in which the
lifted kernel-intersection shell is constant while the quotient represented
forbidden-point matroid varies. The artifact includes the checker
\[
\texttt{\seqsplit{proofs/checks/fixed\_lattice\_matroid\_fiber\_scan.py}}.
\]
Run it from the artifact root with
\begin{verbatim}
python3 proofs/checks/fixed_lattice_matroid_fiber_scan.py
\end{verbatim}
The script enumerates all four-point subsets of $\operatorname{PG}(2,3)$ and
then sweeps fixed sizes two through six. It lifts each quotient point set by a
one-dimensional common core, verifies that the lifted intersection shell is
fixed inside each size fiber, and records the different quotient rank-sieve
values in that shell. Lean handles \LH{ABD61}--\LH{ABD63} record the
same four-point finite cardinalities and pairwise common-core scan, while
\LH{ABD66}--\LH{ABD68} record the corresponding fixed-size Lean distributions
for sizes three, five, and six. Handles \LH{ABD69}, \LH{ABD70}, \LH{ABD71},
\LH{ABD72}, and \LH{ABD73} record the abstract counted graph-section form:
lifted witness pairs are counted by selected map counts over quotient-avoiding
candidates, with a constant-map-count multiplier as a corollary. Handles
\LH{ABD74}, \LH{ABD75}, \LH{ABD76}, and \LH{ABD77} give the corresponding
quotient-search/common-core-lift wrappers. Handles
\LH{ABD64} and \LH{ABD65} give explicit fixed-full-shell witnesses for the
three avoidance counts and the three quotient forbidden-point matroid profiles. The semantic
bridge handles \LH{SPB144}--\LH{SPB149} package those witnesses as a single
fixed-shell fiber with three target count values and three represented
forbidden-point-matroid shells; \LH{SPB153} and \LH{SPB154} expose the same
scan as target-valued fixed-shell reductions. Handles \LH{SPB155}--\LH{SPB157}
record the induced yes/no positivity flip, and \LH{SPB159} records that the
represented forbidden-point matroid shell controls the count within the scan.
The executable checker extends the audit through sizes two to six. In the
compact output, \texttt{matroid\_profile\_classes} lists the pair
\texttt{(N2, multiplicity)} for each represented forbidden-point matroid
rank-distribution class. The expected output includes:
\begin{verbatim}
fixed_lattice_matroid_fiber_scan_F3:
  quotient_points=13 four_point_subsets=715
  lifted_lattice_shells=1
  forbidden_point_matroid_rank_distributions=3
  same_lattice_shell_controls_N2=False
  same_forbidden_point_matroid_controls_N2=True
  N2_values_per_forbidden_point_matroid=[0, 2, 3]
  N2_distribution={0: 13, 2: 468, 3: 234}
  fixed_size_sweep:
    size=2: subsets=78 lattice_shells=1 matroid_rank_distributions=1
      same_lattice_controls_N2=True same_matroid_controls_N2=True
      N2_distribution={6: 78}
      matroid_profile_classes=[(6, 78)]
    size=3: subsets=286 lattice_shells=1 matroid_rank_distributions=2
      same_lattice_controls_N2=False same_matroid_controls_N2=True
      N2_distribution={3: 52, 4: 234}
      matroid_profile_classes=[(3, 52), (4, 234)]
    size=4: subsets=715 lattice_shells=1 matroid_rank_distributions=3
      same_lattice_controls_N2=False same_matroid_controls_N2=True
      N2_distribution={0: 13, 2: 468, 3: 234}
      matroid_profile_classes=[(0, 13), (2, 468), (3, 234)]
    size=5: subsets=1287 lattice_shells=1 matroid_rank_distributions=3
      same_lattice_controls_N2=False same_matroid_controls_N2=True
      N2_distribution={0: 117, 1: 702, 2: 468}
      matroid_profile_classes=[(0, 117), (1, 702), (2, 468)]
    size=6: subsets=1716 lattice_shells=1 matroid_rank_distributions=4
      same_lattice_controls_N2=False same_matroid_controls_N2=True
      N2_distribution={0: 702, 1: 936, 2: 78}
      matroid_profile_classes=[(0, 234), (0, 468), (1, 936), (2, 78)]
\end{verbatim}

The rank-two quotient-plane case admits an additional cross-field audit:
\[
\texttt{\seqsplit{proofs/checks/projective\_line\_rank2\_cross\_field\_scan.py}}.
\]
Run it from the artifact root with
\begin{verbatim}
python3 proofs/checks/projective_line_rank2_cross_field_scan.py
\end{verbatim}
For quotient dimension three, the rank-two avoidance count is the number of
projective lines disjoint from the forbidden point set. The checker verifies
the same line-profile control over $\operatorname{PG}(2,3)$ and
$\operatorname{PG}(2,5)$; over $\operatorname{PG}(2,5)$, the size-five sweep
has $169911$ subsets and normalized avoidance distribution
$\{5:186,8:11625,9:46500,10:93000,11:18600\}$. Handles
\LH{SPB179}--\LH{SPB184} give the generic compression theorem behind this
audit: for any finite incidence structure, the avoidance count is the
zero-incidence bucket of the incidence-count profile, so that profile controls
both exact avoidance values and positivity. Handles \LH{SPB185}--\LH{SPB187}
give the corresponding control-or-collision fork: any proposed semantic shell
either refines this incidence-count profile, or one semantic fiber already
contains two different incidence profiles. Handles \LH{SPB188}--\LH{SPB199}
connect the rank-sieve form to the same profile under the rank-only
containment-count hypothesis: the rank-sieve value is the zero-incidence
bucket, and value or positivity control factors through the incidence profile
unless a semantic fiber already contains distinct profiles. The last four
handles give the direct witness form: a rank-sieve value or positivity flip
inside one semantic fiber forces an incidence-profile collision in that fiber.

\newpage

\section{Composed Check for the Fixed-Shell Critical Bridge}

The fixed-shell critical-bridge discussion also uses one executable composed
audit from graph coloring to quotient rank-sieve positivity and then through
the common-core lift. The artifact includes the checker
\[
\texttt{\seqsplit{proofs/checks/fixed\_shell\_graphic\_critical\_composition.py}}.
\]
Run it from the artifact root with
\begin{verbatim}
python3 proofs/checks/fixed_shell_graphic_critical_composition.py
\end{verbatim}
Lean handles \LH{GCR126}--\LH{GCR129} give the direct counted fixed-shell
bridge that this checker audits. Handles \LH{GCR223} and \LH{GCR224} add the
quotient positive-candidate count and singleton-map lifted witness-pair count
for the same graphic critical instance; \LH{GCR225} packages the witness-pair
count as a target-valued fixed-shell reduction. Handles \LH{GCR226} and
\LH{GCR227} give the witness-level source-count version of the same
witness-pair bridge. Handles \LH{GCR228}, \LH{GCR229}, and \LH{GCR230} give
the direct witness-pair count-fiber and refinement consequences, while
\LH{GCR130}--\LH{GCR132} give the corresponding positive-candidate-count
consequences. Handle \LH{GCR133} packages the direct control-or-collision
fork for indexed coloring witness reductions. Handles \LH{GCR134}--\LH{GCR139}
give the direct finite coloring-count fiber and control-or-collision forms for
simple indexed graphic common-core lifts. Handles \LH{GCR140}--\LH{GCR148}
record the Lean-level finite five-vertex, seven-edge $\mathbb F_3$ witness:
the executable coloring enumerator, the checked counts $0$, $6$, and $12$,
loopless and simple-edge checks, and the resulting fixed-shell non-refinement
theorem. Handles \LH{GCB37}--\LH{GCB42} identify proper coloring as the
zero-violation bucket of the indexed edge-conflict profile. Handles
\LH{GCR174}--\LH{GCR187} lift that profile to source-level, physical
common-core, and active-shell control and collision statements for count value
and positivity. Handles \LH{GCR188}--\LH{GCR192} give the full profile
transfer: the singleton lifted-kernel conflict set equals the indexed graphic
conflict set, so the whole violation profile is preserved by the common-core
lift. Handles \LH{GCR193}--\LH{GCR199} define the corresponding generic
physical view-conflict profile for lifted searches and identify it with the
source indexed coloring violation profile on graphic common-core lifts.
Handles \LH{GCR200}--\LH{GCR201} bundle this equality as a target-valued
fixed-shell reduction. Handles \LH{GCR202} and \LH{GCR203} turn source-side
profile separation into active-shell non-refinement and same-shell physical
profile separation; \LH{GCR204} gives the target-value form.
Handles \LH{GCR205}, \LH{GCR206}, \LH{GCR207}, and \LH{GCR208} define and
transport the bounded finite-domain profile. Handles \LH{GCR209},
\LH{GCR210}, \LH{GCR211}, \LH{GCR212}, and \LH{GCR213} give the finite
source-family range inclusion, cardinality transfer, and non-refinement
consequences for that bounded profile target.
Handles \LH{GCR214}, \LH{GCR215}, and \LH{GCR216} prove that the bounded
profile controls the zero-bucket coloring count and inherits count
injectivity. Handles \LH{GCR217}, \LH{GCR218}, \LH{GCR219}, \LH{GCR220},
\LH{GCR221}, and \LH{GCR222} instantiate this for the finite K5 witness
family, giving three bounded physical profile values in one fixed active
shell and non-refinement of that bounded profile target.
Handles \LH{GCR231}, \LH{GCR232}, \LH{GCR233}, \LH{GCR234}, and \LH{GCR235}
instantiate the witness-pair count target on the same K5 family: the source
witness-count values are \(\{0,6,12\}\), those three values occur as accepted
lifted witness-pair counts in one fixed active shell, and the active shell
does not refine the witness-pair count.
Handles \LH{GCR156} and \LH{GCR157} record the middle count and the
exact three-value count set for that witness family; \LH{GCR158} pushes those
three count values into one fixed active-shell fiber.
Handle \LH{SPB175} records the generic postcomposition rule for target-valued
fixed-shell reductions. Handles \LH{GCR159}--\LH{GCR167} apply it to the
quotient-normalized K5 witness counts: the arithmetic normalization sends
$0,6,12$ to $0,1,2$, and those normalized target values still occur inside the
same fixed active-shell fiber; the fixed shell therefore does not refine the
normalized target.
Handles \LH{GCR149}--\LH{GCR155} mechanize the full seven-edge subgraph sweep
of $K_5$: there are $120$ such subgraphs, their $\mathbb F_3$ proper-coloring
counts have value set $\{0,6,12\}$, and the corresponding fibers have sizes
$20$, $70$, and $30$. Handles \LH{GCR168}--\LH{GCR173} record the executable
Lean normalization of the same sweep: after division by $6$, the value set is
$\{0,1,2\}$ with the same fiber sizes $20$, $70$, and $30$.
Handles \LH{SPB165}, \LH{SPB166},
\LH{SPB167}, \LH{SPB168}, \LH{SPB169}, \LH{SPB172}, \LH{SPB173}, and
\LH{SPB174} give the generic same-shape finite-fiber form used when the fixed
shell is allowed to depend on the source shape, while \LH{SPB171} gives the
corresponding source-precomposition step. Handles \LH{SPB176}--\LH{SPB178}
give the counted conjunctive fixed-shell theorem: once a supplied global
physical count is identified with the finite set of witnesses satisfying all
encoded local obligations, witness counts transport through the same fixed
semantic shell. Handles \LH{SPB200}, \LH{SPB201}, \LH{SPB202}, \LH{SPB203},
\LH{SPB204}, and \LH{SPB205} add the finite-family count-fiber consequences
for global conjunctions: source satisfying-witness counts embed into one
physical count fiber, and direct or injective count variation refutes
refinement by the fixed semantic shell.
The script compares two five-vertex, seven-edge graph instances over
$\mathbb F_3$. Both have the same graphic rank and the same common-core lifted
kernel-lattice profile. One graph is properly $3$-colorable and has positive
quotient rank-sieve count; the other contains a $K_4$ and has count zero. It
also enumerates all seven-edge subgraphs of $K_5$ and verifies that the whole
fixed-size family has one lifted shell but three quotient count values. The
expected output is:
\begin{verbatim}
fixed_shell_graphic_critical_composition_F3:
  yes_3_colorable: rank=4 target=3 proper_colorings=12 quotient_N=2
  no_3_colorable: rank=4 target=3 proper_colorings=0 quotient_N=0
  lifted_kernel_lattice_profile_equal=True
  common_core_lift_preserves_rank_sieve_positivity=True
  all_5_vertex_7_edge_graphs=120
  fixed_size_family_lifted_shells=1
  fixed_size_family_quotient_N_distribution={0: 20, 1: 70, 2: 30}
  fixed_size_family_proper_coloring_distribution={0: 20, 6: 70, 12: 30}
\end{verbatim}

\newpage

\section{Lean Audit for the Size-Indexed Fixed-Lattice SAT Spine}

Theorem~\texttt{fixed-lattice-sat-spine} in the main manuscript is supported by
the SAT-list bridge handles in the \texttt{SCB} family. Handles
\LH{SCB979}--\LH{SCB983} define and instantiate the complete fixed-shell SAT
spine: parsimonious \(\#\textsc{SAT}\) count transport, SAT decision
transport, paired decision/count transport, and the three corresponding
non-refinement statements. Handles \LH{SCB984}--\LH{SCB1003} give the
same-shape collision witnesses used to refute determination by the fixed
shell. Handles \LH{SCB1004}--\LH{SCB1011} prove that the complete spine
survives postcomposition with any decoded or coarser lattice/profile
invariant. The endpoint handles \LH{SCB1012} and \LH{SCB1013} are the full
fixed-lattice results for the external common-core lift and the fixed-active
gated common-core lift, respectively.

To verify these handles, run
\begin{verbatim}
lake build
\end{verbatim}
from the supplementary Lean project root and inspect the declarations listed
in the handle ledger under \texttt{Ssot/SATCriticalBridge.lean}.